\documentclass[12pt,a4paper]{article}
    \usepackage[cp1251]{inputenc}
    \usepackage[english]{babel}
\usepackage{algorithm}
\usepackage[noend]{algpseudocode}    

    \usepackage{amsmath,amssymb,amsthm}

\usepackage{hyperref,xcolor} 
\usepackage{tikz}
\usepackage{xypic}
\usetikzlibrary{shapes}
\usetikzlibrary{arrows.meta}
     \usepackage{graphicx}
     \usepackage{wrapfig}
\newcommand{\TL}{
\mathcal{L}
}

\newcommand{\TD}{
\mathcal{D}
}
\newcommand{\TMD}{
\mathcal{MD}
}
\newcommand{\TPC}{
\mathcal{PC}
}
\newcommand{\TS}{
\mathcal{S}
}
\newcommand{\TBA}{
\mathcal{BA}
}

\newcommand{\TC}{
\mathcal{C}
}

\newcommand{\TMPC}{
\mathcal{MPC}
}

\newcommand{\TML}{
\mathcal{ML}
}

  \def\algA{{\mathbf{A}}}

  \DeclareMathOperator\Perm{Perm}
  \DeclareMathOperator\CSP{CSP}
  \DeclareMathOperator\PCSP{PCSP}
  
   \DeclareMathOperator\Sg{Sg}

\DeclareMathOperator{\proj}{pr}
  
  \DeclareMathOperator\Image{Im}
    \DeclareMathOperator{\Congruences}{Con}

  \DeclareMathOperator\Pol{Pol}

  \DeclareMathOperator\Clo{Clo}

 \newcommand{\cover}[1]{
{{#1}^{*}}
}

  \DeclareMathOperator{\ConOne}{Con}
  \DeclareMathOperator{\Var}{Var}
  \DeclareMathOperator{\Sol}{Sol}
  \DeclareMathOperator{\TwoTuples}{TwoTuples}
  \DeclareMathOperator{\RightLinked}{RightLinked}
  \DeclareMathOperator{\LeftLinked}{LeftLinked}
  \DeclareMathOperator{\Expanded}{ExpCov}

\def\R{{\mathbb R}}

\renewcommand{\le}{\leqslant}
\renewcommand{\ge}{\geqslant}

\theoremstyle{definition}
\theoremstyle{plain}
\newtheorem{thm}{Theorem}
\newtheorem{nothm}{Informal Claim}

\newtheorem{lem}[thm]{Lemma}

\newtheorem{cor}[thm]{Corollary}

\newtheorem{remark}{Remark}

\newtheorem*{LEMMultiTypeStillStableLEM}{Lemma~\ref{LEMMultiTypeStillStable}}

\newtheorem*{LEMPreserveLinkdnessLEM}{Lemma~\ref{LEMPreserveLinkdness}}

\newtheorem*{LEMMaximalMultExtentionLEM}{Lemma~\ref{LEMMaximalMultExtention}}

\newtheorem*{LEMBridgeTOPCCongruenceLEM}{Lemma~\ref{LEMBridgeTOPCCongruence}}

\newtheorem*{LEMLInearOnTheTopIsEasyLEM}{Lemma~\ref{LEMLInearOnTheTopIsEasy}}

\newtheorem*{THMMainStableIntersectionTHM}{Theorem~\ref{THMMainStableIntersection}}

\newtheorem*{CORPropagateToRelationsCOR}{Corollary~\ref{CORPropagateToRelations}}

\newtheorem*{CORMainStableIntersectionCOR}{Corollary~\ref{CORMainStableIntersection}}

\newtheorem*{CORPropagateMultiplyByCongruenceCOR}{Corollary~\ref{CORPropagateMultiplyByCongruence}}

\newtheorem*{LEMLinearEquivalentConditionsLEM}
{Lemma~\ref{LEMLinearEquivalentConditions}}

\newtheorem*{LEMPCBridgesAreTrivialLEM}
{Lemma~\ref{LEMPCBridgesAreTrivial}}

\newtheorem*{LEMPCOnTheTopIsEasyLEM}
{Lemma~\ref{LEMPCOnTheTopIsEasy}}

\newtheorem*{THMBuildAReductionFromXYSymmetricTHM}{Theorem~\ref{THMBuildAReductionFromXYSymmetric}}

\newtheorem*{THMMainTheoremOnXYSymmetricTHM}
{Theorem~\ref{THMMainTheoremOnXYSymmetric}}

\newtheorem*{LEMNoBridgeBetweenDifferentTypesLEM}
{Lemma~\ref{LEMNoBridgeBetweenDifferentTypes}}

\newtheorem*{LEMIntersectALLLEM}
{Lemma~\ref{LEMIntersectALL}}

\newtheorem*{LEMUbiquityLEM}
{Lemma~\ref{LEMUbiquity}}

\newtheorem*{LEMPropagationLEM}
{Lemma~\ref{LEMPropagation}}

\newtheorem*{CORPropagateFromFactorCOR}
{Corollary~\ref{CORPropagateFromFactor}}

\newtheorem*{LEMMakePerfectCongruenceFromLinkedLEM}{Lemma~\ref{LEMMakePerfectCongruenceFromLinked}}

\newtheorem*{THMExistenceOfReductionTHM}{Theorem~\ref{THMExistenceOfReduction}}

\newtheorem*{THMpropagateXYSymmetricTHM}{Theorem~\ref{THMpropagateXYSymmetric}}

\begin{document}

\title{A simplified proof of the CSP Dichotomy Conjecture and XY-symmetric operations}
\author{Dmitriy Zhuk \thanks{The author is funded by the European Union (ERC, POCOCOP, 101071674). Views and opinions expressed are however those of the author(s) only and do not necessarily reflect those of the European Union or the European Research Council Executive Agency. Neither the European Union nor the granting authority can be held responsible for them.}
}

\maketitle
\small 
{\tableofcontents}
\newpage 
\begin{abstract}
We develop a new theory of strong subalgebras and linear congruences that are defined globally. 
Using this theory we provide a new proof
of the correctness of Zhuk's algorithm for all tractable CSPs on a finite domain, 
and therefore a new simplified proof of the CSP Dichotomy Conjecture.
Additionally, using the new theory we prove that 
composing a weak near-unanimity operation of an odd arity $n$ we can derive 
an $n$-ary operation that is symmetric on all two-element sets. 
Thus, CSP over a constraint language $\Gamma$ on a finite domain is 
tractable if and only if 
there exist infinitely many polymorphisms of $\Gamma$ 
that are symmetric on all two-element sets.
\end{abstract}

\section{Introduction}

The Constraint Satisfaction Problem (CSP) is the problem 
of deciding whether a set of constraints has a satisfying assignment. 
In general, the problem is NP-hard (or even undecidable for infinite domains) 
and to obtain tractable cases 
we restrict the set of admissible constraints.
Let $A$ be a finite set and $\Gamma$ be a set of relations on $A$, called 
\emph{the constraint language}.
Then $\CSP(\Gamma)$ is the problem of deciding 
whether a conjunctive formula
\begin{align}R_{1}(\dots)\wedge R_{2}(\dots)\wedge \dots\wedge R_{s}(\dots),\tag{\textasteriskcentered}\end{align}
where $R_{1},\dots,R_{s}\in \Gamma$,
is satisfiable.
It was conjectured that
$\CSP(\Gamma)$ is either in P, or NP-complete \cite{FederVardi}.
In 2017, two independent proofs of this conjecture appeared
\cite{zhuk2020proof,ZhukFVConjecture,BulatovFVConjecture,BulatovProofCSP},
and the conjecture became a theorem.
To formulate it properly, we need two definitions.

An operation $f$ on a set $A$ is called \emph{a weak near-unanimity (WNU) operation} if it satisfies 
$f(y,x,\ldots,x) = f(x,y,x,\ldots,x) = \dots = f(x,x,\ldots,x,y)$
for all $x,y\in A$.
We say that an operation
$f\colon A^{n}\to A$ \emph{preserves} a relation $R\subseteq A^{m}$
if
$$(a_{1,1},\ldots,a_{1,m}),\dots,(a_{n,1},\ldots,a_{n,m})\in R
\Rightarrow
(f(a_{1,1},\ldots,a_{n,1}),\ldots,f(a_{1,m},\ldots,a_{n,m}))\in R.$$
We say that an operation \emph{preserves a set of relations $\Gamma$} if it preserves every relation in $\Gamma$.
If $f$ preserves $R$ or $\Gamma$, we also say that 
$f$ is \emph{a polymorphism} of $R$ or 
$f$ is  \emph{a polymorphism} of $\Gamma$, 
and write $f\in\Pol(R)$ or $f\in\Pol(\Gamma)$, respectively.

\begin{thm}[\cite{zhuk2020proof,ZhukFVConjecture,BulatovFVConjecture,BulatovProofCSP}]\label{mainthm}
Suppose $\Gamma$ is a finite set of relations on 
a finite set $A$.
Then $\CSP(\Gamma)$ can be solved in polynomial time if there exists a WNU
preserving $\Gamma$;
$\CSP(\Gamma)$ is NP-complete otherwise.
\end{thm}

The NP-hardness for constraint languages without a WNU follows from 
\cite{bulatov2001algebraic,CSPconjecture}
and \cite{miklos}. 
The essential part of each proof
of the CSP Dichotomy Conjecture is 
a polynomial algorithm that works for all tractable cases,
and the tricky and cumbersome part is to show that the algorithm works correctly.

One of the two main ingredients of Zhuk's proof
is the idea of 
strong/linear subagebras that exist in every finite algebra with a WNU term operation. 
We may assume that the domain of each variable $x$ 
in (\textasteriskcentered)
is a subset (subuniverse) $D_{x}$ of $A$,
and each $D_{x}$ has a strong/linear subset.
We prove the existence of a solution (or some properties) 
of the instance by gradually reducing the domains $D_{x}$ of the variables 
to such strong subsets until all the domains are singletons.

The crucial disadvantage 
of this approach is 
that the linear subalgebras we obtain only exist  locally, whereas the properties we want to prove are global.
For example, we could start with a domain $D_{x}=\{0,1,\dots,99\}$ for some variable $x$.
We gradually reduce this domain
$D_{x}\supsetneq D_{x}^{(1)}\supsetneq D_{x}^{(2)}
\supsetneq\dots \supsetneq D_{x}^{(10)}=\{0,1\}$,
and on $\{0,1\}$ our instance is just a system of linear equations 
modulo 2. Nevertheless, the strong properties  
we have on $\{0,1\}$ do not say much about 
the behaviour on the whole domain $\{0,1,\dots,99\}$.
As a result, we are forced to go forward and backward from
global $\{0,1,\dots,99\}$ 
to local $\{0,1\}$, and use a very complicated induction 
to prove most of the claims.

In this paper, we develop a new theory such that 
every reduction is either strong or global. 
Precisely, 
for every domain 
$D_{x}$ we can build 
a sequence 
$D_{x}\supsetneq D_{x}^{(1)}\supsetneq D_{x}^{(2)}
\supsetneq\dots \supsetneq D_{x}^{(s)}=\{a\}$ 
such that for every $i\in\{1,\dots,s\}$ 
\begin{enumerate}
    \item either $D_{x}^{(i+1)}$ is a strong subset 
    of $D_{x}^{(i)}$, 
    \item or there exists an equivalence relation $\sigma$ on $D_{x}$ satisfying very strong properties such that 
    $D_{x}^{(i+1)}$ is an intersection of $D_{x}^{(i)}$ with a block of 
    this equivalence relation. 
\end{enumerate}
In the above example with 
$D_{x}^{(10)}=\{0,1\}$ we would have an equivalence relation
on $\{0,1,\dots,99\}$ such that $0$ and $1$ are in different blocks of this
equivalence relation, and the linear behaviour on $\{0,1\}$ is due to the properties of the equivalence relation.  

Another good feature of the new approach is that 
whenever we have such a sequence of ``good'' subsets
$D_{s}\subsetneq D_{s-1}\subsetneq\dots\subsetneq D_{1}$, 
we do not care about the types of the subsets in the middle.
We only need to know that such a sequence exists, which we denote by
$D_{s}\lll D_{1}$.

Finally, the new theory connected the two main ideas of Zhuk's proof: the idea of strong subalgebras and the idea of bridges and connectedness. 
Originally, they lived separate lifes.
Using strong subalgebras and reductions we showed that all the relations have the parallelogram property, which gives us an 
irreducible congruence for every constraint and its variable.
Then manipulating with the instance we tried to connect 
the congruences (variables) with bridges.
In the new theory, bridges appear naturally from strong/linear subuniverses: whenever a restriction to strong/linear subuniverses gives an empty set, it immediately gives a bridge 
between congruences (equivalence relations) defining these subuniverses.

Using the new theory we obtain two results presented in the next two subsections.

\subsection{A simplified proof of the CSP Dichotomy Conjecture} 
First, we provide a new proof of the correctness of Zhuk's algorithm. 
Three main statements that imply the correctness are formulated in Section ''Correctness of the algorithm`` 
in \cite{zhuk2020proof}.
Below we formulate informal analogues of these statements, 
and the formal statements can be found in Section \ref{SectionCSPMainClaims}.


\begin{nothm}\label{ICExistenceStrong}
Suppose $\Gamma$ is a constraint language  preserved by 
a WNU operation $w$.
Then each $D_{x}$ of size at least 2 has a strong subset (subuniverse) 
or an equivalence relation $\sigma$ such that 
$D_{x}/\sigma\cong \mathbb Z_{p}$
for some prime $p$.
\end{nothm}

\begin{nothm}\label{ICReductionTOStrong}
Suppose 
\begin{enumerate}
    \item $\Gamma$ is a constraint language  preserved by 
a WNU operation $w$;

\item $\mathcal I$ is a consistent enough (cycle-consistent + irreducible) instance of $\CSP(\Gamma)$;
\item $\mathcal I$ has a solution;
\item $B$ is a strong subset of $D_{x}$, where $x$ is a variable of $\mathcal I$.
\end{enumerate}
Then $\mathcal I$ has a solution with $x\in B$.
\end{nothm}

\begin{nothm}\label{ICCodimensionOne}
Suppose 
\begin{enumerate}
    \item $\Gamma$ is a constraint language  preserved by 
a WNU operation $w$;
\item $\mathcal I$ is a consistent enough (cycle-consistent + irreducible + another one) instance of $\CSP(\Gamma)$ with variables 
$x_1,\dots,x_n$;
\item $D_{x_{i}}$ has no strong subsets 
for any $i$;
\item $\mathcal I$ is linked, i.e. the following graph is connected: 
the vertices are all pairs $(x_i,a)$, where $a\in D_{x_{i}}$, 
and vertices $(x_i,a)$ and $(x_j,b)$ are adjacent whenever 
there is a constraint in $\mathcal I$ whose projection onto $x_i,x_j$ contains $(a,b)$;
\item $\sigma_{x_i}$ is the minimal equivalence relation on each 
$D_{x_i}$ such that $D_{x_i}/\sigma_{x_i}\cong \mathbb Z_{q_1}\times\dots\times \mathbb Z_{q_{n_i}}$;
\item $\varphi:\mathbb Z_{p_1}\times \dots\times \mathbb Z_{p_{m}}\to D_{x_{1}}/\sigma_{x_1}\times 
\dots\times D_{x_{n}}/\sigma_{x_n}$ is a linear map;
\item if we remove any constraint from $\mathcal I$ 
then the obtained instance has a solution inside $\varphi(\alpha)$ for every $\alpha\in \mathbb Z_{p_1}\times \dots\times \mathbb Z_{p_{m}}$. 
\end{enumerate}
Then 
$\{(a_1,\dots,a_{m})\mid \mathcal I \text{ has a solution in }\varphi(a_1,\dots,a_m)\}$ is either empty, or full, 
or an affine subspace of $\mathbb Z_{p_1}\times \dots\times \mathbb Z_{p_{m}}$ of dimension $m-1$.
\end{nothm}

\newcommand{\Output}{\mbox{Output}}
\newcommand{\Changed}{\mbox{Changed}}
\newcommand{\calC}{\mathbf{C}}
\newcommand{\calD}{\mathbf{D}}
\newcommand{\X}{\mathbf{X}}

\newcommand{\Solve}{\mbox{\textsc{Solve}}}
\newcommand{\ForceConsistency}{\mbox{\textsc{ForceConsistency}}}
\newcommand{\SolveLinear}{\mbox{\textsc{SolveLinear}}}
\newcommand{\ReduceDomain}{\mbox{\textsc{ReduceDomain}}}
\newcommand{\FindClass}{\mbox{\textsc{FindClass}}}
\newcommand{\CheckWeakerInstance}{\mbox{\textsc{CheckWeakerInstance}}}
\newcommand{\SolveLinearCase}{\mbox{\textsc{SolveLinearCase}}}
\newcommand{\RemoveTriv}{\mbox{\textsc{RemoveTrivialities}}}
\newcommand{\FactorizeInstance}{\mbox{\textsc{FactorizeInstance}}}
\newcommand{\Reduce}{\mbox{\textsc{Reduce}}}
\newcommand{\SolveLinearSystem}{\mbox{\textsc{SolveLinearSystem}}}
\newcommand{\WeakenConstraint}{\mbox{\textsc{WeakenConstraint}}}
\newcommand{\CheckAllTuples}{\mbox{\textsc{CheckAllTuples}}}
\newcommand{\SolveNonlinked}{\mbox{\textsc{SolveNonlinked}}}
\newcommand{\FindOneEquationLinked}{\mbox{\textsc{FindOneEquationLinked}}}
\newcommand{\FindOneEquationNonlinked}{\mbox{\textsc{FindOneEquationNonlinked}}}
\newcommand{\FindEquationsForNonlinked}{\mbox{\textsc{FindEquationsNonlinked}}}
\newcommand{\WeakenEveryConstraint}{\mbox{\textsc{WeakenEveryConstraint}}}

Let us explain how Zhuk's algorithm works (for the precise algorithm 
see \cite{zhuk2020proof,ZhukFVConjecture}).
The main function $\Solve$ takes a CSP instance $\mathcal I$
with variables $x_1,\dots,x_n$ as an input
(see the pseudocode).
First, it forces a sufficient level of consistency by 
function $\ForceConsistency$.
If we cannot achieve this, then the instance has no solutions, and we answer ``No''.
Then, if there exists a strong subset $B$ of the domain $D_{x_{i}}$
of some variable $x_{i}$, it reduces the domain
of $x_{i}$ to $B$ by $\ReduceDomain$, and forces the consistency again.
This procedure is justified by Informal Claim \ref{ICReductionTOStrong},
which guarantees that we cannot lose all the solutions when reduce to 
a strong subset.
If there are no strong subsets, then 
Informal Claim \ref{ICExistenceStrong}
implies that 
for every domain $D_{x_{i}}$ of size at least 2 there exists 
an equivalence relation $\sigma_{x_{i}}$ such that $D_{x_{i}}/\sigma_{x_{i}}\cong 
\mathbb Z_{q_1}\times\dots\times \mathbb Z_{q_{n_{i}}}$ for some $n_{i}\ge 1$.
Choose $\sigma_{x_{i}}$ to be minimal and therefore $n_{i}$ to be maximal for every $i$.
This case is solved by a separate function $\SolveLinear$.

\begin{algorithm}
\begin{algorithmic}[1]
\Function{\Solve}{$\mathcal I$}
\Repeat
    \State{$\mathcal I: = \ForceConsistency(\mathcal I)$}
    \If{$\mathcal I = false$} \Return{``No"}
    \EndIf
    \If{$D_{x_i}$ has a strong subset $B$} 
        \State{$\mathcal I:= \ReduceDomain(\mathcal I,x_i,B)$}
    \EndIf
\Until{nothing changed}
\State \Return{$\SolveLinear(\mathcal I)$}
\EndFunction
\end{algorithmic}
\end{algorithm}

Let $\varphi:\mathbb Z_{p_1}\times \dots\times \mathbb Z_{p_m}\to 
D_{x_1}/\sigma_{x_1}\times\dots\times D_{x_n}/\sigma_{x_n}$ be a 
linear map.
By $\varphi^{-1}(\mathcal I)$ we denote 
$\{\alpha\in \mathbb Z_{p_1}\times \dots\times \mathbb Z_{p_{m}}\mid
\mathcal I \text{ has a solution in }\varphi(\alpha)\}.$ 
Calculating $\varphi^{-1}(\mathcal I)$ would solve the instance
$\mathcal I$
because $\mathcal I$ has a solution if and only if 
$\varphi^{-1}(\mathcal I)$ is not empty.
We do not know how to calculate $\varphi^{-1}(\mathcal I)$ but we can do the following calculations:
\begin{enumerate}
    \item[(p0)] For a concrete $\alpha\in \mathbb Z_{p_1}\times \dots\times \mathbb Z_{p_{m}}$ check whether 
    $\varphi^{-1}(\mathcal I)$ contains $\alpha$:
    \begin{enumerate}
        \item reduce each domain $D_{x_{i}}$ to the $i$-th element of 
        $\varphi(\alpha)$, which is a block of 
        $\sigma_{x_{i}}$, and solve CSP on a smaller domain by recursion.
    \end{enumerate}
    \item[(p1)] Check whether $\varphi^{-1}(\mathcal I)=\mathbb Z_{p_1}\times \dots\times \mathbb Z_{p_{m}}$:
    \begin{enumerate}
        \item using (p0) check that $(0,\dots,0)\in \varphi^{-1}(\mathcal I)$;
        \item using (p0) check that $(\underbrace{0,\dots,0,1}_{i},0,\dots,0)\in \varphi^{-1}(\mathcal I)$ for every $i$.        
    \end{enumerate}
    \item[(p2)] Calculate $\varphi^{-1}(\mathcal I)$ if $\mathcal I$ is not linked (see condition 4 in Informal Claim \ref{ICCodimensionOne}):
        \begin{enumerate}
        \item split $\mathcal I$ into linked instances $\mathcal I_1,\dots,\mathcal I_r$ on smaller domains;
        \item using recursion calculate $\varphi^{-1}(\mathcal I_{i})$ for every $i$;
        \item $\varphi^{-1}(\mathcal I) = \varphi^{-1}(\mathcal I_1)\cup 
        \dots\cup\varphi^{-1}(\mathcal I_r)$.
    \end{enumerate}
    \item[(p3)] Calculate $\varphi^{-1}(\mathcal I)$ if 
    the dimension of $\varphi^{-1}(\mathcal I)$ is $m-1$ or $\varphi^{-1}(\mathcal I)$ is empty.
        \begin{enumerate}
        \item using (p1) find $(a_1,\dots,a_m)\notin \varphi^{-1}(\mathcal I)$;
        \item for every $i$ using (p0) find $b_{i}$ such that 
        $(a_1,\dots,a_{i-1},b_{i},a_{i+1},\dots,a_m)\in \varphi^{-1}(\mathcal I)$ if it exists;
        \item let $J$ be the set of all $i\in\{1,2,\dots,m\}$ such that $b_{i}$ exists;
        \item the equation defining $\varphi^{-1}(\mathcal I)$ is 
        $\sum_{i\in J}(y_{i}-a_{i})/(b_{i}-a_{i}) = 1$.
        \footnote{The fact that different variables $y_{i}$ take on values from different 
        fields $\mathbb Z_{p_{i}}$ is not a problem as $J$ may contain 
        only variables on the same field.}
    \end{enumerate}
\end{enumerate}

\begin{algorithm}
\begin{algorithmic}[1]
\Function{\SolveLinear}{$\mathcal I$}
\State{$m:=$ the dimension of 
$D_{x_1}/\sigma_{x_1}\times\dots\times D_{x_n}/\sigma_{x_n}$}
\Comment{$m\ge 0$}
\State{$\varphi:=$ a bijective linear map 
$\mathbb Z_{p_1}\times \dots\times \mathbb Z_{p_{m}}\to D_{x_1}/\sigma_{x_1}\times\dots\times D_{x_n}/\sigma_{x_n}$}
\While{$\varphi^{-1}(\mathcal I)\neq\mathbb Z_{p_1}\times \dots\times \mathbb Z_{p_{m}}$}    \Comment{using (p1)}
\State{$\mathcal I':=\mathcal I$} 
        \For{$C\in\mathcal I'$} \Comment{remove unnecessary constraint from $\mathcal I'$} 
            \If{$\varphi^{-1}(\mathcal I'\setminus \{C\})\neq  \mathbb Z_{p_1}\times \dots\times \mathbb Z_{p_{m}}$}
             \Comment{using (p1)} 
                \State{$\mathcal I':=\mathcal I'\setminus \{C\}$}
            \EndIf
        \EndFor
    \State{$F := \varphi^{-1}(\mathcal I')$} \Comment{using (p2) or (p3)} 
    \If{$F=\varnothing$}  \Return{``No"}
    \EndIf
    \State{$m:=$ the dimension of $F$}
    \State{$\psi:=$ a bijective linear map $\mathbb Z_{p_1}\times \dots\times \mathbb Z_{p_{m}}\to F$}\Comment{$p_1,\dots,p_{m}$ are also updated}
    \State{$\varphi:=\varphi\circ\psi$}    
\EndWhile
\State \Return{``Yes"}
\EndFunction
\end{algorithmic}
\end{algorithm}

The function $\SolveLinear$ solving the remaining case works as follows
(see the pseudocode).
We start with a bijective linear map 
$\varphi:\mathbb Z_{p_1}\times \dots\times \mathbb Z_{p_{m}}\to 
D_{x_1}/\sigma_{x_1}\times\dots\times D_{x_n}/\sigma_{x_n}$.
We gradually reduce the dimension $m$ maintaining the property that $\Image{\varphi}$
contains all the solutions of $\mathcal I$.
We stop when 
$\varphi^{-1}(\mathcal I) = \mathbb Z_{p_1}\times \dots\times \mathbb Z_{p_{m}}$ or 
$\varphi^{-1}(\mathcal I)$ is empty.
First, we make a copy $\mathcal I'$ of the instance $\mathcal I$ and remove 
all the constraints from $\mathcal I'$
that can be removed so that $\mathcal I'$ 
maintains the property 
$\varphi^{-1}(\mathcal I') \neq \mathbb Z_{p_1}\times \dots\times \mathbb Z_{p_{m}}$.
This property can be checked in polynomial time 
using (p1).
If we cannot remove any other constraint, 
by Informal Claim \ref{ICCodimensionOne} we have one of the following cases:
either $\mathcal I'$ is not linked  and we can calculate 
$\varphi^{-1}(\mathcal I')$ using (p2); 
or $\mathcal I'$ satisfies all the conditions of 
Informal Claim \ref{ICCodimensionOne}, and 
$\varphi^{-1}(\mathcal I')$ has dimension $m-1$ or is empty.
In the second case we can calculate $\varphi^{-1}(\mathcal I')$ using (p3).
Since $\mathcal I'$ was obtained from $\mathcal I$
by removing some constraints, 
we have $\varphi^{-1}(\mathcal I)\subseteq \varphi^{-1}(\mathcal I')$.
Thus, we found a smaller affine subspace $\varphi^{-1}(\mathcal I')$ that still covers all the solutions of $\mathcal I$.
It remains to replace $m$ with the 
dimension of $\varphi^{-1}(\mathcal I')$
and update the linear map $\varphi$.
Since we cannot reduce the dimension $m$
forever, we will eventually stop 
in one of the two cases:
$\varphi^{-1}(\mathcal I) = \mathbb Z_{p_1}\times \dots\times \mathbb Z_{p_{m}}$ or 
$\varphi^{-1}(\mathcal I)$ is empty.
First of them implies the existence of a solution for $\mathcal I$, the second implies that no solutions exist.

\subsection{Existence of XY-symmetric operations}

The second main result of the paper is a proof of the fact that the existence 
of a WNU term operation (polymorphism) implies the existence of a much stronger term operation 
(polymorphism).

An $n$-ary operation $f$ is called
\emph{symmetric on a tuple of variables 
$(x_{i_1},\dots,x_{i_n})$}
if it satisfies the identity 
$f(x_{i_1},\dots,x_{i_n})=f(x_{i_{\sigma(1)}},\dots,x_{i_{\sigma(n)}})$
for every permutation $\sigma$ on $\{1,2,\dots,n\}$.
For instance, 
an operation $f$ is symmetric on $(x,\dots,x,y)$ if and only if 
$f$ is a WNU operation.
An operation is called \emph{XY-symmetric} if
it is symmetric on 
$(\underbrace{x,\dots,x}_{i},y,\dots,y)$ for any $i$.
An operation $f$ is called \emph{idempotent} if
$f(x,x,\dots,x) = x$.

As it follows from the definition, 
an XY-symmetric operation satisfies much more identities than a WNU operation.
Nevertheless, we managed to prove 
that an XY-symmetric operation can always be derived from a WNU operation.
To formulate the precise statement we will need a definition of a clone. 
A set of operations is called \emph{a clone} if it is closed under composition and contains all projections.
For a set of operations $M$ by $\Clo(M)$ we denote the clone generated by $M$.

\begin{thm}\label{THMMainTheoremOnXYSymmetric}
Suppose $f$ is a WNU operation of an odd arity $n$ on a finite set.
Then there exists an XY-symmetric operation $f'\in\Clo(\{f\})$ of arity $n$.
\end{thm}

Theorem \ref{THMMainTheoremOnXYSymmetric} 
extends known characterization of finite Taylor algebras. 

\begin{cor}\label{corTaylorEquivalentConditions}
Suppose $\mathbf A$ is a finite idempotent algebra. Then the following conditions are equivalent:
\begin{enumerate}
    \item $\mathbf A$ is a Taylor algebra (satisfies nontrivial identities);
    \item there does not exist an algebra $\mathbf B\in \mathrm{HS}(\mathbf A)$ of size 2
    whose operations are projections  \cite{bulatov2001algebraic};
    \item $\mathbf A$ has a WNU term operation of any prime arity $p>|A|$ \cite{MarotiMcKenzie};
    \item $\mathbf A$ has a cyclic term operation of any prime arity $p>|A|$, 
    i.e. an operation $c_{p}$ satisfying
    $$c_{p}(x_1,x_2,\dots,x_p) = c_{p}(x_2,x_3,\dots,x_p,x_1) \cite{barto2012absorbing};$$
    \item $\mathbf A$ has a Siggers term operation, i.e an operation $f$
satisfying 
$$f(y,x,y,z) = f(x,y,z,x) \cite{kearnes2014optimal,siggers2010strong};$$
    \item $\mathbf A$ has an XY-symmetric term operation of any prime arity $p>|A|$.
\end{enumerate}
\end{cor}

Composing a cyclic operation $c_{n}$ of arity $n$ and an XY-symmetric operation $f$
of arity $n$ we can get an operation which is simultaneously 
cyclic and XY-symmetric:

$$
f'(x_1,\dots,x_n) := 
c_{n}(f(x_1,\dots,x_n), 
f(x_2,\dots,x_n,x_1),\dots,
f(x_n,x_1,\dots,x_n)).
$$
Hence, conditions 4 and 6 of Corollary \ref{corTaylorEquivalentConditions} give an infinite sequence of cyclic XY-symmetric operations, which 
are the most symmetric operations known to be in every finite Taylor algebra.

This result cannot be generalized to XYZ-symmetric operations
as witnessed by the following operation on $\{0,1,2\}$:%
$$f(x_1,x_2,x_3) = \begin{cases}
x_1+x_2+x_3 \;(\mathrm{mod} \;\;2),& \text{if $x_{1},x_2,x_3\in\{0,1\}$}\\
2,& \text{if $x_{1}=x_2=x_3=2$}\\
\text{first element different from 2 in $x_1,x_2,x_3$}, & \text{otherwise}
\end{cases}
$$
Note that the clone generated by 
$f$ is a minimal Taylor clone and its operations 
were completely described in \cite{jankovec2023minimalni}.
The following lemma shows that even if we try to generalize 
XY-symmetric operations to some tuples 
with $x$, $y$, and $z$ we fail.

\begin{lem}
$\Clo(\{f\})$ has a WNU operation of any odd arity
but 
$\Clo(\{f\})$ has no operation 
that is symmetric on 
$(\underbrace{x,\dots,x}_{k},\underbrace{y,\dots,y}_{\ell},
\underbrace{z,\dots,z}_{j})$ for some $k,\ell, j\ge 1$.
\end{lem}

\begin{proof}
The operation $f$ is conservative (always returns one of the coordinates), and behaves as a linear sum modulo 2 on $\{0,1\}$, as $\mathrm{min}$ on $\{0,2\}$ and $\{1,2\}$.
Let us show how to derive a WNU (and even XY-symmetric) operation
of any odd arity. 
Put 
$$f_{3} := f,\;\;\;\;\;\;\;\;\;
f_{2n+1}(x_1,\dots,x_{2n+1}) := f_{2n-1}(f(x_1,x_2,x_3),x_4,\dots,x_{2n+1}).$$
It follows immediately from the definition that 
$f_{2n+1}$ is symmetric on $\{0,1\}$, on $\{0,2\}$,  and on $\{1,2\}$.

Also, $f$ has the following properties.
Whenever we substitute $0$ or $1$ in $f$
we obtain either $0$, or $1$. 
Whenever we substitute 2 into some arguments, we get an operation whose restriction to $\{0,1\}$ is a linear operation (in fact a projection).
Finally, the operation preserves the sets $\{0,2\}$ and $\{1,2\}$.
These three properties imply that 
if we put 2 to some arguments of a term operation 
$g\in\Clo(\{f\})$ and restrict the obtained operation to $\{0,1\}$ 
we get an idempotent linear operation on $\{0,1\}$.

Assume that $g$ is symmetric on 
$(\underbrace{x,\dots,x}_{k},\underbrace{y,\dots,y}_{\ell},
\underbrace{z,\dots,z}_{j})$ for some $k,\ell, j\ge 1$.
Without loss of generality assume that 
$j$ is odd
and the first variable of $g$ is not dummy. 
Substituting 2 for the last $j$ coordinates of $g$
and restricting the obtained operation to $\{0,1\}$
we must get an idempotent linear operation $h$ of 
an even arity $k+\ell$.
Since $h$ must return 
the same value on all the tuples with $k$ 0s and $\ell$ 1s, all the variables of $h$ are not dummy. Since $h$ has even number of arguments, 
it cannot be idempotent, which gives a contradiction and completes the proof.
\end{proof}






Notice that both Zhuk's and Bulatov's algorithms for the CSP are not universal in the sense that 
the algorithms work only if the domain is fixed and, therefore, all the algebraic properties are known. It would be great to find a universal algorithm 
for all tractable CSPs.
Recently the importance of symmetric operations was rediscovered 
while studying the limits of universal algorithms for the CSP and its variation, called the Promise CSP \cite{ciardo2023clap,barto2021algebraic,brakensiek2020power}.
For instance, the algorithm known as BLP+AIP solves
$\CSP(\Gamma)$ if and only if 
$\Gamma$ has infinitely many symmetric polymorphisms \cite{brakensiek2020power}.

We believe that Theorem \ref{THMMainTheoremOnXYSymmetric} can be further generalized, 
and finally we will get enough symmetric operations to 
make some universal algorithm work. 
Theorem \ref{THMMainTheoremOnXYSymmetric} already gives us some implications that can be viewed as a tiny step in this direction: 
\begin{itemize}
    \item $\CSP(\Gamma)$ is solvable by BLP+AIP 
    for any multi-sorted language $\Gamma$ on a two-element domain.
    \item if $\PCSP(\mathbb A,\mathbb B)$ is solvable by reducing to 
    a tractable $\CSP(\mathbb C)$, where $\mathbb A\to\mathbb C\to\mathbb B$, $\mathbb C$ is finite, and
    $|A|=2$, then $\PCSP(\mathbb A,\mathbb B)$ is solvable by BLP+AIP
    (see \cite{krokhin2022invitation,barto2021algebraic} for more information about Promise CSP). 
\end{itemize}

To show the second claim we apply Theorem \ref{THMMainTheoremOnXYSymmetric} to 
    $\Pol(\mathbb C)$ and obtain infinitely many XY-symmetric operations on $C$. 
    Composing them with the homomorphisms 
    $\mathbb A\to\mathbb C$ and $\mathbb C\to\mathbb B$
    we obtain infinitely many symmetric polymorphisms 
    $\mathbb A\to \mathbb B$, which implies that BLP+AIP solves 
    $\PCSP(\mathbb A,\mathbb B)$ \cite{brakensiek2020power}.
    
One of the reasons why these two independent results (proof of the CSP Dichotomy Conjecture and the existence of an XY-symmetric operation) appeared in one paper 
is that their proofs have the same flavour.
Even though, the second result has a purely algebraic formulation, 
it is strongly connected to the CSP.
Let us consider the matrix whose rows are all the tuples of length $n$ having exactly two different elements.
We apply a WNU operation to columns of this matrix coordinate-wise deriving new columns till we cannot derive anything new. 
The set of all the derived columns can be viewed as a relation $R$ of some big arity $N$.  
To prove that an XY-symmetric operation can be derived from a WNU we need to show 
that $R$ contains a tuple whose elements corresponding to permutations of the same tuple are equal. This can be written as a CSP instance 
with the constraint $R(x_1,\dots,x_N)$ and many equality constraints
$(x_i=x_j)$, and we need to prove that it has a solution.
Then the proofs of Informal Claims \ref{ICReductionTOStrong} and \ref{ICCodimensionOne} are similar to the proof of Theorem \ref{THMMainTheoremOnXYSymmetric}, only sufficient level of 
consistency is replaced by symmetries of the relation $R$.

\subsection{History and acknowledgements}

The first symmetric operations (WNU) that exist in every Taylor algebra 
appeared in \cite{MarotiMcKenzie}, and the idea was to show that 
every symmetric invariant relation has a constant tuple.
In \cite{zhuk2021strong} I showed the existence of a constant tuple in a symmetric relation
gradually reducing the domain to 
strong subalgebras and keeping the property that the relation is symmetric.
It turned out that the only reason why a constant tuple does not exist in 
a symmetric relation is a linear essence inside, 
for instance 
the relation 
$x_1+\dots+x_p = 1$ does not have a constant tuple in $\mathbb Z_{p}$. 
Thus, the existence of a WNU term operation of an arity $n$ 
is reduced to a pure linear algebra question:
does every affine symmetric subspace of $\mathbb Z_{p}^{n}$
have a constant tuple. 
Similarly, we could try to 
show the existence of a 2-WNU operation  (symmetric on $(x,x,y,y,\dots,y)$).
For the proof to work we need to show that 
every relation of arity ${n\choose 2}$ with symmetries coming from $n$-permutations
has a constant tuple.
This question is again 
reduced to a pure linear algebra question:
does every affine (weak) symmetric subspace of $\mathbb Z_{p}^{n\choose 2}$
have a constant tuple.
We worked on it 
together with Libor Barto, Michael Pinsker, and their students
Johanna Brunar
and Martin Boro\v s.
Martin Boro\v s in his master thesis \cite{borovs2023symmetric} 
proved that a constant tuple always exists if and only if 
$n\cdot {n\choose 2}$ is co-prime with $p$.
Unfortunately this beautiful approach did 
not lead to a proof of existence of 2-WNU or XY-symmetric operations
as it turned out that an XY-symmetric operation may exist even if the condition on the arity is not satisfied. 
Nevertheless,
I am very thankful to Libor Barto and Michael Pinsker
for the exciting play with a beautiful linear algebra and 
the ideas I took from this play.

I would also like to thank Stanislav \v Zivn\'y, Lorenzo Ciardo, and 
Tamio-Vesa Nakajima
for very fruitful discussions about 
algorithms for the CSP based on linear programming and their limits. 
My understanding of what operations we need 
for the algorithms to work came to me during my visit of Oxford University.
The impotence of symmetric operations 
for these algorithms 
motivated me to finish 
my research on XY-symmetric operations.

\subsection{Structure of the paper}

In Section \ref{SectionStrongSubalgebras} 
we give definitions and statements of 
the new theory of strong subalgebras.
In Section \ref{SECTIONCSPDICHOTOMYPROOF} 
we use this theory to prove that the algorithm
for the CSP
presented in \cite{zhuk2020proof,ZhukFVConjecture}
works. 
In Section \ref{SECTIONXYSYMMETRIC} 
we show that an XY-symmetric operation can be derived from 
a WNU operation of an odd arity.
In Section \ref{SectionProofStrongSubalebras}
we prove all the statements formulated in 
Section \ref{SectionStrongSubalgebras}.

The main goal of the paper is to show the power of the new theory of 
strong subalgebras but not to provide a shortest proof of 
the CSP Dichotomy Conjecture.
That is why, we formulate and prove all the properties of strong/linear subalgebras for 
arbitrary finite idempotent algebras, even though 
in Sections \ref{SECTIONCSPDICHOTOMYPROOF} and \ref{SECTIONXYSYMMETRIC} we only consider algebras with a special WNU operation. 
Moreover, many definitions and statements 
could be simplified if we consider only Taylor minimal algebras (see \cite{minimaltayloralgebras}), which would be sufficient to prove two main results 
of this paper.
Also, for better readability we always duplicate statements 
if the proof appears in a later section.
For instance, all the statements 
from Section \ref{SectionStrongSubalgebras} 
are formulated again in Section \ref{SectionProofStrongSubalebras}.

\section{Strong/Linear subuniverses}\label{SectionStrongSubalgebras}

In this section we define 
six types of subuniverses and formulate all the necessary properties of these subuniverses.
We start with auxiliary definitions and notations, 
then we define two types of irreducible congruences, and
introduce notations for all types of subuniverses.
In Subsection \ref{SUBSECTIONSTRONGSUBALGEBRASPROPERTIES} 
we give their properties without a proof.
We conclude the section with a few auxiliary statements.

\subsection{Auxiliary definitions}

For a positive  integer $k$ by $[k]$ we denote the set $\{1,2,\dots,k\}$.
An idempotent WNU $w$ is called \emph{special} if 
$$w(x,\dots,x,y)=w(x,\dots,x,w(x,\dots,x,y)).$$
It is not hard to show that for any idempotent WNU $w$ on a finite set there exists a special WNU $w'\in\Clo(w)$
(see Lemma \ref{LEMExistenceOfSpesialWNULemma}).

\textbf{Algebras.}
We denote algebras by bold letters $\mathbf{A},\mathbf{B}, \mathbf C,\dots$, 
their domains by $A,B,C,\dots$, 
and the basic operations
by $f^{\mathbf A},f^{\mathbf B}, g^{\mathbf C},\dots$.
We use standard universal algebraic notions of term operation, subalgebra,  factor algebra, product of algebras,
see~\cite{bergman2011universal}.
We write $\mathbf B\le \mathbf A$ if $\mathbf B$ is a subalgebra of $\mathbf A$.
A congruence is called \emph{nontrivial} if it is 
not the equality relation and not $A^{2}$.
By $0_{\mathbf A}$ we denote the equality relation on $A$, which is 
the 0-congruence on $\mathbf A$.
To avoid overusing of bold symbols 
sometimes we write capital symbol meaning the algebra. 
An algebra $(A;F_{A})$ is called \emph{polynomially complete (PC)}
if the clone generated by $F_{A}$ and all constants on $A$ is the clone of all operations on $A$
(see \cite{istinger1979characterization,lausch2000algebra}).

By $\mathcal V_{n}$ we denote the class 
of finite algebras $\mathbf A = (A;w^{\mathbf A})$ 
whose basic operation 
$w^{\mathbf A}$ is an idempotent special WNU operation.
Since we only consider finite algebras, $\mathcal V_{n}$ is not a variety.
For a prime $p$ by $\mathbf Z_{p}$ we denote the algebra whose 
domain is $\{0,1,\dots,p-1\}$ and whose 
basic operation
$w^{\mathbf Z_{p}}$ is $x_1+\dots+x_{n} (mod \;\;p)$.
In the paper every algebra 
$\mathbf Z_{p}$ belongs to $\mathcal V_{n}$ for a fixed $n$, 
hence the algebra $\mathbf Z_{p}$ is uniquely defined.
In this paper we assume 
that every algebra is a finite idempotent algebra having a 
WNU term operation.
Moreover, in Sections \ref{SECTIONCSPDICHOTOMYPROOF}
and \ref{SECTIONXYSYMMETRIC} we usually consider algebras
from $\mathcal V_{n}$.


\textbf{Notations.}
A relation $R\subseteq A_{1}\times\dots\times A_{n}$ is called \emph{subdirect} if
for every $i$ the projection of $R$ onto the $i$-th coordinate is $A_{i}$.
A relation $R\subseteq A^{n}$ is called \emph{reflexive} if it contains 
$(a,a,\dots,a)$ for every $a\in A$.
For a relation $R$ by $\proj_{i_1,\ldots,i_{s}}(R)$
we denote the projection of $R$ onto the coordinates
$i_1,\ldots,i_{s}$. 
We write 
$\mathbf R\le_{sd}\mathbf A_{1}\times\dots\times \mathbf A_{m}$ (say that $R$ is a subdirect subalgebra) if 
$R$ is a subdirect relation and $\mathbf R\le\mathbf A_{1}\times\dots\times \mathbf A_{m}$.
For $R\subseteq A^{n}$ by 
$\Sg_{\mathbf A}(R)$ we denote the minimal 
subalgebra of $\mathbf A^{n}$ containing $R$, that is 
the subalgebra of $\mathbf A^{n}$ generated from $R$.

For an equivalence relation $\sigma$ on $A$ and $a\in A$ by $a/\sigma$ we denote the equivalence 
class containing $a$.
For an equivalence relation $\sigma$ on $A$ and 
$B\subseteq A$ denote 
$B/\sigma=\{b/\sigma\mid b\in B\}$.
Similarly, for a relation $R\subseteq A^{n}$ 
denote 
$R/\sigma=\{(b_1/\sigma,\dots,b_{n}/\sigma)\mid (b_1,\dots,b_{n})\in R\}$.
For a binary relation $\sigma$ and $n\ge 2$ by 
$\sigma^{[n]}$ we denote 
the $n$-ary relation 
$\{(a_1,\dots,a_n)\mid \forall i,j\in[n]\colon 
(a_i,a_j)\in\sigma\}$.

For two binary relations $\delta_{1}\subseteq A_{1}\times A_{2}$ 
and $\delta_{2}\subseteq A_{2}\times A_{3}$ 
by $\delta_{1}\circ\delta_2$ we denote the binary 
relation 
$\{(a,b)\mid \exists c\colon (a,c)\in\delta_1\wedge (c,b)\in\delta_2\}$.
Similarly, for $B\subseteq A_{1}$ and $\delta\subseteq A_{1}\times A_{2}$ put 
$B\circ\delta = \{c\mid \exists b\in B \colon (b,c)\in\delta\}$.
For a binary relation $\delta$, we denote
$\delta^{-1}(x,y) = \delta(y,x)$.
A binary subdirect relation $\delta\subseteq A\times B$ 
is called \emph{linked} if
the bipartite graph corresponding to $\delta$ is connected.
A binary subdirect relation $\delta\subseteq A_{1}\times A_{2}$ is called \emph{bijective} 
if $|\delta| = |A_1|= |A_{2}|$.


\textbf{Parallelogram property and rectangularity.}
We say that an $n$-relation $R$ \emph{has the parallelogram property}
if any permutation of its variables gives a relation $R'$ satisfying
\begin{align*}
\forall \ell\in\{1,2,\dots,n-1\} \;\;\;\;\;\;\;\;\;\;\;\;\;\;\;\;&(a_1,\dots,a_{\ell},b_{\ell+1},\dots,b_{n})\in R'\\
\forall a_1,\dots,a_n,b_1,\dots,b_n
\colon \;\;\;\;&(b_1,\dots,b_{\ell},a_{\ell+1},\dots,a_{n})\in R'\;\;\;\Rightarrow\;\;\; (a_1,\dots,a_n)\in R'.\\
&(b_1,\dots,b_{\ell},b_{\ell+1},\dots,b_{n})\in R'
\end{align*}
Note that the parallelogram property plays an important role in universal algebra (see \cite{agnes} for more details).
We say that \emph{the $i$-th variable of a relation $R$ is rectangular}, if
\begin{align*}
\;\;\;\;\;\;\;\;\;\;\;\;\;\;\;\;&(a_1,\dots,a_{i-1},b_{i},a_{i+1},\dots,a_{n})\in R\\
\forall a_1,\dots,a_n,b_1,\dots,b_n
\colon \;\;\;\;&(b_1,\dots,b_{i-1},a_{i},b_{i+1},\dots,b_{n})\in R\;\;\;\Rightarrow\;\;\; (a_1,\dots,a_n)\in R.\\
&(b_1,\dots,b_{i-1},b_{i},b_{i+1},\dots,b_{n})\in R
\end{align*}
As it follows from the definitions,
if a relation has the parallelogram property then it is rectangular.
\emph{The rectangular closure} of a relation $R$  is 
the minimal rectangular relation $R'$ containing $R$.

\textbf{Irreducible congruences.}
For a relation 
$R\subseteq A_{1}\times\dots\times A_{n}$ and 
a congruence $\sigma$ on $A_{i}$, 
we say that the $i$-th variable of the relation $R$ is 
\emph{stable under $\sigma$}
if $(a_{1},\ldots,a_{n})\in R$ and $(a_{i},b_{i})\in\sigma$
imply
$(a_{1},\ldots,a_{i-1},b_{i},a_{i+1},\ldots,a_{n})\in R$.
We say that a relation is \emph{stable under} $\sigma$ if every variable of this relation is stable under $\sigma$.
We say that a congruence $\sigma$ on $\mathbf A$ is \emph{irreducible} if it cannot be represented as an intersection of other binary subalgebras of $\mathbf A\times\mathbf A$ 
that are stable under $\sigma$.
Equivalently, a congruence is irreducible if 
there are no subalgebras 
$\mathbf S_{1},\mathbf S_{2},\dots,\mathbf S_{k}\le \mathbf A/\sigma \times \mathbf A/\sigma$ 
such that 
$0_{\mathbf A/\sigma} = S_{1}\cap S_{2}\cap \dots\cap S_{k}$ and 
$0_{\mathbf A/\sigma} \neq S_{i}$ for every $i\in[k]$. 
Then for an irreducible congruence $\sigma$ on $\mathbf A$
by $\cover{\sigma}$ we denote the minimal
$\delta\le \mathbf A\times \mathbf A$ such that 
$\delta\supsetneq\sigma$ and $\delta$ is stable under $\sigma$.

\textbf{Bridges.}
Suppose $\sigma_{1}$ and $\sigma_{2}$ are congruences on $\mathbf D_{1}$ and $\mathbf D_{2}$, respectively.
A relation $\delta\le \mathbf D_{1}^{2}\times \mathbf D_{2}^{2}$ is called \emph{a bridge} from $\sigma_{1}$ to $\sigma_{2}$ if the following conditions hold:
\begin{enumerate}
    \item the first two variables of $\delta$ are stable under $\sigma_{1}$,
    \item the last two variables of $\delta$ are stable under $\sigma_{2}$,
\item $\proj_{1,2}(\delta) \supsetneq \sigma_{1}$,
$\proj_{3,4}(\delta) \supsetneq \sigma_{2}$,
\item $(a_{1},a_2,a_{3},a_{4})\in \delta$ implies
$(a_1,a_2)\in \sigma_{1}\Leftrightarrow (a_3,a_4)\in \sigma_{2}.$
\end{enumerate}

An example of a bridge 
is the 
relation 
$\delta=\{(a_{1},a_{2},a_{3},a_{4})\mid
a_{1},a_{2},a_{3},a_{4}\in  \mathbb Z_{4}:
a_{1}-a_{2} = 2 a_{3} - 2 a_{4}\}$.
We can check that 
$\delta$ is a bridge from 
the equality relation (0-congruence) 
and $(mod\;2)$ equivalence relation.
The notion of a bridge is strongly related to other notions in Universal Algebra and Tame Congruence Theory
such as similarity and centralizers
(see \cite{RossSlides} for the detailed comparison).

For a bridge $\delta$ by $\widetilde{\delta}$ we denote 
the binary relation defined by 
$\widetilde{\delta}(x,y) = \delta(x,x,y,y)$.

We can compose a bridge $\delta_1$ from 
$\sigma_0$ to $\sigma_1$ and a bridge $\delta_{2}$ from $\sigma_1$ to 
$\sigma_2$ using the following formula:
$$\delta(x_1,x_2,z_{1},z_{2}) = \exists y_{1}\exists y_{2}\; \delta_{1}(x_{1},x_{2},y_{1},y_{2})\wedge \delta_{2}(y_{1},y_{2},z_{1},z_{2}).$$
We can prove (Lemma \ref{LEMBridgeComposition}) that $\delta$ is a bridge from $\sigma_0$ to $\sigma_2$ 
whenever the congruences $\sigma_1$, $\sigma_2$, and $\sigma_3$ are irreducible.
Moreover, $\widetilde {\delta} = \widetilde {\delta_1}\circ \widetilde {\delta_2}$.

%

A congruence $\sigma$ on $\mathbf A=(A;w)$ is called \emph{a perfect linear congruence} if it is
irreducible and 
there exists $\zeta\le \mathbf A\times \mathbf A\times \mathbf Z_{p}$ 
such that $\proj_{1,2} \zeta =\sigma^{*}$
and $(a_{1},a_{2},b)\in\zeta$ 
implies that 
$(a_{1},a_{2})\in \sigma\Leftrightarrow (b=0)$.
Such congruences are important for us because we can control relaxation 
of $\sigma$ to 
$\sigma^{*}$ by an additional element from $\mathbf Z_{p}$. 

In our proofs we compose bridges to get a bridge $\delta$ whose 
binary relation $\widetilde {\delta}$ is linked and then apply the following lemma 
that will be proved in Subection \ref{AuxuliaryStrongStatementsSubsection}.

\begin{lem}\label{LEMMakePerfectCongruenceFromLinked}

Suppose $\sigma$ is a irreducible congruence on $\mathbf A\in \mathcal V_{n}$,
$\delta$ is a bridge from
$\sigma$ to $\sigma$ such that $\widetilde{\delta}$ is linked.
Then $\sigma$ is a perfect linear congruence.
\end{lem}

\subsection{Definition of strong subuniverses}


\textbf{(Binary) absorbing subuniverse.}
We say $B$ is \emph{an absorbing subuniverse} of an algebra 
$\algA$ if 
there exists $t\in \Clo(\mathbf{A})$ such that
$t(B,B,\dots,B,A,B,\dots,B) \subseteq B$ for any position of $A$. Also in this case we say that 
\emph{$B$ absorbs $\algA$ with a term $t$}.

If the operation $t$ can be chosen binary then we say that
$B$ is \emph{a binary absorbing subuniverse} of $\mathbf A$.
To shorten sometimes we will write \emph{BA} instead of 
binary absorbing.
If $t$ can be chosen ternary the we call $B$ 
\emph{a ternary absorbing subuniverse}.
For more information about absorption and its connection to CSP see \cite{barto2017absorption}.

\textbf{Central subuniverse.}
A subuniverse $C$ of $\algA$ is called \emph{central} 
if it is an absorbing subuniverse 
and 
for every $a\in A\setminus C$ 
we have 
$(a,a)\notin \Sg_{\mathbf A}((\{a\}\times C)\cup (C\times \{a\}))$.

Central subuniverses are strongly connected with ternary absorption.

\begin{lem}[\cite{zhuk2021strong}, Corollary 6.11.1]\label{LEMCenterImpliesTernaryAbsorption}
Suppose $B$ is a central subuniverse of $\mathbf A$, 
then $B$ is a ternary absorbing subuniverse of $\mathbf A$.
\end{lem}

In general ternary absorption does not imply 
central subuniverse, but they are equivalent for minimal Taylor algebras (see \cite{minimaltayloralgebras}).
We say that an algebra $\mathbf A$ is \emph{BA and center free} 
if $\mathbf A$ has no proper nonempty binary absorbing subuniverse or 
proper nonempty central subuniverse.

\textbf{Linear and PC congruences.}  There are two different types of 
irreducible congruences.
A congruence $\sigma$ on $\mathbf A$ is called \emph{linear} if
\begin{enumerate}
    \item $\sigma$ is irreducible
    \item $\sigma^{*}$ is a congruence
    \item there exist prime $p$ and $S\le (\sigma^{*})^{[4]}$ 
    such that for any block $B$ of $\sigma^{*}$ 
    there exists $n\ge 0$ with $(B/\sigma;S\cap (B/\sigma)^{4})\cong 
    (\mathbb Z_{p}^{n}; x_1-x_2=x_3-x_4)$.
\end{enumerate}
Notice that the relation $S$ above is a bridge from 
$\sigma$ to $\sigma$ such that 
$\widetilde S =\proj_{1,2}(S) = \proj_{3,4}(S)= \sigma^{*}$.

An irreducible congruence is called a \emph{PC congruence} if it is not linear. 
Notice that a congruence $\sigma$ is 
an irreducible/PC/linear congruence if and only if 
$0_{\mathbf A/\sigma}$ is an irreducible/PC/Linear congruence.

\begin{lem}\label{LEMLinearEquivalentConditions}

Suppose $\sigma$ is an irreducible congruence on $\mathbf A$.
Then the following conditions are equivalent:
\begin{enumerate}
\item $\sigma$ is a linear congruence;
\item there exists a bridge $\delta$ from $\sigma$ to $\sigma$ such that 
$\widetilde\delta\supsetneq \sigma$.
\end{enumerate}
\end{lem}

Another important fact is that 
there cannot be a bridge between PC and linear congruences.

\begin{lem}\label{LEMNoBridgeBetweenDifferentTypes}

Suppose $\sigma_1$ is a linear congruence, 
$\sigma_2$ is an irreducible congruence, 
$\delta$ is a bridge from $\sigma_1$ to $\sigma_2$.
Then $\sigma_2$ is a also linear congruence.
\end{lem}

Unlike bridges for linear congruences, bridges from PC congruences are trivial.

\begin{lem}\label{LEMPCBridgesAreTrivial}

Suppose $\sigma$ is a PC congruence on $A$.
Then any reflexive bridge $\delta$ from $\sigma$ to $\sigma$
such that $\proj_{1,2}(\delta) = \proj_{3,4}(\delta)=\sigma^{*}$
can be represented as 
$\delta(x_1,x_2,x_3,x_4)= \sigma(x_1,x_3)\wedge \sigma(x_2,x_4)$ or 
$\delta(x_1,x_2,x_3,x_4)= \sigma(x_1,x_4)\wedge \sigma(x_2,x_3)$.
\end{lem}

\begin{lem}\label{LEMBridgeTOPCCongruence}

Suppose 
$\delta$ is a bridge from 
a PC congruence $\sigma_1$ on $\mathbf A_{1}$ to 
an irreducible congruence $\sigma_2$ on 
$\mathbf A_{2}$, 
$\proj_{1,2}(\delta) = \sigma_1^{*}$, and
$\proj_{3,4}(\delta) = \sigma_2^{*}$.
Then \begin{enumerate}
    \item $\sigma_2$ is a PC congruence;
    \item $\mathbf A_1/\sigma_1\cong \mathbf A_2/\sigma_2$;
    \item $\{(a/\sigma_{1},b/\sigma_{2})\mid (a,b)\in\widetilde \delta\}$ is bijective; 
    \item $\delta(x_1,x_2,x_3,x_4) = 
\widetilde \delta(x_1,x_3)\wedge \widetilde \delta(x_2,x_4)$ or
$\delta(x_1,x_2,x_3,x_4) = 
\widetilde \delta(x_1,x_4)\wedge \widetilde \delta(x_2,x_3)$.
\end{enumerate}
\end{lem}

The following claims show the connection of the new definitions 
with the linear and PC subuniverses from the 
original proof of the CSP Dichotomy Conjecture
\cite{zhuk2020proof}.

\begin{lem}\label{LEMLInearOnTheTopIsEasy}

Suppose $\sigma$ is a linear congruence on $\mathbf A\in\mathcal V_{n}$ 
such that 
$\sigma^{*} = A^{2}$.
Then $\mathbf A/\sigma\cong \mathbf Z_{p}$ for some prime $p$.
\end{lem}


\begin{lem}\label{LEMPCOnTheTopIsEasy}

Suppose $\sigma$ is a PC congruence on $\mathbf A$ and  
$\sigma^{*} = A^{2}$. 
Then $\mathbf A/\sigma$ is a PC algebra.
\end{lem}



\textbf{All types of subuniverses.} 
Suppose $\varnothing\neq C\lneq B\le A$. We write 
\begin{itemize}
\item $C<_{\TBA}^{A} B$
if $C$ is a BA subuniverse of $B$.
\item $C<_{\TC}^{A} B$
if $C$ is a central subuniverse of $B$.
\item $C<_{\TD}^{A} B$
if there exists an irreducible  congruence $\sigma$ such that 
\begin{enumerate}
    \item $B^2\subseteq \sigma^{*}$;
    \item $C=B\cap E$ for some block $E$ of $\sigma$;
    \item $B/\sigma$ is BA and center free.
\end{enumerate}
\item $C<_{\TL}^{A} B$ if $C<_{\TD}^{A} B$  and the 
congruence $\sigma$ from the definition of $<_{\TD}^{A}$ is linear.
\item $C<_{\TPC}^{A} B$ if $C<_{\TD}^{A} B$  and the 
congruence $\sigma$ from the definition of $<_{\TD}^{A}$ is a PC congruence.
\item $C<_{\TS}^{A} B$
if there exists a BA and central (simultaneously) subuniverse $D$ in $\mathbf B$
such that $D\le C$.
\end{itemize}

When we want to specify what congruence was used
in the definition we write 
$C<_{T(\sigma)}^{A} B$. 
Sometimes, we also put a congruence there even if $T\in\{\TBA,\TC,\TS\}$, which 
means that $\sigma$ is a full congruence.
If $C<_{T}^{A} B$ then we say that $C$ is \emph{a subuniverse of $B$ of type $T$}.
Sometimes we also call $B$ 
\emph{a dividing subuniverse} for the type $\TD$,
\emph{a linear subuniverse} for the type $\TL$,
and 
\emph{a PC subuniverse} for the type $\TPC$.
Also, we say that 
$\sigma$ is a \emph{dividing/linear/PC 
congruence for $B\le A$} 
if 
$C<_{\mathcal T(\sigma)}^{A}B$
for some $C$ and $\mathcal T = \TD/\TL/\TPC$.
We say that an algebra $\mathbf A$ is \emph{S-free} if 
there is no $D\le A$ such that $D<_{\TBA}A$ and $D<_{\TC}A$.
Equivalently, 
an algebra $\mathbf A$ is \emph{S-free} if
there does not exist $C<_{\TS}\mathbf A$.

Sometimes instead of 
$C<_{\TBA}^{A} B$,
$C<_{\TC}^{A} B$,
and $C<_{\TS}^{A} B$
we write 
$C<_{\TBA} B$,
$C<_{\TC} B$,
and $C<_{\TS} B$, which is justified because $A$ is irrelevant to the definition.
Also, we write
$C<_{\TBA,\TC} B$ meaning that 
$C<_{\TBA} B$ and $C<_{\TC} B$.




We write 
$C\lll^{A} B$ if there exist $B_{0},B_{1},\dots,B_{n}\subseteq B$
and $T_{1},\dots,T_{n}\in\{\TBA, \TC, \TS,\TD\}$ 
such that 
$C=B_{n}<_{T_{n}}^{A} B_{n-1}<_{T_{n-1}}^{A}<
\dots <_{T_{2}}^{A}<B_1<_{T_{1}}^{A} B_{0} = B$.
Notice that $n$ can be 0 and 
the relation $\lll^{A}$ is reflexive.
We say that 
a congruence \emph{comes from 
$C\lll^{A} B$} if it is one of the dividing congruences
used in the sequence $C\lll^{A} B$.
We usually write 
$B\lll A$ instead of $B\lll^{A} A$.
We write $C\le_{T(\sigma)}^{A}B$ if $C=B$ or $C<_{T(\sigma)}^{A} B$.

Let us introduce the types 
$\TML,\TMPC,\TMD$ of subuniverses.
Suppose $T\in\{\TL,\TPC,\TD\}$.
We write 
$C<_{\mathcal {M}T}^{A} B$ if
$C\neq\varnothing$ and 
$C = C_1\cap \dots\cap C_{t}$, 
where $C_i<_{T}^{A} B$ for every $i\in[t]$.


Notice that we do not allow empty subuniverses and 
the condition
$\varnothing\lll A$ never holds.
Nevertheless,  sometimes we need to allow an empty set.
In this case we add a dot above and write 
$B\dot\lll A$ meaning that 
$B\lll A$ or $B = \varnothing$.
With the same meaning we use dots in the following notations 
$C\dot<_{T}^{A}B$ or
$C\dot\le_{T}^{A}B$.

\subsection{Properties of strong subuniverses}\label{SUBSECTIONSTRONGSUBALGEBRASPROPERTIES}

Recall that all the algebras in the following 
statements are assumed finite idempotent  algebras having a WNU term operation (Taylor).
To avoid listing all the possible types 
in the following lemmas we assume that 
if the type $T$ is not specified then 
$T\in\{\TBA,\TC,\TS,\TPC,\TL,\TD\}$.
If we write the type $\mathcal {M}T$ then we 
assume that 
$T\in\{\TPC,\TL,\TD\}$ and, therefore,
$\mathcal{M}T\in\{\TMPC,\TML,\TMD\}$.

\begin{lem}\label{LEMUbiquity}

Suppose $B\lll A$ and $|B|>1$.
Then there exists $C<_{T}^{A} B$, where $T\in \{\TBA, \TC, \TL, \TPC\}$.
\end{lem}

\begin{lem}\label{LEMPropagation}

Suppose $f\colon \mathbf A\to \mathbf A'$ is a surjective homomorphism,
then 
\begin{enumerate}
\item[(f)] $C\lll^{A} B \Rightarrow f(C)\lll f(B)$;
\item[(b)] $C'\lll^{A'} B' \Rightarrow f^{-1}(C')\lll f^{-1}(B)$;    
\item[(ft)] $C<_{T(\sigma)}^{A} B\lll A \Longrightarrow 
(f(C)=f(B)
\text{ or } 
f(C)<_{S}f(B)
\text{ or }
f(C)<_{T}^{A'}f(B))$;
\item[(bt)] $C'<_{T(\sigma)}^{A'} B' \Rightarrow f^{-1}(C')<_{T(f^{-1}(\sigma))}^{A'}f^{-1}(B')$;
\item[(fs)] $T\in\{\TBA,\TC, \TS\}$ and $C<_{T}B$ $\Longrightarrow f(C)\le_{T}f(B)$;
\item[(fm)] 
$C\le_{\mathcal{M}T}^{A}B\lll A$ and $f(B)$ is S-free $\Longrightarrow f(C)\le_{\mathcal{M}T}^{A'}f(B)$;
\item[(bm)] 
$C'\le_{\mathcal{M}T}^{A'}B'\lll A'\Longrightarrow f^{-1}(C)\le_{\mathcal{M}T}^{A}f^{-1}(B)$.
\end{enumerate}
\end{lem}

\begin{cor}\label{CORPropagationModuloCongruence}

Suppose $\delta$ is a congruence on $\mathbf A$. Then
\begin{enumerate}
\item[(f)] $C\lll^{A} B \Rightarrow C/\delta\lll^{A/\delta} B/\delta$;
\item[(t)] $C<_{T(\sigma)}^{A} B\lll A \Longrightarrow 
(C/\delta=B/\delta
\text{ or } 
C/\delta<_{\TS}B/\delta
\text{ or }
C/\delta<_{T}^{A/\delta}B/\delta)$;
\item[(s)] $T\in\{\TBA,\TC, \TS\}$ and $C<_{T}B$ $\Longrightarrow C/\delta\le_{T}B/\delta$;
\item[(m)] 
$C\le_{\mathcal{M}T}^{A}B\lll A$ and $B/\delta$ is S-free $\Longrightarrow C/\delta\le_{\mathcal{M}T}^{A/\delta}B/\delta$.
\end{enumerate}
\end{cor}

\begin{cor}\label{CORPropagateFromFactor}
 Suppose $\delta$ is a congruence on $\mathbf A$,
$\mathbf B,\mathbf C\le \mathbf A$. Then 
\begin{enumerate}
    \item[(f)] $C/\delta\lll^{A/\delta}
B/\delta \Longleftrightarrow C\circ\delta\lll^{A}
B\circ \delta$;
    \item[(t)] $C/\delta<_{T}^{A/\delta}
B/\delta \Longleftrightarrow C\circ\delta<_{T}^{A}
B\circ \delta$.
\end{enumerate}
\end{cor}

\begin{cor}\label{CORPropagateMultiplyByCongruence}
 
Suppose $\delta$ is a congruence on $\mathbf A$. Then 
\begin{enumerate}
\item[(f)] $C\lll^{A}B  \Rightarrow C\circ \delta \lll^{A} B\circ \delta$;
    \item[(t)] $C<_{T(\sigma)}^{A}B\lll^{A}A\Longrightarrow (C\circ\delta= B\circ \delta \text{ or }
    C\circ\delta<_{\TS}^{A} B\circ \delta \text{ or }
    C\circ\delta<_{T}^{A} B\circ \delta)$;
    \item[(e)] $\delta\subseteq \sigma$ and $C<_{T(\sigma)}^{A}B\lll^{A}A\Longrightarrow 
    C\circ\delta<_{T}^{A} B\circ \delta$.
\end{enumerate}
\end{cor}

\begin{cor}\label{CORPropagateToRelations}
 
Suppose $R\le_{sd} A_{1}\times\dots\times A_{n}$,
$B_{i}\lll A_{i}$ for $i\in[n]$. Then 
\begin{enumerate}
    \item[(r)] $R\cap (B_1\times\dots\times B_{n}))\dot\lll R$;
    \item[(r1)] $\proj_{1}(R\cap (B_1\times\dots\times B_{n}))\dot\lll A_{1}$;
    \item[(b)] $\forall i\colon  C_{i}\lll^{A_{i}}B_{i}\Longrightarrow (R\cap (C_1\times\dots\times C_{n}))\dot\lll^{R} 
    (R\cap (B_1\times\dots\times B_{n}))$;   
    \item[(b1)] $\forall i\colon  C_{i}\lll^{A_{i}}B_{i}\Longrightarrow \proj_{1}(R\cap (C_1\times\dots\times C_{n}))\dot\lll^{A_1} \proj_{1}(R\cap (B_1\times\dots\times B_{n}))$;
    \item[(m)] $\forall i\colon  C_{i}\le_{\mathcal{M}T}^{A_{i}}B_{i} \Longrightarrow R\cap (C_1\times\dots\times C_{n})\dot\le_{\mathcal{M}T}^{R} R\cap (B_1\times\dots\times B_{n})$;
    \item[(m1)] $\forall i\colon  C_{i}\le_{\mathcal{M}T}^{A_{i}}B_{i}$, 
    $\proj_{1}(R\cap (B_1\times\dots\times B_{n}))$ is S-free $\Longrightarrow$
    
    \;\;\;\;\;\;\;\;\;\;\;\;\;\;\;\;\;\;\;\;\;\;\;\;\;\;\;\;\;\;
\;\;\;\;\;\;\;\;\;\;\;\;\;\;\;\;\;\;\;\;\;\;\;$\proj_{1}(R\cap (C_1\times\dots\times C_{n}))\dot\le_{\mathcal{M}T}^{A_{1}} \proj_{1}(R\cap (B_1\times\dots\times B_{n}))$.
\end{enumerate}

\end{cor}

For binary absorbing and central subuniverses  
we can prove a stronger claim.

\begin{lem}[\cite{zhuk2021strong}, Corollaries 6.1.2 and  6.9.2]\label{LEMBACenterImplies}
Suppose $R\le A_{1}\times\dots\times A_{n}$,
$C_{i}\le_{T} A_{i}$ for every $i\in[n]$, where 
$T\in\{\TBA,\TC\}$.
Then 
$\proj_{1}(R\cap (C_{1}\times\dots\times C_{N}))\dot\le_{T} A_{1}$.
\end{lem}

\begin{lem}\label{LEMIntersectALL}

Suppose $B\lll A$, $D\lll A$. 
Then 
\begin{enumerate}
    \item[(i)] $B\cap D\dot\lll A$;
    \item[(t)] $C<_{T(\sigma)}^{A}B \Rightarrow C\cap D\dot\le _{T(\sigma)}^{A} B\cap D$.
\end{enumerate}
\end{lem}

\begin{thm}\label{THMMainStableIntersection}
 
Suppose 
\begin{enumerate}
    \item $C_{i}<_{T_{i}(\sigma_{i})}^{A} B_{i}\lll A$, where 
    $T_{i}\in\{\TBA, \TC,\TS,\TL,\TPC\}$ for $i=1,2,\dots,n$, $n\ge 2$;
    \item $\bigcap\limits_{i\in[n]} C_{i} = \varnothing$;
    \item $B_{j}\cap\bigcap\limits_{i\in[n]\setminus\{j\}} C_{i} \neq  \varnothing$
    for every $j\in[n]$.
\end{enumerate}
Then one of the following conditions hold:
\begin{enumerate}
    \item[(ba)] $T_1=\dots= T_n=\TBA$;
    \item[(l)] 
        $T_1=\dots= T_n=\TL$ and 
        for every $k,\ell\in [n]$
        there exists a bridge $\delta$ from 
    $\sigma_{k}$ and $\sigma_{\ell}$ 
    such that 
    $\widetilde \delta= \sigma_{k}\circ \sigma_{\ell}$;
    \item[(c)] $n=2$ and $ T_1= T_2=\TC$;
    \item[(pc)] $n=2$, $ T_1= T_2=\TPC$,
    and $\sigma_{1} = \sigma_{2}$.
\end{enumerate}     
\end{thm}

\begin{cor}\label{CORMainStableIntersection}
 
Suppose 
\begin{enumerate}
    \item $R\le_{sd} A_{1} \times \dots\times 
    A_{n}$;
    \item $C_{i}<_{T_{i}(\sigma_{i})}^{A_{i}} B_{i}\lll A_{i}$, where 
    $T_{i}\in\{\TBA, \TC,\TS,\TL,\TPC\}$ for $i=1,2,\dots,n$, $n\ge 2$;
    \item $R\cap (C_{1}\times\dots\times C_{n}) = \varnothing$;
    \item 
    $R\cap (C_{1}\times\dots\times C_{j-1}\times 
    B_{j}\times C_{j+1}\times\dots\times C_{n})  \neq  \varnothing$
    for every $j\in[n]$.
\end{enumerate}
Then one of the following conditions hold:
\begin{enumerate}
    \item[(ba)] $T_1=\dots=T_n=\TBA$;
    \item[(l)] 
        $T_1=\dots= T_n=\TL$ and 
        for every $k,\ell\in [n]$
        there exists a bridge $\delta$ from 
    $\sigma_{k}$ and $\sigma_{\ell}$ such that $\widetilde \delta=
    \sigma_{k}\circ \proj_{k,\ell}(R)\circ \sigma_{\ell}$;
    \item[(c)] $n=2$ and $T_1=T_2=\TC$;
    \item[(pc)] $n=2$, $T_1=T_2=\TPC$, $A_{1}/\sigma_{1} \cong A_{2}/\sigma_{2}$,
    and the relation $\{(a/\sigma_1,b/\sigma_{2})\mid (a,b)\in R\}$ is bijective. 
\end{enumerate}         
\end{cor}

\begin{remark}
Notice that sometimes we want to have several restrictions 
on one coordinate of a relation. To keep the statement 
of Corollary \ref{CORMainStableIntersection} simple we do not add this possibility into the claim, but we can always duplicate the coordinate of the relation
and apply restrictions separately on different coordinates.
\end{remark}

\begin{lem}\label{LEMMultiTypeStillStable}
 
Suppose $C\le_{\mathcal{M}T}^{A}B$. 
Then $C<_{T}^{A}\dots<_{T}^{A}B$ and $C\lll^{A} B$.
\end{lem}



\begin{lem}\label{LEMPreserveLinkdness}
 
Suppose $R\le_{sd} \mathbf A_1\times\mathbf A_2$,
$C_{i}\le_\mathcal {\TMD}^{A_{i}} B_{i}\lll A_i$
for $i\in\{1,2\}$, 
$S$ is a rectangular closure of $R$,
$R\cap (B_1\times B_2)\neq \varnothing$,
$S\cap (C_1\times C_2)\neq \varnothing$.
Then 
$R\cap (C_{1}\times C_{2})\neq \varnothing$.
\end{lem}

\begin{lem}\label{LEMMaximalMultExtention}
Suppose $C_{1}<_{\mathcal{M}T}^{A}B_1\lll A$, 
$B_{2}\lll A$, 
$C_{1}\cap B_{2} = \varnothing$, 
$B_{1}\cap B_{2}\neq \varnothing$,
$\sigma$ is a maximal congruence on $\mathbf A$ such that 
$(C_{1}\circ \sigma)\cap B_2 = \varnothing$.
Then 
$\sigma = \omega_{1}\cap\dots\cap \omega_{s}$, 
where 
$\omega_{1},\dots,\omega_s$ are congruences
of type $T$ 
on $\mathbf A$
such that $\omega_{i}^{*}\supseteq B_{1}^{2}$.
\end{lem}




\subsection{Auxiliary Statements}\label{AuxuliaryStrongStatementsSubsection}

\begin{lem}[\cite{miklos} Lemma 4.7]\label{LEMExistenceOfSpesialWNULemma}
Suppose $w$ is an idempotent WNU operation on $A$.
Then there exists a special idempotent WNU operation $w'\in\Clo(w)$ of arity $n^{n!}$.
\end{lem}

\begin{lem}[\cite{zhuk2020proof} Corollary 8.17.1] \label{LEMBuildingPerfectCongruence}
Suppose $\sigma$ is an irreducible congruence on
$\mathbf A\in \mathcal V_{n}$
and $\delta$ is a bridge from
$\sigma$ to $\sigma$ such that $\widetilde{\delta} = A^{2}$.
Then $\sigma$ is a perfect linear congruence.
\end{lem}

\begin{lem}[\cite{zhuk2020proof} Lemma 6.3]\label{LEMBridgeComposition}
Suppose $\sigma_{1}$, $\sigma_{2}$, $\sigma_{3}$ are irreducible congruences, 
$\rho_{1}$ is a bridge  from $\sigma_{1}$ to $\sigma_{2}$,
$\rho_{2}$ is a bridge from $\sigma_{2}$ to $\sigma_{3}$.
Then the formula
$$\rho(x_1,x_2,z_{1},z_{2}) = \exists y_{1}\exists y_{2}\; \rho_{1}(x_{1},x_{2},y_{1},y_{2})\wedge \rho_{2}(y_{1},y_{2},z_{1},z_{2})$$
defines a bridge from $\sigma_{1}$
to $\sigma_{3}$.
Moreover, 
$\widetilde{\rho} = 
\widetilde{\rho_{1}}\circ\widetilde{\rho_{2}}$.
\end{lem}

\begin{LEMMakePerfectCongruenceFromLinkedLEM}
Suppose $\sigma$ is a irreducible congruence on $\mathbf A\in \mathcal V_{n}$,
$\delta$ is a bridge from
$\sigma$ to $\sigma$ such that $\widetilde{\delta}$ is linked.
Then $\sigma$ is a perfect linear congruence.
\end{LEMMakePerfectCongruenceFromLinkedLEM}

\begin{proof}
Let $\delta^{-1}$ be the bridge defined by 
$\delta^{-1}(y_1,y_2,x_1,x_2) = \delta(x_1,x_2,y_1,y_2)$.
Since $\widetilde{\delta}$ is linked, 
$\underbrace{\widetilde{\delta}\circ\widetilde{\delta^{-1}}\circ 
\dots\widetilde{\delta}\circ\widetilde{\delta^{-1}}}_{2N}=
A^2$ for sufficiently large $N$.
Using Lemma \ref{LEMBridgeComposition} we compose $2N$ bridges $\delta, \delta^{-1},
\dots,\delta, \delta^{-1}$ and obtain a new bride $\delta'$
such that 
$\widetilde \delta' = A^{2}$.
By Lemma \ref{LEMBuildingPerfectCongruence}
$\sigma$ is a perfect linear congruence.
\end{proof}


\begin{lem}\label{LEMNoAbsCenterPCInLinearAlgebra}
Suppose 
$B\le \mathbf Z_{p}\times \dots \times \mathbf Z_{p}$.
Then there does not exist 
$C<_{\mathcal T} B$ such that $T\in \{\TBA,\TC\}$.
\end{lem}
\begin{proof}
It is sufficient to check for any term $\tau$
that $\tau^{\mathbf Z_{p}}(a,\dots,a,x)$
takes all the values if the last variable is not dummy.
Hence, there cannot be an absorbing subuniverse in 
$\mathbf Z_{p}\times \dots \times \mathbf Z_{p}$.
\end{proof}

\section{Proof of the CSP Dichotomy Conjecture}\label{SECTIONCSPDICHOTOMYPROOF}

In this section we prove 
Theorems \ref{THMCSPDReductionsAreSafe}
and \ref{THMCodimensionOneTheorem}
that show the correctness of Zhuk's algorithm 
for the CSP.
We start with Subsection \ref{SUBSECTIONCSPDADDITIONALDEFINITIONS}, where we give additional definitions
such as
irreducible, linked, and crucial instances of the CSP. 
Crucial instances are the instances that have no solutions but any weakening of the instance (like removing a constraint) gives an instance with a solution. We will show 
any constraint in a crucial and consistent enough instance has the parallelogram property.
This allows us to define a congruence 
for every constraint and its variable and talk about connectedness of the variables by bridges.
Also, we explain how we weaken the instance: usually we just replace a constraint by a weaker constraint but sometimes we also need to 
disconnect two constraints by adding an additional variable, which leads us to the notion of Expanded Coverings.

In the next subsection we give all the auxiliary statements necessary for the main proof.
Mainly we explain how our 
new theory works for the CSP and it works 
especially well if the solution set of the instance is subdirect.

The core of the proof of both main theorems is Theorem 
\ref{THMMainInductiveCSPClaim}, which states 
that all constraints in a crucial instance
have the parallelogram property, 
and there exists a crucial expanded covering with a connected subinstance.
Additionally it states that 
a restriction of the domains to strong subalgebras 
cannot destroy all the solutions.
As in the original proof, 
Theorem 
\ref{THMMainInductiveCSPClaim} is proved by induction on the size of the domain
but this time we connect the variables using 
dividing congruences coming from the reductions, 
which significantly simplifies the whole argument.

\subsection{Additional definitions}\label{SUBSECTIONCSPDADDITIONALDEFINITIONS}

\textbf{CSP Instances.}
An instance $\mathcal I$ of $\CSP(\Gamma)$ is a list (or conjunction) 
of constraints of the form $R(x_{1},\dots,x_{m})$,
where $R\in \Gamma$. 
We write $C\in \mathcal I$ meaning that $C$ is a constraint of $\mathcal I$.
For an instance $\mathcal I$ and a constraint 
$C$ by $\Var(\mathcal I)$ and $\Var(C)$ we denote the set of variables 
appearing in $\mathcal I$ and $C$, respectively.
Every variable $x$ appearing in an instance has its domain, which we denote by $D_{x}$.
Every domain can be viewed as an algebra 
$\mathbf D_{x}=(D_{x};w^{\mathbf D_{x}})\in \mathcal V_{m}$.
A subset of constraints of 
an instance $\mathcal I$ is called
\emph{a subinstance} of $\mathcal I$. 
Then for every constraint 
$R(x_{1},\dots,x_{h})$
the relation $R$ is a subuniverse of 
$\mathbf D_{x_1}\times\dots\times \mathbf D_{x_{h}}$.
We say that \emph{a solution set of an instance $\mathcal I$ is subdirect} 
if for every $x$ and every $a\in D_{x}$ the instance 
has a solution with $x=a$.

\textbf{Reductions.} 
\emph{A reduction} $D^{(\top)}$ for a CSP instance 
$\mathcal I$ 
is mapping that assign a 
subuniverse 
$D_{x}^{(\top)}\le \mathbf D_{x}$ to every variable $x$ of $\mathcal I$.
$D$ can be viewed as a trivial reduction.
For two reductions  $D^{(\bot)}$ and $D^{(\top)}$
we write $D^{(\bot)}\lll D^{(\top)}$ and $D^{(\bot)}\le_{T} D^{(\top)}$ 
whenever
$D_{i}^{(\bot)}\lll D_{i}^{(\top)}$ for every $i\in I$ and 
$D_{i}^{(\bot)}\le_{T} D_{i}^{(\top)}$ for every $i\in I$, respectively.
For an instance $\mathcal I$ and a reduction $D^{(\top)}$ 
by $\mathcal I^{(\top)}$ we denote the instance whose 
variables $x$ are restricted to $D_{x}^{(\top)}$. 
A reduction $D^{(\top)}$ is called \emph{nonempty} if 
$D_{x}^{(\top)}\neq\varnothing$ for every $x$.

\textbf{Induced congruences.}
For a relation $R$ of arity $n$ and $i\in[n]$ by $\ConOne(R,i)$
we denote the binary relation $\sigma(y,y')$ defined by
$$\exists x_{1}\dots\exists x_{i-1}\exists x_{i+1}\dots\exists x_{n}\;R(x_{1},\ldots,x_{i-1},y,x_{i+1},\ldots,x_{n})\wedge
R(x_{1},\ldots,x_{i-1},y',x_{i+1},\ldots,x_{n}).$$
For a constraint $C = R(x_{1},\ldots,x_{n})$
by $\ConOne(C,x_{i})$ we denote $\ConOne(R,i)$.
For an instance $\mathcal I$ by
$\Congruences(\mathcal I,x)$ we denote
the set $\{\ConOne(C,x)\mid C\in \mathcal I\}$.
By $\Congruences(\mathcal I)$ we denote 
$\bigcup\limits_{x\in\Var(\mathcal I)}\Congruences(\mathcal I,x)$.
Notice that the $i$-th variable of a relation $R$ is rectangular if and only if $R$ is stable under $\ConOne(R,i)$.
Moreover, if the $i$-th variable of a subdirect relation $R$ is rectangular then $\ConOne(R,i)$ is a congruence;


\textbf{Linear-type and PC-type.}
We say that a relation $R$ is of \emph{the PC/Linear type} if 
$R$ is rectangular 
and each congruence $\ConOne(R,i)$ is 
a PC/Linear congruence.
We say that an instance has \emph{the PC/Linear type} if 
all of its constraints are of the PC/Linear type.

\textbf{A path and a tree-covering.} We say that $z_{1}-C_{1}-z_{2}-\dots - C_{l-1}-z_{l}$ is
\emph{a path} in a CSP instance $\mathcal I$ if $z_{i},z_{i+1}\in\Var(C_{i})$.
We say that \emph{a path $z_{1}-C_{1}-z_{2}-\dots- C_{l-1}-z_{l}$  connects $b$ and $c$}
if there exists $a_{i}\in D_{z_{i}}$ for every $i$
such that
$a_{1} = b$, $a_{l} = c$, and
the projection of $C_{i}$ onto $z_{i}, z_{i+1}$
contains the tuple $(a_{i},a_{i+1})$.
We say that an instance is \emph{a tree-instance} if there is no a path
$z_{1}-C_{1}-z_{2}-\dots -z_{l-1}-C_{l-1}-z_{l}$
such that $l\ge 3$, $z_{1} = z_{l}$, and all the constraints $C_{1},\ldots,C_{l-1}$ are different.

\textbf{Consistency conditions.}
A CSP instance $\mathcal I$ is called \emph{1-consistent} if $\proj_{z}(C)=D_{z}$
for any 
constraint $C$ of $\mathcal I$ and  
any variable $z$ of $C$.
A reduction $D^{(\top)}$ is called \emph{1-consistent}
for an instance $\mathcal I$ if 
the instance $\mathcal I^{(\top)}$ is 1-consistent.
An instance $\mathcal I$ is called \emph{cycle-consistent} if
it is 1-consistent and 
for every variable $z$ and $a\in D_{z}$
any path starting and ending with $z$ in $\mathcal I$ 
connects $a$ and $a$.
Other types of local consistency and 
its connection with the complexity of the CSP
are considered in \cite{kozik2016weak,brady2022notes}.

\textbf{Linkedness and irreducibility.}
An instance $\mathcal I$ is called \emph{linked}
if for every variable $z\in\Var(\mathcal I)$ and every $a,b\in D_{z}$
there exists a path starting and ending with $z$ in $\mathcal I$ that connects $a$ and $b$. 
We say that an instance $\mathcal I$ is \emph{fragmented}
if $\Var(\mathcal I)$ can be divided into 2 disjoint nonempty sets $\mathbf{X_1}$ and
$\mathbf{X_2}$ such that 
$\Var(C)\subseteq \mathbf X_1$ or
$\Var(C)\subseteq \mathbf X_1$ for any $C\in\mathcal I$.
An instance $\mathcal I$ is called \emph{irreducible} if
there is no instance $\mathcal I'$ satisfying the following conditions:
\begin{enumerate}
    \item $\Var(\mathcal I')\subseteq \Var(\mathcal I)$,
    \item each constraint of $\mathcal I'$ is a projection of 
    a constraint of $\mathcal I$ on some variables,
    \item $\mathcal I'$ is not fragmented,
    \item $\mathcal I'$ is not linked,
    \item the solution set of $\mathcal I'$ is not subdirect.
\end{enumerate}

\textbf{Weakening of an instance.} 
We say that a constraint $R_{1}(y_{1},\ldots,y_{t})$ is \emph{weaker or equivalent to}
a constraint $R_{2}(z_{1},\ldots,z_s)$
if $\{y_{1},\ldots,y_{t}\}\subseteq \{z_{1},\ldots,z_s\}$
and 
$R_{2}(z_{1},\ldots,z_s)$ implies $R_{1}(y_{1},\ldots,y_{t})$.
We say that 
$C_{1}$ is \emph{weaker than} 
$C_{2}$ 
if $C_{1}$ is weaker or equivalent to $C_{2}$
but $C_{1}$ does not imply $C_{2}$.
\emph{The weakening} of a constraint $C$ in an instance $\mathcal I$ is the 
replacement of $C$ by all weaker constraints.
An instance $\mathcal I'$ is called \emph{a weakening of 
an instance $\mathcal I$} 
if $\Var(I')\subseteq \Var(I)$ are every constraint 
of $\mathcal I'$ is weaker or equivalent to
a constraint of $\mathcal I$.

\textbf{Crucial instance.} 
We say that a variable $y_{i}$ of the constraint $R(y_{1},\ldots,y_{t})$ is \emph{dummy} if 
$\R$ does not depend on its $i$-th variable.
Let $D_{i}'\subseteq D_{i}$ for every $i$.
Suppose $D^{(\top)}$ is a reduction for an instance $\mathcal I$.
A constraint $C$ of $\mathcal I$ is called \emph{crucial in $D^{(\top)}$}
if it has no dummy variables, $\mathcal I^{(\top)}$ has no solutions but
the weakening of $C\in\Theta$ gives an instance $\mathcal I'$ 
with a solution in $D^{(\top)}$.
An instance $\mathcal I$ is called \emph{crucial in $D^{(\top)}$}
if it has at least one constraint and
all its constraints are crucial in $D^{(\top)}$.

\begin{remark}\label{GetCrucialInstance}
Suppose $\mathcal I^{(\top)}$ has no solutions.
Then we can iteratively replace every constraint
by all weaker constraints having no dummy variables until it is crucial
in $D^{(\top)}$. Notice that
$R\le \mathbf D_{x_{1}}\times \dots\times \mathbf D_{x_n}$ for any weaker constraint 
$R(x_1,\dots,x_n)$ we introduce.
\end{remark}

\textbf{Relations defined by instances.}
For an instance $\mathcal I$ and $x_{1},\dots,x_{n}\in\Var(\mathcal I)$
by 
$\mathcal I(x_{1},\dots,x_{n})$ we denote the
set of all tuples $(a_1,\dots,a_n)$ such that 
$\mathcal I$ has a solution with 
$x_{i} = a_i$ for every $i$.
Thus, $\mathcal I(x_{1},\dots,x_{n})$
defines an $n$-ary relation.
Note that the obtained relation 
is a subuniverse of 
$\mathbf D_{x_1}\times \dots \times \mathbf D_{x_n}$ 
as it is defined by a primitive positive formula over the relations in $\mathcal I$ (see \cite{geiger1968closed,bond1,bond2}).

\textbf{Expanded coverings.}
For an instance $\mathcal I$ by $\Expanded(\mathcal I)$ (\emph{Expanded Coverings}) we denote the set of all instances $\mathcal I'$
such that there exists a mapping $S:\Var(\Omega')\to\Var(\Omega)$
satisfying the following conditions:
\begin{enumerate}
\item
if $x\in\Var(\mathcal I)\cap\Var(\mathcal I')$ then 
$S(x) = x$;
\item $D_{x} = D_{S(x)}$ for every $x\in \Var(\mathcal I')$;
\item for every constraint $R(x_{1},\ldots,x_{n})$ of $\mathcal I'$
either the variables $S(x_{1}),\ldots,S(x_{n})$ are different and the constraint $R(S(x_{1}),\ldots,S(x_{n}))$ is weaker or equivalent to some constraint of $\Omega$,
or $S(x_{1}) = \dots = S(x_{n})$ and $\{(a,a,\ldots,a)\mid a\in D_{x_{1}}\}\subseteq R$;
\end{enumerate}

An expanded covering $\mathcal I'$ of $\mathcal I$ is called \emph{a covering} if
for every constraint $R(x_{1},\ldots,x_{n})$ of $\mathcal I'$
the constraint $R(S(x_{1}),\ldots,S(x_{n}))$ is in $\mathcal I$.
An instance is called 
\emph{a tree-covering} if it is a covering and also a tree-instance.
For a variable $x$ we say that $S(x)$ is \emph{the parent of x}
and $x$ is \emph{a child} of $S(x)$.
The same child/parent terminology will also be applied to constraints.

The following easy facts 
can be derived from the definition.
\begin{enumerate}
    \item[(p1)] If we replace every variable $x$ by $S(x)$ in an expanded covering of $\mathcal I$
    (and remove all the constraints $R(x,x,\dots,x)$)
    we get a weakening of $\mathcal I$; 
    \item[(p2)] A weakening is an expanded covering such that $S(x)= x$ for every $x$;
    \item[(p3)] any solution of an instance can be naturally expanded to a solution of its expanded covering;
    \item[(p4)] if an instance is 1-consistent and 
    its expanded covering is a tree-covering, then 
    the solution set of the covering is subdirect;
    \item[(p5)] the union (union of all constraints) of two expanded coverings is 
    also a expanded covering;
    \item[(p6)] an expanded covering of an expanded covering is an expanded covering.
    \item[(p7)] an expanded covering of a cycle-consistent irreducible instance is 
    cycle-consistent and irreducible (see Lemma \ref{LEMExpandedConsistencyLemma}).
    \item[(p8)] any reduction of an instance can be 
    naturally extended to its expanded covering; moreover, if the reduction was 
    1-consistent for the instance, it is 1-consistent for the covering.
\end{enumerate}

\textbf{Connected instances.} A bridge $\delta\subseteq D^{4}$ is called \emph{reflexive} if
$(a,a,a,a)\in \delta$ for every $a\in D$.
We say that two congruences $\sigma_{1}$ and $\sigma_{2}$ on $\mathbf D_{x}$ are \emph{adjacent}
if there exists a reflexive bridge from $\sigma_{1}$ to $\sigma_{2}$.
Since we can always put
$\delta(x_{1},x_{2},x_{3},x_{4}) = \sigma(x_{1},x_{3})\wedge \sigma (x_{2},x_{4})$,
any proper congruence $\sigma$ is adjacent with itself.
We say that two rectangular constraints $C_{1}$ and $C_{2}$ are \emph{adjacent} in a common variable $x$ if
$\ConOne(C_{1},x)$ and $\ConOne(C_{2},x)$ are adjacent.
An instance $\mathcal I$ is called \emph{connected} if
all its constraints are rectangular, all the congruences of $\Congruences(\mathcal I)$ are irreducible, and
the graph, whose vertices are constraints 
and edges are adjacent constraints,
is connected.

\subsection{Auxiliary statements}


\begin{lem}[\cite{zhuk2020proof}, Lemma 6.1]\label{LEMExpandedConsistencyLemma}
Suppose $\mathcal I$ is a cycle-consistent irreducible CSP instance
and $\mathcal I'\in\Expanded(\mathcal I)$.
Then $\mathcal I'$ is cycle-consistent and irreducible.
\end{lem}

\begin{lem}\label{LEMMinimalPCLinearReductionIsConsistent}
Suppose 
\begin{enumerate}
\item $D^{(1)}$ is a 1-consistent reduction for an instance 
$\mathcal I$, 
\item $D_{x}^{(1)}$ is S-free for every $x\in\Var(\mathcal I)$,
\item $T\in\{\TPC, \TL,\TD\}$,
\item $D^{(1)}\lll D$,
\item $D_{x}^{(2)}\le_{\mathcal M T} D_{x}^{(1)}$ is a minimal $\mathcal{M}T$ subuniverse 
for every $x\in\Var(\mathcal I)$.
\end{enumerate}
Then either there exists a constraint $C$ such that 
$C^{(2)}$ is empty, or 
$\mathcal I^{(2)}$ is 1-consistent. 
\end{lem}

\begin{proof}
If $C^{(2)}$ is empty for some constraint $C$ then we are done.
Otherwise, consider some constraint $R(x_1,\dots,x_n)$.
By Lemma \ref{LEMPropagation}(fm) 
$\proj_{i}(R^{(2)})\le_{\TMD}^{D_{x_{i}}} D_{x_{i}}^{(1)}$
for every $i\in [n]$.
Since $D_{x_{i}}^{(2)}$ is a minimal 
subuniverse $B$ such that 
$B\le_{\TMD}^{D_{x_{i}}} D_{x_{i}}^{(1)}$
we have $\proj_{i}(R^{(2)}) = D_{x_{i}}^{(2)}$.
Hence $\mathcal I^{(2)}$ is 1-consistent.
\end{proof}

\begin{lem}\label{LEMCrucialMeansIrreducible}
Suppose $R(x_1,\dots,x_n)$ is a rectangular constraint of a 1-consistent instance $\mathcal I$, 
$R(x_1,\dots,x_n)$ is crucial
in $D^{(\top)}$.
Then $\ConOne(R,i)$ is an irreducible congruence for every $i\in[n]$.
\end{lem}

\begin{proof}
To simplify notations assume that $i=1$.
Assume the converse, then $\ConOne(R,1) = \omega_1\cap \omega_2$ for 
some $\ConOne(R,1)\lneq\omega_1,\omega_2\le \mathbf D_{x_1}\times\mathbf D_{x_1}$.
Define the relation $R_{j}$ for $j\in\{1,2\}$ by 
$$R_{j}(x_1,x_2,\dots,x_n) = \exists y
(R(y,x_2,\dots,x_n)\wedge \omega_1(y,x_1)).$$ 
Since $\omega_{i}\supsetneq \ConOne(R,1)$ we have
$R_{j}\supsetneq R$ for each $j\in\{1,2\}$.
Since $\omega_1\cap\omega_2 = \ConOne(R,1)$ we have
$R= R_{1}\cap R_{2}$. Thus $R(x_1,\dots,x_n)$ could be replaced by
two weaker constraints $R_1(x_1,\dots,x_n)$ and $R_2(x_1,\dots,x_n)$
and still be without a solution in $D^{(1)}$. This contradicts the cruciality.
\end{proof}

\begin{lem}\label{LEMBridgeFromRelation}
Suppose $R\le_{sd} \mathbf A_{1}\times\dots\times \mathbf A_{n}$,
the first and the last variables of $R$ are rectangular,
and there exist
$(b_{1},a_{2},\ldots,a_n),(a_{1},\ldots,a_{n-1},b_{n})\in R$
such that $(a_{1},a_{2},\ldots,a_n)\notin R$.
Then there exists a bridge $\delta$ from $\ConOne(R,1)$ to $\ConOne(R,n)$
such that $\widetilde{\delta} = \proj_{1,n}(R)$.
\end{lem}
\begin{proof}
The required bridge can be defined by
$$\delta(x_{1},x_{2},y_{1},y_{2}) =
\exists z_{2}\dots\exists z_{n-1}\;
R(x_{1},z_{2},\ldots,z_{n-1},y_{1})\wedge
R(x_{2},z_{2},\ldots,z_{n-1},y_{2}).$$
In fact, since the first and the last variables of $R$ are rectangular, 
we have $(x_{1},x_{2})\in\ConOne(R,1)$
if and only if
$(y_{1},y_{2})\in\ConOne(R,n)$.
It remains to notice that 
$(b_{1},a_{1},a_{n},b_{n})\in\delta$,
$(b_{1},a_{1})\notin\ConOne(R,1)$, 
and $\widetilde{\delta} = \proj_{1,n}(R)$.
\end{proof}

\begin{lem}\label{LEMConnectedProperties}
Suppose $\mathcal I$ is a cycle-consistent connected instance.
Then 
\begin{enumerate}
    \item[(a)] any two constraints with a common variable 
    are adjacent;
    \item[(b)] for any constraints $C_{1}, C_{2}\in\mathcal I$, 
    variables $x_1\in \Var(C_{1})$, $x_{2}\in \Var(C_2)$, 
    and any path from $x_1$ to $x_2$, 
    there exists a bridge $\delta$ from $\ConOne(C_{1},x_1)$ to 
    $\ConOne(C_2,x_2)$ such that 
    $\widetilde \delta$ contains 
    all pairs connected by this path;
    \item[(p)] if $\mathcal I$ is linked then
    $\ConOne(C,x)$ is a perfect linear congruence 
    for every constraint $C\in\mathcal I$ and $x\in\Var(C)$.
\end{enumerate}
\end{lem}

\begin{proof}
Let us prove (a) for two constraints $C_{1}$ and $C_{2}$ with a common variable $x$.
Since $\mathcal I$ is connected, there exists a path 
$z_{1}-C_{1}'-z_{2}-C_{2}'-\dots - z_{\ell}-C_{\ell}'-z_{\ell+1}$
such that $z_{1} = z_{\ell+1} = x$, 
$C_{1}' = C_{1}$, $C_{\ell}' = C_{2}$, 
$C_{j}'$ and $C_{j+1}$ are adjacent in a common variable 
$z_{j+1}$ for each $j\in[\ell-1]$.
Let $\omega_j$ be a reflexive
bridge from $\ConOne(C_{j}',z_{j+1})$ to $\ConOne(C_{j+1}',z_{j+1})$.
By Lemma \ref{LEMBridgeFromRelation}
for every $i\in[\ell]$ there exists a 
bridge $\delta_i$ from 
$\ConOne(C_{i}',z_{i})$ to 
$\ConOne(C_{i}',z_{i+1})$
such that $\widetilde{\delta_{i}}= \proj_{z_{i},z_{i+1}}(C_{i}')$.
Since $\mathcal I$ is connected, all the congruences 
$\ConOne(C_{i}',z_{i}),\ConOne(C_{i}',z_{i+1})$ are irreducible.
Composing bridges 
$\delta_1,\omega_1,\delta_2,\omega_2,\dots,\delta_{\ell-1},\omega_{\ell-1},\delta_{\ell}$ 
using Lemma \ref{LEMBridgeComposition}
we get the required bridge from $\ConOne(C_{1}',z_1)$ to $\ConOne(C_{\ell}',z_{\ell+1})$.
Since $\mathcal I$ is cycle-consistent, 
the bridge is reflexive, and therefore $C_{1}$ and $C_{2}$ are 
adjacent.

To prove (b) we repeat the whole argument of (a) for the path in $\mathcal I$.
Since we already proved (a),
$C_{j}'$ and $C_{j+1}'$ are adjacent in a common variable 
$z_{j+1}$ for any path.
As a result we obtain the required bridge 
$\ConOne(C_{1}',z_1)$ to $\ConOne(C_{\ell}',z_{\ell+1})$.


Let us prove (p). Since $\mathcal I$ is connected, for any $a,b\in D_{x}$ 
there exists a path from $x$ to $x$ connecting $a$ and $b$.
Let us build a bridge $\delta_{a,b}$ using (b) for this path. 
Since $\mathcal I$ is cycle-consistent, $\delta_{a,b}$ is reflexive.
Composing all the bridges $\delta_{a,b}$ for $a,b\in D_{x}$
we get a bridge $\delta$ from $\ConOne(C,x)$ to $\ConOne(C,x)$ 
such that $\widetilde{\delta}= D_{x}^{2}$.
By Lemma \ref{LEMBuildingPerfectCongruence}
$\ConOne(C,x)$ is a perfect linear congruence.
\end{proof}

\begin{lem}[\cite{zhuk2021strong}, Lemma 5.6]\label{LEMExistenceOfTreeCoverings}
Suppose $D^{(\top)}$ is a reduction for an instance $\mathcal I$, 
$D^{(\bot)}$ is an inclusion maximal 1-consistent reduction 
for $\mathcal I$
such that $D^{(\bot)}\le D^{(\top)}$. Then for every variable 
$y\in\Var(\mathcal I)$ there exists a 
tree-covering $\Upsilon_{y}$ of $\mathcal I$ 
such that 
$\Upsilon_{y}^{(\top)}(y)$ defines $D_{y}^{(\bot)}$.
\end{lem}

\begin{cor}\label{CORExistenceOfTreeCoverings}
Suppose $D^{(\top)}$ is a reduction of a 1-consistent instance $\mathcal I$, $D^{(\top)}\lll D$, 
$D^{(\bot)}$ is an inclusion-maximal nonempty 1-consistent 
reduction of $\mathcal I$ such that 
$D^{(\bot)}\le D^{(\top)}$.
Then $D^{(\bot)}\lll^{D}D^{(\top)}\lll D$. 
\end{cor}
\begin{proof}
By Lemma \ref{LEMExistenceOfTreeCoverings} 
for every variable 
$y\in\Var(\mathcal I)$ there exists a 
tree-covering $\Upsilon_{y}$ of $\mathcal I$ 
such that 
$\Upsilon_{y}^{(\top)}(y)$ defines $D_{y}^{(\bot)}$.
Since $\mathcal I$ is 1-consistent, the solution set of
$\Upsilon_{y}$ can be viewed as a subdirect relation.
By Corollary \ref{CORPropagateToRelations}(r1) 
we obtain 
$\Upsilon_{y}^{(\top)}(y)=D_{y}^{(\bot)}\lll^{D_{y}}D_{y}^{(\top)}$ 
\end{proof}

\begin{lem}\label{LEMFindOneConsistentForAll}
Suppose $D^{(1)}$ is a 1-consistent reduction of a cycle-consistent instance $\mathcal I$, 
$D^{(1)}\lll D$, 
$B<_{T}^{D_{x}} D_{x}^{(1)}$ for some variable $x$,
and $T\in\{\TBA, \TC,\TPC\}$.
Then there exists a nonempty 1-consistent reduction 
$D^{(2)}\lll D^{(1)}$ such that $D_{x}^{(2)}\le B$.
Moreover, 
\begin{enumerate}
    \item if $T\in\{\TBA,\TC\}$ then $D^{(2)}\le_{T} D^{(1)}$;
    \item if $T = \TPC$ and $D_{y}^{(1)}$ is S-free for every $y\in\mathcal I$ 
    then $D^{(2)}\le_{\TMPC} D^{(1)}$.
\end{enumerate}
\end{lem}

\begin{proof}
Define the reduction $D^{(\top)}$ by
$D_{x}^{(\top)} = B$ and $D_{y}^{(\top)}=D_{y}^{(1)}$
for every $y\neq x$.
Let $D^{(2)}$ be an inclusion maximal 1-consistent reduction 
for $\mathcal I$ such that $D^{(2)}\le D^{(\top)}$.
By Lemma \ref{LEMExistenceOfTreeCoverings} 
for every variable 
$y\in\Var(\mathcal I)$ there exists a 
tree-covering $\Upsilon_{y}$ of $\mathcal I$ 
such that 
$\Upsilon_{y}^{(\top)}(y)$ defines $D_{y}^{(2)}$.

Assume that $D_{y}^{(2)}=\varnothing$ and 
$\Upsilon_{y}^{(\top)}$ has no solutions for some $y$. Since $\mathcal I$ is 1-consistent, 
the solution set of $\Upsilon_{y}$ can be viewed as a subdirect relation.
By Corollary \ref{CORMainStableIntersection} 
there should two children of $x$ in $\Upsilon_{y}$ such that 
if we restrict them to $D_{x}^{(\top)}$ we kill all the solutions of 
$\Upsilon_{y}$.
Since $\Upsilon_{y}$ is a tree-covering of $\mathcal I$ and 
$\mathcal I$ is cycle-consistent, this cannot happen.

Thus, $D^{(2)}$ is a nonempty reduction.
Assume that $\mathcal T\in\{\TBA,\TC\}$.
Again considering the solution set 
$\Upsilon_{y}$ and applying 
Lemma \ref{LEMBACenterImplies} 
we derive that $D_{y}^{(2)}\le_{\mathcal T} D_{y}^{(1)}$, 
and therefore $D^{(2)}\le_{\mathcal T}D^{(1)}$.
For $\mathcal T=\TPC$ we do the same but apply 
Corollary \ref{CORPropagateToRelations}(rm)
instead and obtain 
$D^{(2)}\le_{\TMPC}^{D}D^{(1)}$.
\end{proof}

\begin{lem}\label{LEMParalPropertyFromCrucialInMultiType}
Suppose \begin{enumerate}
    \item $D^{(1)}$ is a 1-consistent reduction  
for a constraint $R(x_1,\dots,x_{n})$,
\item $T\in\{\TL,\TPC,\TD\}$,
\item $D^{(2)}\le_{\mathcal M T}^{D}D^{(1)}\lll D$,
\item $R(x_1,\dots,x_{n})$ is crucial (as the whole instance) in 
$D^{(2)}$.
\end{enumerate}
Then $R$ has the parallelogram property and 
$\ConOne(R,i)$ is a congruence of type $T$ 
such that $\ConOne(R,i)^{*}\supseteq (D_{x_i}^{(1)})^{2}$
for every $i\in[n]$.
Moreover, if $T = \TPC$
then $n=2$.
\end{lem}

\begin{proof}
First, let us prove that $R$ has the parallelogram property.
We need to check the parallelogram property for any splitting of the variables of $R$ into two disjoint sets. 
Without loss of generality we assume that this splitting is 
$\{x_{1},\dots,x_{k}\}$ and $\{x_{k+1},\dots,x_{n}\}$.
Let us define a binary relation 
$R'\le_{sd} \mathbf E_1 \times \mathbf E_2$ 
by $((a_1,\dots,a_k),(a_{k+1},\dots,a_n))\in R'\Leftrightarrow
(a_1,\dots,a_n)\in R$, where 
$E_1 = \proj_{1,\dots,k}(R)$, 
$E_2 = \proj_{k+1,\dots,n}(R)$. 
Let us define $E_1^{(1)}$, $E_{2}^{(1)}$, $E_1^{(2)}$, $E_{2}^{(2)}$, naturally (we just reduce 
the corresponding coordinates to $D^{(1)}$ or $D^{(2)}$).
By Corollary \ref{CORPropagateToRelations}(r) and (m)
we have $E_i^{(2)}\le_{\mathcal M T}^{E_{i}}E_{i}^{(1)}\lll E_{i}$ for each $i\in\{1,2\}$.
Put 
$S' = R'\circ R'^{-1}\circ R'$ 
and $S = \{(a_1,\dots,a_{n})\mid ((a_1,\dots,a_k),(a_{k+1},\dots,a_n))\in S'\}$.
Since 
$R'^{(2)}=\varnothing$,
Lemma \ref{LEMPreserveLinkdness} implies that 
$S'^{(2)}= \varnothing$. Since $R(x_1,\dots,x_n)$ is crucial 
and $S\supseteq R$, 
we obtain $S = R$. Hence $R$ has the parallelogram property.

Put 
$E = \proj_{1}(R\cap (D_{x_1}\times D_{x_2}^{(2)}\times 
\dots \times D_{x_{n}}^{(2)}))$.
If $E = \varnothing$ then 
the constraint could be weakened to 
$R_{0}(x_2,\dots,x_n)$,
where $R_{0}(x_2,\dots,x_n) = \exists x_1 R(x_1,\dots,x_n)$, 
which contradicts the cruciality.
Hence $E\neq\varnothing$ and 
by Corollary \ref{CORPropagateToRelations}(r1)
$E\lll D_{x_1}$.
Since $R(x_1,\dots,x_n)$ is crucial in $D^{(2)}$ 
we have $E\cap D_{x_{1}}^{(2)}=\varnothing$.

If $E\cap D_{x_1}^{(1)}=\varnothing$, choose 
 $C$, $B$, and $\mathcal T\in\{\TBA,\TC,\TS,\TL,\TPC,\TD\}$ such that
$D_{x_1}^{(1)}\lll^{D_{x_1}} C<_{\mathcal T}^{D_{x_1}}B
\lll D_{x_{1}}$,
$E\cap C=\varnothing$, and
 $E\cap B\neq \varnothing$.
Since $R^{(1)}$ is not empty
and $D^{(2)}\le_{\mathcal M T}^{D}D^{(1)}$, 
Corollary \ref{CORMainStableIntersection} implies that
$\mathcal T = T$.
If $E\cap D_{x_1}^{(1)}\neq\varnothing$ put 
$C = D_{x_1}^{(2)}$ and $B = D_{x_1}^{(1)}$.
Thus, in both cases we have 
 $C<_{\mathcal M T}^{D_{x_1}}B\lll D_{x_{1}}$.

Notice that $E\circ \ConOne(R,1) = E$,
hence $(E\circ \ConOne(R,1))\cap C=\varnothing$.
Let $\sigma\supseteq \ConOne(R,1)$ be a maximal congruence 
such that $(E\circ \sigma)\cap C=\varnothing$.
If $\sigma\supsetneq \ConOne(R,1)$ then 
we weaken the constraint 
$R(x_1,\dots,x_n)$ to 
$R_0(x_1,\dots,x_n)$, where $R_0(x_1,\dots,x_n) = 
\exists z R(z,x_2,\dots,x_n)\wedge \sigma(z,x_1)$.
The obtained constraint must have a solution 
in $D^{(2)}$, which means that $(E\circ \sigma)\cap D_{x_1}^{(2)}\neq\varnothing$
and contradicts $(E\circ \sigma)\cap C=\varnothing$.
Thus, $\sigma = \ConOne(R,1)$.
By Lemma \ref{LEMMaximalMultExtention} 
$\ConOne(R,1)=\omega_1\cap \dots\cap \omega_{s}$ for some  congruences
$\omega_1,\dots,\omega_s$ of type $T$ such that $\omega_{i}^{*}\supseteq B^{2}$.
By Lemma \ref{LEMCrucialMeansIrreducible}
$\ConOne(R,1)$ is irreducible, 
hence $\ConOne(R,1)$ is a congruence of type $T$ 
satisfying $\ConOne(R,1)^{*}\supseteq B^{2}\supseteq (D_{x_{1}}^{(1)})^{2}$.

It remains to show that $n=2$ for $T = \TPC$. By Corollary \ref{CORMainStableIntersection}
there exist $i,j\in[n]$, 
$B_{i}<_{PC(\sigma_{i})}^{D_{x_{i}}}D_{x_{i}}^{(1)}$,
and 
$B_{j}<_{PC(\sigma_{j})}^{D_{x_{j}}}D_{x_{j}}^{(1)}$
such that 
$R$ has no tuples whose $i$-th element is from $B_{i}$ and 
$j$-th element is from $B_{j}$.
If $n\ge 3$ the we can existentially quantify all the variables of 
$R$ but $i$-th and $j$-th and obtain a weaker constraint 
without a solution in $D^{(2)}$, 
which contradicts cruciality. 
\end{proof}

\begin{lem}\label{LEMGetABridgeFromSubdirectPCLinearInstance}
Suppose
\begin{enumerate} 
    \item $\mathcal I$ is an instance having a subdirect solution set,
    \item $D^{(1)}$ is a reduction for $\mathcal I$ such that $D_{x}^{(1)}\lll D_{x}$ for every $x$,
    \item $C$ is a constraint in $\mathcal I$ of type $T\in\{\TPC,\TL\}$,
    \item $B<_{\mathcal T(\xi)}^{D_{z}} D_{z}^{(1)}$ for some variable $z$, where
    $\mathcal T\in\{\TBA, \TC, \TS,\TPC,\TL\}$,
    \item if $T = \TPC$ then $\mathcal T\in\{\TPC,\TL\}$,    
    \item $\mathcal I^{(1)}$ has a solution,
    \item $\mathcal I^{(1)}$ has no solutions with $z\in B$,
    \item weakening of $C$ in $\mathcal I$ gives an instance with a solution in $D^{(1)}$ and $z\in B$. 
\end{enumerate}
Then 
$\mathcal T=T$ and for any variable $x$ of $C$ 
there exists a bridge $\delta$ from 
$\xi$ to $\ConOne(C,x)$ 
such that $\widetilde \delta$ contains $\mathcal I(z,x)$.
\end{lem}
\begin{proof}
By $D^{(2)}$ we denote the reduction that differs from $D^{(1)}$ only 
on the variable $z$ and $D_{z}^{(2)} = B$.

Choose some variable $x_{0}$ in $C$.
Let $\omega = \ConOne(C,x_{0})$. 
By condition 3, $\omega$ is either a PC, or  linear congruence.
We take $\mathcal I$, replace the variable $x_{0}$ in $C$ by $x_{0}''$,
all the other variables $x_{i}$  by
$x_{i}'$, 
and add
a new constraint 
$\omega^{*}(x_{0}',x_{0}'')$.
The obtained instance we denote by $\Theta$.
Extend our reduction $D^{(1)}$ and $D^{(2)}$ to $\Theta$ by 
$D_{x_{0}''}^{(1)} = D_{x_{0}''}^{(2)}=D_{x_{0}''}=D_{x_{0}}$,
$D_{x_{i}'}^{(2)} = D_{x_{i}}^{(2)}$, and $D_{x_{i}'}^{(1)} = D_{x_{i}}^{(1)}$.
Notice that the solution set of $\Theta$ is still subdirect
and $\Theta^{(2)}$ has a solution.


Let us consider a minimal reduction $D^{(\top)}$ 
for $\mathcal I$
such that 
$D^{(1)}_{x}\lll^{D_{x}} D^{(\top)}_{x}\lll D_{x}$ for every $x$ 
and 
$\Theta^{(2)}\wedge \mathcal I^{(\top)}\wedge \omega(x_0,x_0'')$
has a solution.
If $D^{(\top)}\neq D^{(1)}$, 
choose a variable $y$ and $G <_{\mathcal T_{0}(\nu)}^{D_{y}} D_{y}^{(\top)}$
such that 
$D_{y}^{(1)}\lll^{D_{y}} G$.
If $D^{(\top)}= D^{(1)}$ then 
put $G = B$, $y = z$, $\mathcal T_{0} = \mathcal T$, and $\nu = \xi$.

Define a new reduction $D^{(\bot)}$
by $D_{y}^{(\bot)}=G$
and $D_{x}^{(\bot)} =D_{x}^{(\top)}$ for every $x\neq y$.
We extend the reduction $D^{(\top)}$ and $D^{(\bot)}$ 
so that the reductions on $x_{i}$ and $x_{i}'$ coincide
and 
$D_{x_{0}''}^{(\top)} = D_{x_{0}''}^{(\bot)}=D_{x_{0}}$.
Since the instance $\Theta\wedge \mathcal I\wedge \omega(x_0,x_0'')$
has a subdirect solution set, 
by Corollary \ref{CORMainStableIntersection}
the types $\mathcal T$ and $\mathcal T_{0}$ are the same.
Moreover, if $\mathcal T \in \{\TL,\TPC\}$, 
there exists a bridge $\delta'$ from $\xi$ to $\nu$
such that 
$\widetilde \delta' \supseteq \Theta(z',x_{0}'')\circ \mathcal I(x_{0},y)\supseteq 
\mathcal I(z,y)$
(if $z=y$ then it is just a trivial reflexive bridge).
Let $F$ be 
set of possible values of $x_{0}''$ in the solutions of $\Theta^{(2)}\wedge \mathcal I^{(\top)}\wedge \omega(x_0,x_0'')$.
In other words, 
$F = \Theta^{(2)}(x_{0}'')\cap (\mathcal I^{(\top)}(x_{0})\circ\omega)$.
Since the variable $x_{0}''$ only appears in the constraint $C$ and $\omega^{*}$ 
we have $F\circ \omega = F$.
By Corollaries \ref{CORPropagateToRelations}(r1)
$F\lll D_{x_{1}}$.
By Lemma \ref{LEMUbiquity}
we find a single block $E$ of $\omega$ such that 
$\{E\}\lll^{D_{x_{0}}/\omega} F/\omega$.
By Corollary \ref{CORPropagateMultiplyByCongruence}(f) 
$E\lll F$.

As
$\mathcal I^{(\top)}$ has a solution with $x_{0}\in E\subseteq F$,   
the instance $\Theta^{(\top)}$
has a solution with $x_{0}',x_{0}''\in E$.
As $\Theta^{(2)}$ has a solution with $x_{0}''\in E\subseteq F$, 
the instance $\Theta^{(\bot)}$ has a solution with 
$x_{0}''\in E$.
As $\mathcal I^{(\bot)}$ has no solutions 
with $x_{0}\in F\supseteq E$, the instance $\Theta^{(\bot)}$ has no solutions 
with $x_{0}',x_{0}''\in E$.
Consider two cases:

Case 1. $T = \TL$ and $\mathcal T\in\{\TBA,\TC,\TS\}$.
Let $G_1$ be the
set of all values of $x_{0}'$ in 
solutions of $\Theta^{(\top)}$ with $x_{0}''\in E$, 
and 
$G_2$ be the
set of all values of $x_{0}'$ in 
solutions of $\Theta^{(\bot)}$ with $x_{0}''\in E$.
By Lemma \ref{LEMBACenterImplies} 
$G_{2}\le_{\mathcal T} G_{1}$.
By Corollary \ref{CORPropagationModuloCongruence}(s)
$G_{2}/\omega\le_{\mathcal T} G_{1}/\omega$.
Since $E\subseteq  G_{1}$ and $E\not\subseteq G_{2}$, 
we have $G_{2}/\omega<_{\mathcal T} G_{1}/\omega$.
By the construction of $\Theta$, $G_{1}$ and $G_{2}$ are 
from the same block of $\omega^{*}$.
Hence we obtained a BA or central subuniverse in
a block of $\omega^{*}$, which contradicts the properties of 
a linear congruence.

Case 2. $\mathcal T\in\{\TPC,\TL\}$.
Choose $E'$ and $E''$  
such that 
$E\lll^{D_{x_{0}}} E' <_{T_{0}(\zeta)}^{D_{x_{0}}} E''\lll D_{x_{0}}$ and
$\Theta^{(\bot)}$ has a solution with $x_{0}''\in E, x_{0}'\in 
E''$ 
but has no solutions with $x_{0}''\in E, x_{0}'\in E'$.
Notice that we can choose $E'$, $E''$, and $\zeta$ stable under $\omega$ 
as $E'$ and $E''$ may come from 
$\{E\}\lll D_{x_{0}}/\omega$
and Corollary \ref{CORPropagateMultiplyByCongruence}(f).
By Corollary \ref{CORMainStableIntersection}
we derive that 
$T_0=\mathcal T_{0}$ and there exists a bridge $\delta''$ from 
$\nu$ to
$\zeta$ 
such that $\widetilde \delta'' =\Theta(y'',x_0'')\supseteq \mathcal I(y,x_{0})$.
Notice that 
$\zeta$ must be equal to $\omega$ as otherwise 
$\omega^{*}\subseteq\zeta$ and 
any solution 
of $\Theta^{(\bot)}$ with $x_{0}\in E$ and $x_{0}'\in E''$
also 
satisfies $x_{0}'\in E'$. 
It remains to compose 
bridges $\delta'$ and $\delta''$ to obtain a bridge $\delta$
from $\xi$ to $\omega$
such that  
$\widetilde \delta \supseteq\mathcal I(z,y)\circ \mathcal I(y,x_0)\supseteq \mathcal I(z,x_{0})$.
\end{proof}

\begin{cor}\label{CORSameTypeReductionAndConstraint}
Suppose 
$\mathcal I$ is an instance having subdirect solution set, 
$D^{(1)}$ and $D^{(2)}$ are reductions for $\mathcal I$,
$\mathcal I^{(1)}$ has a solution,
$D^{(2)}\le_{T}^{D}D^{(1)}\lll D$, where $T\in\{\TBA,\TC,\TS\}$,
$C$ is a constraint of $\mathcal I$ of type $\TL$. 
Then $C$ is not crucial in $D^{(2)}$.
\end{cor}

\begin{proof}
We take a minimal reduction $D^{(\top)}$ 
such that 
$D_{x}^{(\top)}\in \{D_{x}^{(1)},D_{x}^{(2)}\}$
for every $x$ and 
$\mathcal I^{(\top)}$ has a solution.
Take some variable $z$ such that 
$D_{z}^{(\top)}= D_{z}^{(1)}$, 
take $B = D_{z}^{(2)}$, and apply Lemma \ref{LEMGetABridgeFromSubdirectPCLinearInstance}
for the reduction $D^{(\top)}$.
\end{proof}

\subsection{Main Statements}
\begin{thm}\label{THMMainInductiveCSPClaim}
Suppose 
\begin{itemize}
    \item $D^{(1)}$ is a 1-consistent reduction of an irreducible, cycle-consistent instance $\mathcal I$; 
    \item $D^{(1)}\lll D$.
\end{itemize}
If  $\mathcal I$ is crucial in $D^{(1)}$ then  (1a) and ((1b) or (1c)).
    \begin{enumerate} 
    \item[(1a)] every constraint of $\mathcal I$ has the parallelogram property;
    \item[(1b)] $\mathcal I$ is a connected linear-type instance having a subdirect solution set;  
    \item[(1c)] there exists a expanded covering $\mathcal J$ of $\mathcal I$ with 
    a linked connected subinstance $\Upsilon$ such that the solution set of $\Upsilon$ 
    is not subdirect and $\mathcal J$ is crucial in $D^{(1)}$.
    \end{enumerate}
If $D^{(2)}\le_{\mathcal T} D^{(1)}$ is a 1-consistent reduction
    of $\mathcal I$, where $\mathcal T\in\{\TBA, \TC\}$, and $\mathcal I^{(1)}$ has a solution, then
    \begin{enumerate}
    \item[(2)]     $\mathcal I^{(2)}$ has a solution.
    \end{enumerate}
\end{thm}

\begin{proof}
We prove the claim  by induction on the size of $D^{(1)}$.

\textbf{Let us prove (2)} first.
Assume that $\mathcal I^{(2)}$ has no solutions.
Weaken $\mathcal I^{(2)}$ to make it crucial in $D^{(2)}$
and denote the obtained instance by $\mathcal I'$.
By the inductive assumption for $\mathcal I'$ and $D^{(2)}$ 
the instance $\mathcal I'$ satisfies (1a) and also (1b) or (1c).
Assume that $\mathcal I'$ satisfies (1c), then  
there exists an expanded covering $\mathcal J$ of $\mathcal I'$ with a linked connected subinstance 
$\Upsilon$ such that $\mathcal J$ is crucial in $D^{(2)}$ .
Let $x$ be a variable of a constraint $C\in\Upsilon$.
By Lemma \ref{LEMConnectedProperties}(p)
$\ConOne(C,x)$ is a perfect linear congruence.
Choose $\zeta\subseteq \mathbf D_{x}\times \mathbf D_{x}\times \mathbf Z_{p}$
such that
$(y_{1},y_{2},0)\in \zeta\Leftrightarrow (y_{1},y_{2})\in\ConOne(C,x)$
and $\proj_{1,2}(\zeta) = \cover{\ConOne(C,x)}$.
Let us replace the variable $x$ of $C$ in $\mathcal J$ by $x'$
and add the constraint $\zeta(x,x',z)$.
The obtained instance we denote by $\Theta$.
We extend the reductions $D^{(1)}$ and $D^{(2)}$ to $x'$ by 
$D_{x'}^{(1)}=D_{x'}^{(2)}=D_{x'}$.
Let $E_{1}$ and $E_{2}$ be the set of all $z$ such that 
$\Theta$ has a solution in $D^{(1)}$ and in $D^{(2)}$, respectively.
By Lemma \ref{LEMBACenterImplies} 
$E_{2}\le_{\mathcal T}E_{1}$.
Since $\mathcal J$ is crucial in $D^{(2)}$, 
$\Theta^{(2)}$ must have some solution
but not a solution with $z=0$.
Since $\mathcal I^{(1)}$ has a solution, $\mathcal J^{(1)}$ also has a solution 
and  
$\Theta^{(1)}$ has a solution with $z=0$.
Hence, $E_{1}$ contains at least two different elements, 
$0\in E_{1}$ and $0\notin E_{2}$. 
Since $\mathbf Z_{p}$ does not have proper subalgebras of size greater than 1, 
we have $E_{1} = \mathbf Z_{p}$,
which contradicts the fact that 
$\mathbf Z_{p}$ has no BA or central subuniverses
(Lemma \ref{LEMNoAbsCenterPCInLinearAlgebra}).

Assume that $\mathcal I'$ satisfies (1b).
Applying Corollary \ref{CORSameTypeReductionAndConstraint} we derive 
a contradiction.

\textbf{Let us prove (1a) and ((1b) or (1c))}. Notice that 
$|D_{x}^{(1)}|>1$ for some $x$ as otherwise we would get a contradiction of 
1-consistency of $\mathcal I$ and its cruciality in $D^{(1)}$.
Consider two cases.

Case 1. There exists a nontrivial $\TBA$ or central subuniverse on some $D^{(1)}_{x}$.
By Lemma \ref{LEMFindOneConsistentForAll} there exists a
1-consistent reduction $D^{(2)}\le_{T} D^{(1)}$ for $\mathcal I$ such that
$T\in\{BA,C\}$. 
Let us show that $\mathcal I$ is crucial in $D^{(2)}$.
Let $\mathcal J$ be obtained from $\mathcal I$ by a weakening of 
some constraint $C\in\mathcal I$. 
Then $\mathcal J^{(1)}$ has a solution. 
By the inductive assumption for $D^{(2)}$ we derive that 
$\mathcal J^{(2)}$ has a solution. 
Thus, any weakening of $\mathcal I$ has a solution in $D^{(2)}$
and $\mathcal I$ is crucial in $D^{(2)}$.
Again applying the inductive assumption to $D^{(2)}$ 
we derive the required conditions, which completes this case.

Case 2. Otherwise. 
By Lemma \ref{LEMUbiquity}, there exists  $E<_{T(\sigma)}^{D_{z}} D^{(1)}_{z}$ for some $z$
and $T\in\{\TPC,\TL\}$.
Let us prove (1a) first.
Choose some constraint $C$ in $\mathcal I$. 
Since $\mathcal I$ is crucial, a weakening of this constraint gives a solution 
in $D^{(1)}$. By $s(x)$ we denote the value of $x$ in this solution.
For every variable $x$ 
choose the minimal $D^{(2)}_{x}\le_{\mathcal M T}D^{(1)}$ containing $s(x)$.
By Lemma \ref{LEMMultiTypeStillStable}
$D^{(2)}\lll^{D} D^{(1)}$ and
by Lemma \ref{LEMMinimalPCLinearReductionIsConsistent} we have two subcases.

Subcase 1. $\mathcal I^{(2)}$ is 1-consistent. 
Weaken the instance to make it crucial in $D^{(2)}$. Notice that the constraint $C$ must be in there because 
weakening of $C$ gives a solution in $D^{(2)}$.
Then applying the inductive assumption to the obtained instance 
(crucial in $D^{(2)}$) we obtain the required property 
(1a) for $C$.

Subcase 2. There exists some constraint $C'$ in $\mathcal I$ such that 
${C'}^{(2)}$ is empty. Since the weakening of $C$ gives a solution in $D^{(2)}$, 
$C'$ must be $C$. 
By Lemma \ref{LEMParalPropertyFromCrucialInMultiType}
$C$ has the parallelogram property, which is the property (1a).

Let us prove that (1b) or (1c) holds.
Recall that we have $E<_{T(\sigma)}^{D_{z}} D^{(1)}_{z}$ for some $z$ and $T\in\{\TPC,\TL\}$.
Let $\mathcal B = \{B\mid B<_{T(\sigma)}^{D_{z}} D^{(1)}_{z}\}$.
For every $B\in \mathcal B$ we do the following. 
Let us consider the reduction $D^{(B,\top)}$ such that 
$D_{x}^{(B,\top)}=D_{x}^{(1)}$ if $x\neq z$ and $D_{z}^{(B,\top)}=B$.
Let $D^{(B,\bot)}$ be the maximal 1-consistent (probably empty) reduction for $\mathcal I$ 
such that $D^{(B,\bot)}\le D^{(B,\top)}$. 
By Lemma \ref{LEMExistenceOfTreeCoverings},
for every variable $x$ and $B\in\mathcal B$ there exists a tree-covering $\Upsilon_{B,x}$
such that $\Upsilon_{B,x}^{(B,\top)}(x)$ defines $D_{x}^{(B,\bot)}$.
By $\Upsilon_{x} = \bigwedge_{B\in\mathcal B} \Upsilon_{B,x}$ 
we define one universal tree-covering, 
that is, 
$\Upsilon_{x}^{(B,\top)}(x)$ defines $D_{x}^{(B,\bot)}$
for every $B\in\mathcal B$.
We extend this definition to variables from an expanded covering of $\mathcal I$.
Precisely, 
for a variable $x'$ that is a child of $x$ 
by $\Upsilon_{x'}$ we denote $\Upsilon_{x}$ whose variable $x$ is replaced by $x'$. 
Let $\mathcal B_{0}$ be the set of all $B\in\mathcal B$ such that 
$D_{x}^{(B,\bot)}$ is not empty.
Let us consider two cases:

Case 1. $\mathcal B_{0}$ is empty. 
Consider a tree-covering $\Upsilon$ such that 
$\Upsilon^{(B,\top)}$ has no solutions for every $B\in \mathcal B$.
Since $\Upsilon$ is a tree-covering, 
its solution set is subdirect. 
Notice that $T$ cannot be the $\TPC$ type, because 
Lemma \ref{LEMFindOneConsistentForAll}
guarantees the existence of a nonempty reduction
$D_{x}^{(B,\bot)}$ for every $B\in\mathcal B$.
Hence $T=\TL$.

Then weaken $\Upsilon$ while we can keep the property that 
$\Upsilon^{(B,\top)}$ has no solutions for every $B\in \mathcal B$.
The obtained instance we denote by $\Upsilon'$.
Since $\mathcal I$ is crucial in $D^{(1)}$, 
$\Upsilon'$ must contain every constraint relation that appeared in $\mathcal I$.
Let us prove that $\mathcal I$ is connected.
Take two constraints $C_{1}$ and $C_{2}$ of $\mathcal I$ having a common variable $x$.
Applying Lemma \ref{LEMGetABridgeFromSubdirectPCLinearInstance} to $\Upsilon'$ 
we obtain a bridge from $\ConOne(C_1,x)$ to $\sigma$
and a bridge from $\sigma$ to $\ConOne(C_2,x)$.
Composing these bridges we obtain a bridge 
from $\ConOne(C_1,x)$ to $\ConOne(C_2,x)$.
This bridge is reflexive because $\Upsilon$ 
a tree-covering and the path from 
a child of $x$ to a child of $z$ and back is just a path in the cycle-consistent instance $\mathcal I$.
Additionally, we derived from Lemma \ref{LEMGetABridgeFromSubdirectPCLinearInstance}
that all the congruences of $\Congruences(\mathcal I)$ are of 
the linear type.
Thus, we proved that any two constraints of $\mathcal I$ with a common variable are adjacent,
which means that $\mathcal I$ is connected and 
satisfies (1b) if its solution set is subdirect or  (1c)
otherwise.

Case 2. $\mathcal B_{0}$ is not empty.
For every expanded covering $\mathcal J$ of $\mathcal I$  
by $\Sol(\mathcal J)$ we denote the set of all $B\in \mathcal B_{0}$ such that 
$\mathcal J^{(B,\bot)}$ has a solution.

We want to find a set of instances $\Omega$ 
 satisfying the following conditions:
\begin{enumerate}
    \item Every instance in $\Omega$ is a weakening of $\mathcal I$.
    \item $\bigcap_{\mathcal J\in\Omega}\Sol(\mathcal J)=\varnothing$. 
    \item If we replace any instance in $\Omega$ 
    by all weaker instances then 2 is not satisfied.
    \item For every $\mathcal J\in \Omega$ there exists $B\in \mathcal B_{0}$ 
    such that 
    \begin{enumerate} 
    \item $\mathcal J$ is crucial in $D^{(B,\bot)}$, and 
   \item $B\in \Sol(\mathcal J')$ for every $J'\in \Omega\setminus\{\mathcal J\}$.
      \end{enumerate}
\end{enumerate}

We start with $\Omega = \{\mathcal I\}$.
It already satisfies conditions 1 and 2.
If 3 is not satisfied then we replace the corresponding instance by all weaker instances and get a new $\Omega$.
We cannot weaken forever, that is why at some moment 
conditions
1-3 will be satisfied. Let us show that it also satisfies 
condition 4.
Take some $\mathcal J\in\Omega$.
For every constraint $C$ in $\mathcal J$ by $\mathcal J_{C}$ we denote the instance 
obtained from $\mathcal J$ by weakening $C$.
By condition 3,
$\Omega\cup\{\mathcal J_{C}\mid C\in \mathcal J\}\setminus\{\mathcal J\}$
cannot satisfy condition 2, which means that there exists $B\in\mathcal B_{0}$ such that 
$\mathcal J$ is crucial in $D^{(B,\bot)}$ 
and $B\in\Sol(\mathcal J')$ for every $\mathcal J'\in\Omega\setminus \{J\}$.
Thus, we have $\Omega$ satisfying conditions 1-4.

For an expanded covering 
$\mathcal J$ of $\mathcal I$ 
by $\bot(\mathcal J)$ we denote the instance 
$\mathcal J\wedge \bigwedge_{x\in\Var(\mathcal J)} \Upsilon_{x}$, where 
we rename the variables so that the only common variable of 
$\Upsilon_{x}$ and $\mathcal J$ is $x$.
Also, by $\Delta(\mathcal J)$ we denote the instance that  
is obtained from $\mathcal J$ by adding the constraints 
$\sigma(z',z'')$ for every pair of variables  
whose parent is $z$.


For any weakening $\mathcal J$ of $\mathcal I$ and 
any $B\in\mathcal B_{0}$ such that 
$\mathcal J$ is crucial in $D^{(B,\bot)}$ 
we can apply the inductive assumption,
which proves that either 
$\mathcal J$ satisfies (1b) or
$\mathcal J$ satisfies (1c).
Let us consider two subcases.

Subcase 1. Some instance $\mathcal I'\in \Omega$
does not satisfy (1b).
Let 
$\Omega'$ be the set of all instances that are weaker than $\mathcal I'$ joined with the instances 
from $\Omega\setminus\{\mathcal I'\}$.
Put 
$\mathcal B_{1} = 
\bigcap_{\mathcal J\in \Omega'}\Sol(\mathcal J)$.
Notice that condition 3 implies that 
$\mathcal B_{1}$ is not empty.
It follows from the definition that 
$\mathcal I'$ is crucial in 
$D^{(B,\bot)}$ for every $B\in \mathcal B_{1}$.
We want to build a sequence $\mathcal J_1,\dots,\mathcal J_{s}$ of expanded coverings 
of $\mathcal I'$ such that
$\mathcal B_{1}\cap\bigcap_{i\in[s]} \Sol(\mathcal J_{i}) = \varnothing$, 
some $B$ belongs to $\mathcal B_{1}\cap \bigcap_{i\in[s-1]} \Sol(\mathcal J_{i})$,
$\mathcal J_{s}$ is crucial in $D^{(B,\bot)}$ and has a connected subinstance whose solution set is not subdirect.
Take some $B\in \mathcal B_{1}$
and apply the inductive assumption 
to $\mathcal I'$ and $D^{(B,\bot)}$.
Since $\mathcal I'$ does not satisfy (1b) 
there exists an expanded covering  
$\mathcal J_1$ of  $\mathcal I'$
such that $\mathcal J_{1}$ is crucial in $D^{(B,\bot)}$ and 
$\mathcal J_{1}$ has a connected subinstance whose solution set is not subdirect.
If $\Sol(\mathcal J_{1})\cap \mathcal B_{1} = \varnothing$, then we are done.
Otherwise, put 
$\mathcal B_{2} = \Sol(\mathcal J_{1})\cap \mathcal B_{1}$, choose some $B\in \mathcal  B_{2}$ and 
apply the inductive assumption to 
$\mathcal I'$ and $D^{(B,\bot)}$ to obtain $\mathcal J_{2}$. Since $\mathcal B_{1}$ is finite, 
and the sequence $\mathcal B_{1},\mathcal B_{2},\dots$
is decreasing, at some moment 
the required condition 
$\mathcal B_{1}\cap\bigcap_{i\in[s]} \Sol(\mathcal J_{s}) = \varnothing$, 
will be satisfied.

Put 
$\Theta = \Delta((\bigwedge_{\mathcal J\in \Omega'}\bot(\mathcal J))\wedge (\bigwedge_{i=1}^{s}\bot(\mathcal J_{s})))$.
It follows from the definition that 
$\Theta$ is an expanded covering of $\mathcal I$
not having a solution in $D^{(1)}$.
Let $\Theta'$ be the weakening of $\Theta$ such that 
$\Theta'$ is crucial in $D^{(1)}$.
Notice that all the constraints of $\Theta$ that came from $\mathcal J_{s}$ 
are crucial in $D^{(1)}$, which means that they stay in 
$\Theta'$. Therefore, $\Theta'$ has a connected subinstance whose solution set is not subdirect.
Thus, we proved that $\mathcal I$ satisfies (1c).

Subcase 2. Every instance $\mathcal J\in\Omega$ satisfies (1b).
This implies that each instance 
$\bot(\mathcal J)$ has a subdirect solution set.
Notice that if $\Omega = \{\mathcal I\}$
then $\mathcal I$ satisfies (1b), which completes this case.
Otherwise, both $\mathcal J^{(1)}$ and $(\bot(\mathcal J))^{(1)}$ have a solution for every $\mathcal J\in\Omega$.
Since 
$\bot(\mathcal J)$ is crucial in $D^{(B,\top)}$ 
for some $B\in\mathcal B_{0}$ (property 4(a) of $\Omega$),
Lemma \ref{LEMGetABridgeFromSubdirectPCLinearInstance} implies that 
$T=\TL$.
Let us show that $\mathcal I$ is connected.
Take some constraint $C\in\mathcal I$ and a variable $x$ of $C$.
Put $\Theta = \Delta(\bigwedge_{\mathcal J\in \Omega}\bot(\mathcal J))$.
Let us weaken $\Theta$ to make all the  constraints except for the constraints $\sigma(z',z'')$
crucial in $D^{(1)}$. We do not weaken the constrains $\sigma(z',z'')$. 
The obtained instance we denote by $\Theta'$.
Condition 4 for $\Omega$ implies that all the constraints of $\Theta$ 
coming from some $\mathcal J\in\Omega$ are still in $\Theta'$.
Since $\mathcal I$ is crucial in $D^{(1)}$, a child of $C$ must appear in $\Theta'$.
Let the child appear in $\bot(\mathcal J)$-part 
of $\Theta'$.
Denote the $\bot(\mathcal J)$-part of 
$\Theta'$ by $\Theta_{\mathcal J}'$. 
By the cruciality of $C$, there exists $B\in\mathcal B_{0}$ such that 
$\Theta_{\mathcal J}'^{(B,\top)}$ has no solutions
but a weakening of the child of $C$ gives a solution
inside $D^{(B,\top)}$.
Since $\bot(\mathcal J)$ has a solution in $D^{(1)}$,  $\Theta_{\mathcal J}'$ also has a solution 
in $D^{(1)}$.
Let $\mathcal M$ be a minimal set of children of $z$ we need to restrict to $B$
in $\Theta_{\mathcal J}'$ to kill all the solutions of $\Theta_{\mathcal J}'^{(1)}$.
We will build a bridge 
$\delta$
from $\sigma$ to $\ConOne(C,x)$ 
such that $\widetilde{\delta}$ is larger than a binary relation coming from
some path from $z$ to $x$ in $\mathcal I$.
Consider two subsubcases.

Subsubcase 1. The child of $C$ appears in a child of $\Upsilon_{y}$.
Since $\Upsilon_{y}$ is a tree-covering and $\mathcal I^{(1)}$ is 1-consistent, 
the set $\mathcal M$ contains at least one variable from the child of $\Upsilon_{y}$.
Applying Lemma \ref{LEMGetABridgeFromSubdirectPCLinearInstance} 
to the solution set of $\Theta_{\mathcal J}'$, we 
get a required bridge from $\sigma$ to $\ConOne(C,x)$ and also prove that 
$\ConOne(C,x)$ is a linear congruence.

Subsubcase 2. The child of $C$ appears in $\mathcal J\in\Omega$.
Let $y$ be the variable of $\Theta_{\mathcal J}'$ 
such that some variable from $\mathcal M$ appears in 
a child of $\Upsilon_{y}$ in $\Theta_{\mathcal J}'$.
Applying Lemma \ref{LEMGetABridgeFromSubdirectPCLinearInstance}
to the solution set of $\Theta_{\mathcal J}'$,
we get a bridge from $\sigma$ to $\xi$ for some $\xi\in\Congruences(\mathcal J,y)$.
Since $\mathcal J$ is connected,
by Lemma \ref{LEMConnectedProperties}(b) there is a corresponding bridge from $\xi$ to $\ConOne(C,x)$.
Composing these bridges we get a required bridge $\delta$  from $\sigma$ to $\ConOne(C,x)$.
Notice that we could also build a bridge $\delta$  from $\sigma$ to $\ConOne(C,x)$
without an intermediate step but $\widetilde \delta$ would not satisfy the required 
property.

Thus, for every constraint $C$ and every variable $x$ in $\mathcal I$ 
we have a bridge 
$\delta$
from $\sigma$ to $\ConOne(C,x)$ 
such that $\widetilde{\delta}$ is larger than a binary relation coming from
some path from $z$ to $x$ in $\mathcal I$.
To prove that $\mathcal I$ is connected we do the following.
We take 
two constraints $C_{1}$ and $C_{2}$ with a common variables $x$.
We proved that there is a bridge from
$\ConOne(C_1,x)$ to $\sigma$, 
and a bridge from $\ConOne(C_{2},x)$ to $\sigma$.
Composing these bridges (and using cycle-consistency of $\mathcal I$) 
we obtain a required 
reflexive bridge. Hence $C_{1}$ and $C_{2}$ are adjacent, $\mathcal I$ is connected, 
and $\mathcal I$ satisfies (1b) if it has subdirect solution set, or (1c)
otherwise.
\end{proof}

\begin{thm}\label{THMPCDoesnotKillAllSolutions}
Suppose $\mathcal I$ is a cycle-consistent irreducible instance, 
$B<_{PC(\sigma)}^{D_{y}} D_{y}$ for some $y\in\Var(\mathcal I)$,
$\mathcal I$ has a solution.
Then $\mathcal I$ has a solution 
with $y\in B$.
\end{thm}

\begin{proof}

For any 
$G<_{PC(\sigma)}^{D_{y}} D_{y}$
by $D^{(G,\top)}$ we denote the reduction of $\mathcal I$ 
such that 
$D_{z}^{(G,\top)}= G$ and
$D_{x}^{(G,\top)}= D_{x}$ if $x\neq y$.
For an expanded covering $\mathcal J$ of $\mathcal I$ 
by $\Sol(\mathcal J)$ we denote the set of
all $G\in D_{y}/\sigma$ such that 
$\mathcal J^{(G,\top)}$ has a solution.

Assume that $\mathcal I$ has no solutions with $y\in B$. Let $\mathcal B\subsetneq D_{y}/\sigma$ be an inclusion-maximal 
set such that 
$\Sol(\mathcal J) = \mathcal B$ 
for some expanded covering $\mathcal J$ of $\mathcal I$.
Let $\mathcal J$ be the expanded covering witnessing this.
Choose $G\in (D_{y}/\sigma)\setminus \mathcal B$.

By Lemma
\ref{LEMFindOneConsistentForAll}
there exists a 1-consistent reduction 
for $\mathcal I$ smaller than 
$D^{(G,\top)}$. 
Since $\mathcal J$ is an expanded covering, 
the maximal 1-consistent reduction $D^{(G,\bot)}$
for $\mathcal J$
such that 
$D^{(G,\bot)}\le D^{(G,\top)}$
is also nonempty.
By Lemma \ref{LEMExistenceOfTreeCoverings}
for every $x\in \Var(\mathcal J)$ 
there exists a tree-covering $\Upsilon_{x}$ of $\mathcal J$
such that $\Upsilon_{x}^{(G,\top)}(x)$ defines 
$D_{x}^{(G,\bot)}$. Notice that 
the reduction $D^{(G,\top)}$ was defined for $\mathcal I$ and then 
extended to $\mathcal J$ but 
$D^{(G,\bot)}$ was originally defined for $\mathcal J$ and does not exist for $\mathcal I$.

Weaken $\mathcal J$ to make it crucial in 
$D^{(G,\bot)}$ and denote the obtained instance 
by $\mathcal J'$.
By Theorem \ref{THMMainInductiveCSPClaim}
applied to $\mathcal J'$ and $D^{(G,\bot)}$, 
$\mathcal J'$ satisfies (1b) or (1c).

Assume that $\mathcal J'$ satisfies (1b).
Then the solution set of $\mathcal J'$ is subdirect.
Put $\mathcal J''= \mathcal J'\wedge \bigwedge_{x\in\Var(\mathcal J')} \Upsilon_{x}$.
Notice that $\mathcal J''$ is an expanded covering of 
$\mathcal I$ with a subdirect solution set.
Since $\mathcal J''^{(G,\top)}$ has no solutions, $\mathcal J''$ has a solution (as $\mathcal I$ has a solution),
and any weakening of a constraint from $\mathcal J'$ inside $\mathcal J''$ 
gives an instance with a solution in $D^{(G,\top)}$, Lemma
\ref{LEMGetABridgeFromSubdirectPCLinearInstance} implies that 
the type $\TPC$ coincides with the type of the crucial constraints, which is linear by (1b). This contradiction completes this case.

Assume that $\mathcal J'$ satisfies (1c).
Let $\Theta$ be the expanded covering of $\mathcal J'$ 
that is crucial in $D^{(G,\bot)}$ and 
$\Upsilon$ be the linked connected subinstance of $\Theta$. 
Put $\Theta'= \Theta\wedge \bigwedge_{x\in\Var(\Theta)} \Upsilon_{x}$.
Notice that $\Theta'$ is an expanded covering of 
$\mathcal I$ with a subdirect solution set and 
$\Theta'$ has no solutions in 
$D^{(G,\top)}$.
Let $x$ be a variable of a constraint $C\in\Upsilon$.
By Lemma~\ref{LEMConnectedProperties}(p),
$\ConOne(C,x)$ is a perfect linear congruence and 
there exists 
$\zeta\le \mathbf D_{x}\times \mathbf D_{x}\times \mathbf Z_{p}$
such that
$(y_{1},y_{2},0)\in \zeta\Leftrightarrow (y_{1},y_{2})\in\ConOne(C,x)$
and $\proj_{1,2}(\zeta) = \cover{\ConOne(C,x)}$.
Let us replace the variable $x$ of $C$ in $\Theta'$ by $x'$
and add the constraint $\zeta(x,x',z)$.
The obtained instance we denote by $\Theta''$.
Another instance we build from $\Theta''$ by 
replacing  
$\zeta(x,x',z)$ by $\cover{\omega}(x,x')$, where $\omega = \ConOne(C,x)$.
We denote it by $\Theta'''$.
Since 
$\Theta'''$ is an expanded covering of $\mathcal J$
we have $\Sol(\Theta''')\supseteq \Sol(\mathcal J)$.
Since the weakening of $C$ in $\Theta$ gives 
an instance with a solution in $D^{(G,\bot)}$ 
and
$\mathcal B$ was chosen maximal, 
we have
$\Sol(\Theta''') = D_{y}/\sigma$. 
Let $R\le D_{y}/\sigma\times \mathbf Z_{p}$ 
be the set of all pairs $(F,j)$ such that 
$\Theta''$ has a solution 
in $D^{(F,\top)}$ with $z=j$.
We know that 
$R$ is subdirect, 
$(G,0)\notin R$, 
$(G,j)\in R$ for some $j\in Z_{p}$.
Applying Corollary \ref{CORMainStableIntersection} to 
$R$, $\{G\}<_{\TPC}D_{y}/\sigma$, 
and
$\{0\}<_{\TL} \mathbf Z_{p}$,
we get a contradiction as we mixed linear and PC types.
\end{proof}

\subsection{Statements sufficient to prove that Zhuk's algorithm works}\label{SectionCSPMainClaims}


\begin{thm}\label{THMCSPDReductionsAreSafe}
Suppose $\Theta$ is a cycle-consistent irreducible CSP instance, and 
$B<_{T}^{D_{x}}D_{x}$, where $T\in\{BA,C,PC\}$. 
Then $\Theta$ has a solution if and only if
$\Theta$ has a solution with $x\in B$.
\end{thm}

\begin{proof}
For $T=PC$ it follows from Theorem 
\ref{THMPCDoesnotKillAllSolutions}. 
Assume that $T\in\{BA,C\}$.
By Lemma \ref{LEMFindOneConsistentForAll} 
there exists a 1-consistent reduction 
$D^{(1)}\le_{T} D$ such that $D_{x}^{(1)}\le B$.
By Theorem \ref{THMMainInductiveCSPClaim}
$\Theta^{(1)}$ has a solution, and therefore  
$\Theta$ has a solution with $x\in B$.
\end{proof}

\begin{thm}\label{THMCodimensionOneTheorem}
Suppose the following conditions hold:
\begin{enumerate}
\item $\mathcal I$ is a linked cycle-consistent irreducible CSP instance with $\Var(\mathcal I)=\{x_1,\dots,x_n\}$;
\item $D_{x_{i}}$ is S-free 
for every $i\in[n]$;
\item if we weaken all the constraints of $\Theta$, 
we get an instance whose solution set is subdirect.
\item $\sigma_{x_i}$ is the intersection of all the linear congruences $\sigma$ on $D_{x_{i}}$ such that 
$\sigma^{*}=D_{x_{i}}\times D_{x_{i}}$.
\item $L_{x_i} = D_{x_{i}}/\sigma_{x_i}$ for every $i\in[n]$;
\item $\phi:\mathbf Z_{q_{1}}\times \dots \times \mathbf Z_{q_{k}}
\to L_{x_1}\times\dots\times L_{x_n}$ is a homomorphism,
where $q_{1},\dots,q_{k}$ are prime numbers;
\item if we weaken any constraint of 
$\mathcal I$ then for every $(a_{1},\ldots,a_{k})\in \mathbf Z_{q_{1}}\times \dots \times \mathbf Z_{q_{k}}$ 
there exists a solution of the obtained instance in 
$\phi(a_{1},\ldots,a_{k})$.
\end{enumerate}
Then 
$\{(a_{1},\dots,a_{k})\mid \Theta \text{ has a solution in }\phi(a_1,\dots,a_{k})\}$ is
either empty, or is full, or is an affine subspace of $\mathbf Z_{q_{1}}\times \dots \times \mathbf Z_{q_{k}}$ of codimension 1 (the solution set of a single linear equation).
\end{thm}
\begin{proof}
Put 
$\Delta=\{(a_{1},\dots,a_{k})\mid \Theta \text{ has a solution in }\phi(a_1,\dots,a_{k})\}$.
If $\Delta$ is full then we are done.
Otherwise, consider 
$(b_1,\ldots,b_{k})\in (\mathbf Z_{q_{1}}\times \dots \times \mathbf Z_{q_{k}})\setminus \Delta$.
Notice that $\phi(b_{1},\ldots,b_{k})$
can be viewed as a reduction for $\mathcal I$.
We denote this reduction by $D^{(1)}$.
It follows from condition (7) that 
$\mathcal I$ is crucial in $D^{(1)}$.

Let us prove that there exists a constraint $C\in\mathcal I$ and 
its variable $x$ such that 
$\ConOne(C,x)$ is a perfect linear congruence.
By Lemma \ref{LEMMinimalPCLinearReductionIsConsistent} 
either $C_{0}^{(1)}$ is empty for some $C_{0}\in\mathcal I$, 
or the reduction $D^{(1)}$ is 1-consistent for $\mathcal I$.
Consider two cases.

Case 1. $C_{0}^{(1)}$ is empty. 
Since $\mathcal I$ is crucial in $D^{(1)}$, it 
consists of just one constraint $C_{0}$.
Let $C_0 = R(y_1,\dots,y_t)$.
By Lemma \ref{LEMParalPropertyFromCrucialInMultiType} 
$R$ has the parallelogram property
and 
$\ConOne(R,1)$ is a linear congruence 
such that $\ConOne(R,1)^{*} = D_{y_{1}}^{2}$.
By Lemma \ref{LEMLInearOnTheTopIsEasy}
$\mathbf D_{y_{1}}/\delta \cong \mathbf Z_{p}$.
Let $\psi\colon \mathbf D_{y_{1}}\to \mathbf Z_{p}$ be the homomorphism.
Then the required ternary relation 
$\zeta\le \mathbf D_{y_{1}}\times \mathbf D_{y_{1}} \times \mathbf Z_{p}$ can be defined by 
$\{(a_1,a_2,b)\mid \psi(a_1)-\psi(a_2)=b\}$.
Hence $\ConOne(R,1)$ is a perfect linear congruence.

Case 2. The reduction $D^{(1)}$ is 1-consistent.
By Theorem~\ref{THMMainInductiveCSPClaim}, every constraint of $\mathcal I$ has the parallelogram property
and satisfies condition (1b) or (1c).
If $\mathcal I$ satisfies (1c) then there exists an instance $\Theta\in\Expanded(\mathcal I)$
that is crucial in $D^{(1)}$ and contains a linked connected subinstance $\Upsilon$ such that
the solution set of $\Upsilon$ is not subdirect.
By condition 4, since the solution set of $\Upsilon$ is not subdirect, $\Upsilon$ must contain a constraint relation from
the original instance $\mathcal I$.
Applying Lemma \ref{LEMConnectedProperties}(p), 
we derive that $\ConOne(C,x)$ is a perfect linear congruence
for the corresponding child of the original constraint and its variable.
If $\mathcal I$ satisfies (1b), then $\mathcal I$ is linked connected itself
and the existence of a perfect linear congruence 
again follows from 
Lemma \ref{LEMConnectedProperties}(p).

Thus, 
$\ConOne(C,x)$ is a perfect linear congruence 
for some $C\in\mathcal I$ and 
its variable $x$. Let 
$\zeta$ be the corresponding ternary relation.
We add a new variable $z$ with domain $\mathbf Z_{p}$, replace
the variable $x$ in $C$ by $x'$, 
and add the constraint 
$\zeta(x,x',z)$.
We denote the obtained instance by $\mathcal I'$.
Let
$L$ be the set of all tuples $(a_{1},\ldots,a_{k},b)\in
\mathbf Z_{q_{1}}\times \dots \times \mathbf Z_{q_{k}} \times \mathbf Z_{p}$
such that $\mathcal I'$ has a solution with $z=b$ in $\phi(a_{1},\ldots,a_{k})$.
Notice that 
$L\le \mathbf Z_{q_{1}}\times \dots \times \mathbf Z_{q_{k}} \times \mathbf Z_{p}$.
By condition 7, the projection of $L$ onto
the first $k$ coordinates is a full relation
and $(b_{1},\dots,b_{k},0)\notin L$.
Therefore $L$ has dimension $k$ and can be 
defined by one linear equation.
If this equation is $z = b$ for some $b\neq 0$, then
$\Delta$ is empty.
Otherwise, we put $z=0$ in this equation and get
an equation describing all $(a_{1},\ldots,a_{k})$ such that
$\mathcal I$ has a solution in $\phi(a_{1},\ldots,a_{k})$.
Hence the dimension of $\Delta$ is $k-1$.
\end{proof}

\section{XY-symmetric operations}
\label{SECTIONXYSYMMETRIC}
In this section we prove that 
a weak near unanimity operation of an odd arity 
implies an operation that is symmetric on all tuples having exactly two different elements.
The idea of the proof is to generate a relation 
such that the existence of an XY-symmetric operation was equivalent to existence of 
a block-constant tuple. We gradually reduce coordinates of this relation to strong subalgebras trying to achieve this tuple.
If we cannot make the next reduction, it means 
that we found a linear congruence such that 
there is no block-constant 
tuple even modulo this congruence.
Since this linear congruence lives on the original domain, and the generated relation 
must be linked, we immediately obtain a perfect linear congruence. This allows us 
to represent the domain as a product of a 
smaller domain $B$ (where we have an XY-symmetric operation by the inductive assumption) and $\mathbf Z_{p}$.
The rest of the proof is purely operational:
we start with an XY-symmetric operation on $B$ 
and show how composing this operation with itself we can gradually increase the number of tuples where it behaves well on the whole domain. 

This section is organised as follows.
First, we explain how we define the relation 
for a tuple of algebras, how we apply and denote reductions. We also define symmetries this relation has and $\boxtimes$-product of $\mathbf B$ and $\mathbf Z_{p}$.
In Subsection \ref{SUBSECTIONXYSymProofOfTheMainResults} 
we show how to derive the main result
from three theorems that are proved later.
In the next section we prove 
two out of three theorems explaining how to build a smaller reduction if possible, 
and how to build a reduction if 
it is known that an XY-symmetric operation exists.
Finally, in Subsection \ref{SUBSECTIONFixingAnOperation} we show 
how to improve an operation gradually to make it 
XY-symmetric on $\mathbf B\boxtimes\mathbf Z_{p}$
even if originally it was XY-symmetric only on $B$.

\subsection{Definitions}

\textbf{The free generated relation $R_{\mathbf A_1,\dots,\mathbf A_{s}}$.}
For a tuple of algebras 
$\mathbf A_1,\dots,\mathbf A_{s}\in\mathcal V_{n}$
by 
$R_{\mathbf A_1,\dots,\mathbf A_{s}}$
we denote the relation 
of arity 
$N:=(2^{n-1}-1)\cdot \sum\limits_{i=1}^{s}|\mathbf A_{i}|\cdot (|\mathbf A_{i}|-1)$
defined as follows.
Coordinates of the relation 
are indexed by 
$(\mathbf A_{i},\alpha)$, where $\alpha\in\{a,b\}^{n}$ for some $a,b\in A_{i}$, $a\neq b$.
The set of all indexes denote by $I$. 
For a set of tuples $S$ by 
$\TwoTuples(S)$ we denote the set of tuples from $S$ having exactly 2 different elements.
Then $I =\{(\mathbf A_{i},\alpha)\mid i\in[s], \alpha\in\TwoTuples(A^{n})\}$.
For $i\in[n]$
by $\gamma_{i}$ we denote the tuple of length $N$ 
whose $(\mathbf A_{i},\alpha)$-th element is equal to 
$\alpha(i)$ for every $(\mathbf A_{i},\alpha)\in I$.
Then $R_{\mathbf A_1,\dots,\mathbf A_{n}}$
is the minimal subuniverse 
of $\prod\limits_{i\in[s]} \mathbf A_{i}^{(2^{n-1}-1)|A_{i}|\cdot |A_{i}-1|}$
containing $\gamma_1,\dots,\gamma_n$.
We also say that $R_{\mathbf A_1,\dots,\mathbf A_{n}}$ 
is the subalgebra generated by 
$\gamma_1,\dots,\gamma_n$, 
and the tuples $\gamma_1,\dots,\gamma_n$ 
are called \emph{the generators} of $R_{\mathbf A_1,\dots,\mathbf A_{n}}$.

We will use terminology similar to the one we used in the previous section.
For every $(\mathbf A_{i},\alpha)\in I$ by 
$\mathbf D_{(\mathbf A_{i},\alpha)}^{(0)}$ we denote the subalgebra of $\mathbf A_{i}$ generated
by  elements of $\alpha$.
Notice that $\proj_{i}(R_{\mathbf A_1,\dots,\mathbf A_{s}})=D_{i}^{(0)}$
for every $i\in I$.
By $\mathcal R_{\mathbf A_1,\dots,\mathbf A_{s}}$ we denote the set of 
all relations $R$ of arity $N$ whose coordinates are indexed by $I$ such that the domain of the $i$-th coordinate of $R$ is 
$D_{i}^{(0)}$ for every $i\in I$.

\textbf{Reductions.}
In our proof we reduce the relation $R_{\mathbf A_1,\dots,\mathbf A_{s}}$ 
by reducing their coordinates.
\emph{A reduction} $D^{(\top)}$ for $R\in\mathcal R_{\mathbf A_1,\dots,\mathbf A_{s}}$ 
is a mapping that assigns a subuniverse 
$D_{i}^{(\top)}\le D_{i}^{(0)}$ to every $i\in I$.
$D^{(0)}$ can be viewed as a trivial reduction.
As in the previous section we write 
$D^{(\bot)}\lll D^{(\top)}$ and $D^{(\bot)}\le_{T} D^{(\top)}$ 
whenever
$D_{i}^{(\bot)}\lll D_{i}^{(\top)}$ for every $i\in I$ and 
$D_{i}^{(\bot)}\le_{T} D_{i}^{(\top)}$ for every $i\in I$, respectively.
Notice that any reduction $D^{(\top)}$ can be viewed as a relation 
from $\mathcal R_{\mathbf A_1,\dots,\mathbf A_{s}}$.
Then for any $R\in \mathcal R_{\mathbf A_1,\dots,\mathbf A_{s}}$  
and a reduction $D^{(\bot)}$ 
by $R^{(\bot)}$ we denote 
$R\cap D^{(\bot)}$.
A reduction $D^{(\bot)}$ is called \emph{1-consistent} for $R\in \mathcal R_{\mathbf A_1,\dots,\mathbf A_{s}}$
if 
$\proj_{i}(R^{(\bot)}) = D_{i}^{(\bot)}$ for every $i\in I$.

\textbf{$k$-WNU.} An operation $f:A^{n}\to A$ is called 
\emph{a $k$-WNU operation} 
if it is symmetric on 
$(\underbrace{x,x,\dots,x}_{k},y,y,\dots,y)$.
Then, 1-WNU is just a usual WNU.







\textbf{Permutations and symmetries.}
For a tuple $\alpha\in A^{n}$ by $\Perm(\alpha)$ we denote the set of 
all tuples that can be obtained from $\alpha$ by a permutation of elements.
For an index 
$i = (\mathbf A_{j},\alpha)$ by $\Perm(i)$ we denote the set of 
indexes $(\mathbf A_{j},\beta)$ with $\beta\in \Perm(\alpha)$.
For a tuple $\alpha\in A^{n}$ 
and a permutation on $[n]$ by 
$\sigma(\alpha)$ we denote the tuple 
$\alpha'$ such that 
$\alpha'(j) = \alpha(\sigma(j))$ for every $j\in[n]$.
For a tuple $\gamma$ of arity $N$ whose coordinates 
are indexed by elements from $I$ 
and a permutation $\sigma$ on $[n]$ 
by $\gamma^{\sigma}$ we denote the
tuple $\gamma'$ such that 
$\gamma'((\mathbf A_{i},\alpha)) = 
\gamma((\mathbf A_{i},\sigma(\alpha)))$
for any $(\mathbf A_{i},\alpha)\in I$. 
Similarly, for a relation $R\in\mathcal R_{\mathbf A_1,\dots,\mathbf A_{s}}$
put $R^{\sigma}=\{\gamma^{\sigma}\mid \gamma\in R\}$. 
A relation $R$ is called \emph{$\sigma$-symmetric} if 
$R^{\sigma} = R$.
A relation $R$ is called 
\emph{symmetric} if it is $\sigma$ symmetric for 
every permutation $\sigma$ on $[n]$.
Similarly, a 
reduction $D^{(\top)}$ is called \emph{symmetric} if 
$D_{i}^{(\top)}=
D_{j}^{(\top)}$
for any $j\in\Perm(i)$.

\textbf{$\boxtimes$-product of $\mathbf B$ and $\mathbf Z_{p}$.} 
For $x = (a,b)$
by $x^{(1)}$ and  $x^{(2)}$ we denote 
$a$ and $b$ respectively.
For an algebra $\mathbf B=(B;w^{\mathbf B})$ by 
$\mathbf B\boxtimes \mathbf Z_{p}$ we denote 
the set of algebras $\mathbf A$ such that 
$A = B\times  Z_{p}$, 
$(w^{\mathbf A}(x_1,\dots,x_n))^{(1)} =
w^{\mathbf B}(x_1^{(1)},\dots,x_n^{(1)})$
and 
$(w^{\mathbf A}(x_1,\dots,x_n))^{(2)} = 
f(x_1^{(1)},\dots,x_{n}^{(1)})
+a_1 x_1^{(2)}+\dots+a_n x_{n}^{(2)}$
for some mapping
$f\colon B^{n}\to \mathbf Z_{p}$
and 
$a_{1},\dots,a_n\in \mathbf Z_{p}$.

\subsection{Proof of the main result}\label{SUBSECTIONXYSymProofOfTheMainResults}



\begin{thm}\label{THMExistenceOfReduction}
Suppose $\mathbf A_{1},\dots,\mathbf A_{s}\in\mathcal V_{n}$, $n$ is odd, $D^{(1)}$ is a 1-consistent symmetric reduction of $R_{\mathbf A_{1},\dots,\mathbf A_{s}}$, 
$D^{(1)}\lll D^{(0)}$. 
Then one of the following conditions hold
\begin{enumerate}
    \item $|D^{(1)}_{(\mathbf A_{i},\alpha)}|=1$ 
    for all $(\mathbf A_{i},\alpha)$.
    \item there exists a 1-consistent symmetric reduction $D^{(2)}$ for $R_{\mathbf A_{1},\dots,\mathbf A_{s}}$ such that 
    $D^{(2)}\lll D^{(1)}$ and $D^{(2)}\neq D^{(1)}$.
    \item there exists a perfect linear congruence $\sigma$ on some
    $D^{(0)}_{(\mathbf A_{i},\alpha)}$ such that 
    \begin{enumerate}
        \item $D_{(\mathbf A_{i},\alpha)}^{(1)}\times D_{(\mathbf A_{i},\alpha)}^{(1)}\not\subseteq \sigma$ 
        \item $D_{(\mathbf A_{i},\alpha)}^{(1)}\times D_{(\mathbf A_{i},\alpha)}^{(1)}\subseteq \sigma^{*}$
    \end{enumerate}
\end{enumerate}
\end{thm}

\begin{thm}\label{THMBuildAReductionFromXYSymmetric}
Suppose 
$\mathbf A_1,\dots,\mathbf A_s\in\mathcal V_{n}$, 
$n$ is odd, 
there exists an $n$-ary term $\tau_{0}$ such that 
$\tau_0^{\mathbf A_{i}}$ is 
XY-symmetric for every $i$. 
Then there exists a 
1-consistent symmetric reduction 
$D^{(\triangle)}\lll D^{(0)}$ of 
$R_{\mathbf A_{1},\dots,\mathbf A_{s}}$
and an $n$-ary term $\tau$ 
such that $\tau^{\mathbf A_{i}}$ is XY-symetric 
and 
$D_{(\mathbf A_{i},\alpha)}^{(\triangle)}=\{\tau(\alpha)\}$
    for every $i$ and 
    $\alpha\in \TwoTuples(A_{i}^{n})$.
\end{thm}

\begin{thm}\label{THMpropagateXYSymmetric}
Suppose 
$\mathbf A,\mathbf B\in \mathcal V_{n}$,
$0_{\mathbf A}$ is a perfect linear congruence,
$\mathbf A/{0_{\mathbf A}}^{*}\times \mathbf B$ has an XY-symmetric term operation of arity $n$.
Then $\mathbf A\times \mathbf B$ has an XY-symmetric term operation.
\end{thm}

\begin{thm}\label{THMMainInductive}
Suppose 
$\mathbf A_1,\dots,\mathbf A_{s}\in \mathcal V_{n}$,
$n$ is odd.
Then there exists a term $\tau$ such that 
$\tau^{\mathbf A_i}$ is an XY-symmetric operation for every $i$.
\end{thm}

\begin{proof}
First, we reorder algebras so that 
$|A_1|\ge |A_2|\ge \dots \ge|A_s|$.  
We prove the claim by induction on the size of algebras.
Precisely, we assign an infinite tuple 
$(|A_1|,|A_2|,\dots,|A_s|,0,0,\dots)$ to the sequence of algebras, 
and our inductive assumption is that
the statement holds for algebras 
$\mathbf A_1',\dots,\mathbf A_{t}'\in \mathcal V_{n}$
such that
$(|A_1'|,|A_2'|,\dots,|A_t'|,0,0,\dots)<(|A_1|,|A_2|,\dots,|A_s|,0,0,\dots)$
(lexicographic order).

The base of our induction is the case when 
$|A_1|=|A_2|=\dots=|A_s|=1$, which is obvious.

Let us prove the inductive step. 
First, we add all nontrivial 
subalgebras of $\mathbf A_{1}$ to the list 
$\mathbf A_1,\dots,\mathbf A_{s}$ and prove even stronger claim.
We do not want to introduce new notations that is why we assume that 
$\mathbf A_1,\dots,\mathbf A_{s}$ contains all
nontrivial subalgebras of $\mathbf A_{1}$.
Consider two cases.

Case 1. Suppose $\mathbf A_{1}$ has two nontrivial congruences $\sigma$ and $\delta$ such that 
$\sigma\cap \delta$ is the equality relation (0-congruence) on $\mathbf A_1$.
Consider algebras 
$\mathbf A_1/\sigma,\mathbf A_1/\delta,\mathbf A_2,\mathbf A_3,\dots,\mathbf A_{s}\in \mathcal V_{n}$
and apply the inductive assumption.
Then there exists a term $t$ such that 
$t^{\mathbf A_{i}}$ is XY-symmetric for every $i\ge 2$,
$t^{\mathbf A_{1}/\sigma}$ and
$t^{\mathbf A_{1}/\delta}$ are XY-symmetric.
Therefore $t^{\mathbf A_{1}}$ is also XY-symmetric, which completes the proof.

Case 2. There exists a unique minimal nontrivial congruence $\delta$ on $\mathbf A_{1}$.
By the inductive assumption, there exists 
a term $\tau_{0}$ such that 
$\tau_{0}^{\mathbf A_1/\delta}$
and $\tau_{0}^{\mathbf A_i}$ for $i\ge 2$ 
are XY-symmetric.
By Theorem \ref{THMBuildAReductionFromXYSymmetric}
there exists a 1-consistent symmetric reduction 
$D^{(\top)}$ for $R_{\mathbf A_1/\delta, \mathbf A_{2},\dots,\mathbf A_s}$
and an $n$-ary term $\tau_1$
satisfying the corresponding condition.
We define a new reduction for $R_{\mathbf A_1, \mathbf A_{2},\dots,\mathbf A_s}$
as follows.
We put 
$D^{(1)}_{(\mathbf A_{i},\alpha)} = 
D^{(\top)}_{(\mathbf A_{i},\alpha)}$ for $i\ge 2$,
and
$D^{(1)}_{(\mathbf A_{1},\alpha)} = E$
whenever 
$D^{(\top)}_{(\mathbf A_{1},\alpha/\delta)}=\{E\}$.
Applying term $\tau_1$
to the generators of 
$R_{\mathbf A_1, \mathbf A_{2},\dots,\mathbf A_s}$
we obtain a tuple $\gamma\in R_{\mathbf A_1, \mathbf A_{2},\dots,\mathbf A_s}^{(1)}.$
To make the reduction $D^{(1)}$ 1-consistent,
put 
$D^{(2)}_{(\mathbf A_{i},\alpha)} = 
\proj_{(\mathbf A_{i},\alpha)} 
R_{\mathbf A_1, \mathbf A_{2},\dots,\mathbf A_s}^{(1)}$.
By Corollary \ref{CORPropagateMultiplyByCongruence}(t)
and Corollary \ref{CORPropagateToRelations}(r1)
we have $D^{(2)}\lll D^{(0)}$.

Notice that 
$|D^{(2)}_{(\mathbf A_{i},\alpha)}|=1$
for $i\ge 2$.
By Theorem \ref{THMExistenceOfReduction} we have one of the three cases.
In case 2, we can apply Theorem \ref{THMExistenceOfReduction} again 
and obtain even smaller reduction. 
Since we cannot reduce forever, 
we end up with one of the two subcases.

Subcase 1. 
There exists a symmetric 1-consistent reduction $D^{(3)}\lll D^{(2)}$
such that $|D_{(\mathbf A_{i},\alpha)}^{(3)}|=1$.
Take the tuple $\gamma\in R_{\mathbf A_1, \mathbf A_{2},\dots,\mathbf A_s}^{(3)}$ and a term $\tau$ giving
$\gamma$ on the generators of $R_{\mathbf A_1, \mathbf A_{2},\dots,\mathbf A_s}$. It follows from the symmetricity of $D^{(3)}$ that 
$\tau^{\mathbf A_{i}}$ is XY-symmetric for every $i$.

Subcase 2. 
There exists 
a perfect linear congruence 
$\sigma$ on some $\mathbf D^{(0)}_{(\mathbf A_{1},\alpha)}$
such that $D^{(2)}_{(\mathbf A_{1},\alpha)}\times
D^{(2)}_{(\mathbf A_{1},\alpha)}\not\subseteq \sigma$.
Assume that $D^{(0)}_{(\mathbf A_{1},\alpha)}\neq 
A_{1}$.
Since we assumed that all subalgebras of $\mathbf A_1$ are in the list, there exists $k$ such that 
$\mathbf A_{k} = \mathbf D^{(0)}_{(\mathbf A_{1},\alpha)}$.
By the definition of $R_{\mathbf A_{1},\dots,\mathbf A_{s}}$ we have  
$\gamma(\mathbf A_{k},\alpha) = 
\gamma(\mathbf A_{1},\alpha)$ for all 
$\gamma\in R_{\mathbf A_{1},\dots,\mathbf A_{s}}$.
Since the reduction $D^{(2)}$ is 1-consistent 
we obtain that 
$|\mathbf D^{(2)}_{(\mathbf A_{1},\alpha)}|=
|\mathbf D^{(2)}_{(\mathbf A_{k},\alpha)}|=1$,
which contradicts $D^{(2)}_{(\mathbf A_{1},\alpha)}\times
D^{(2)}_{(\mathbf A_{1},\alpha)}\not\subseteq \sigma$.
Thus, 
$D^{(0)}_{(\mathbf A_{1},\alpha)}=A_{1}$
and $\sigma$ is a perfect linear congruence on $\mathbf A_{1}$.
Since $D^{(2)}_{(\mathbf A_{1},\alpha)}\times
D^{(2)}_{(\mathbf A_{1},\alpha)}\not\subseteq \sigma$
and $D^{(2)}_{(\mathbf A_{1},\alpha)}$ is smaller than or equal to  
an equivalence block of $\delta$, 
we have
$\delta\not\subseteq \sigma$.
Since $\delta$ is the minimal nontrivial congruence, 
we obtain that $\sigma=0_{\mathbf A_{1}}$.
Applying Theorem \ref{THMpropagateXYSymmetric}
to 
$\mathbf A_1/0_{\mathbf A_{1}}^{*} \times \mathbf A_2\times\mathbf A_3\times
\dots\times \mathbf A_{s}$ 
we obtain a term $\tau$ such that 
$\tau^{\mathbf A_i}$ is XY-symmetric
for every $i\in[s]$.
\end{proof}

\begin{THMMainTheoremOnXYSymmetricTHM}
Suppose $f$ is a WNU of an odd arity $n$ on a finite set.
Then there exists an XY-symmetric operation $f'\in\Clo(\{f\})$ of arity $n$.
\end{THMMainTheoremOnXYSymmetricTHM}

\begin{proof}
Let $f$ be an operation on a finite set $A$.
By Lemma \ref{LEMExistenceOfSpesialWNULemma}
there exists a special WNU $w\in\Clo(f)$ of arity $N = n^{n!}$.
Consider the algebra $\mathbf A = (A;w)\in\mathcal V_{N}$.
By Theorem \ref{THMMainInductive} there exists 
an $N$-ary operation $w'\in\Clo(w)$ such that 
$w'$ is XY-symmetric.
Then the required $n$-ary XY-symmetric operation can be defined by 
$$
f'(x_1,\dots,x_n) = 
w'(\underbrace{x_1,\dots,x_1}_{n^{n!-1}},
\underbrace{x_2,\dots,x_2}_{n^{n!-1}},
\dots,
\underbrace{x_n,\dots,x_n}_{n^{n!-1}}
)
$$%
\end{proof}

\subsection{Proof of Theorems \ref{THMExistenceOfReduction} and \ref{THMBuildAReductionFromXYSymmetric} (Finding a  reduction)}

\begin{thm}\label{ExistenceOfStrongReductionTHM}
Suppose 
\begin{enumerate}
    \item $\mathbf A_{1},\dots,\mathbf A_{s}\in\mathcal V_{n}$, where $n$ is odd.
    \item 
    $D^{(1)}$ is a 1-consistent symmetric reductions
of a symmetric relation $R\in \mathcal R_{\mathbf A_{1},\dots,\mathbf A_{s}}$, 
    \item $D^{(1)}\lll D^{(0)}$,
\item $B<_{\mathcal T}^{D^{(0)}_{(\mathbf A_{j},\beta)}} 
D^{(1)}_{(\mathbf A_{j},\beta)}$, 
where $\mathcal T\in\{\TBA, \TC, \TPC\}$.

\end{enumerate} 
Then there exists a 1-consistent symmetric reduction $D^{(2)}$ for $R$ such that 
    $D^{(2)}\lll D^{(1)}$ and
    $D^{(2)}\neq D^{(1)}$.
    Moreover, $D^{(2)}\le_{\mathcal T} D^{(1)}$
    if $\mathcal T\neq \TPC$.
\end{thm}

\begin{proof}
For a relation $S$ whose variables are indexed with 
$(\mathbf A_{j},\beta)$
by $S\downarrow^{(\mathbf A_{j},\beta)}_{B}$
denote $S$ whose coordinate $(\mathbf A_{j},\beta)$
is restricted to $S$.

Condition 4 just says that there exists 
$B$ of type $\mathcal T\in\{\TBA,\TC,\TPC\}$.
We want to choose BA or central subuniverse if possible 
and $\TPC$, only if 
none of the 
domains $D^{(1)}_{(\mathbf A_{j},\beta)}$
have proper BA or central subuniverse.
Thus, below we assume that 
$D^{(1)}_{(\mathbf A_{j},\beta)}$ is S-free 
whenever $\mathcal T=\TPC$.

Choose an index $(\mathbf A_{j},\beta)$ 
and 
$B<_{\mathcal T}^{D^{(0)}_{(\mathbf A_{j},\beta)}}D^{(1)}_{(\mathbf A_{j},\beta)}$
such that 
$R\downarrow^{(\mathbf A_{j},\beta)}_{B}$
is inclusion maximal. 

By Corollary \ref{CORPropagateToRelations}(r)
we have $R^{(1)}\lll^{R} R$.
By Lemma \ref{LEMPropagation}(b,bt)
we have 
$R\downarrow^{(\mathbf A_{j},\beta)}_{B}<_{\mathcal T}^{R}
R\downarrow^{(\mathbf A_{j},\beta)}_{D^{(1)}_{(\mathbf A_{j},\beta)}}\lll R$.
By Lemma \ref{LEMIntersectALL}(it)
$R^{(1)}\downarrow^{(\mathbf A_{j},\beta)}_{B}
\le_{\mathcal T}^{R} R^{(1)}$.
Choose $\alpha\in\Perm(\beta)$ 
and put 
$C=\proj_{(\mathbf A_{i},\alpha)}(R^{(1)}\downarrow^{(\mathbf A_{j},\beta)}_{B})$.
Then Lemma \ref{LEMPropagation}(ft) 
implies 
 $C\le_{\mathcal T}^{D_{(\mathbf A_{i},\alpha)}^{(0)}} D_{(\mathbf A_{i},\alpha)}^{(1)}$.
Because of the choice of 
$(\mathbf A_{j},\beta)$ and $B$, 
if $C\neq D_{(\mathbf A_{i},\alpha)}^{(1)}$,
then $R^{(1)}\downarrow^{(\mathbf A_{j},\beta)}_{B} = 
R^{(1)}\downarrow^{(\mathbf A_{i},\alpha)}_{C}$.
If additionally $\sigma(\alpha) = \alpha$ for some permutation on $[n]$, 
then from the symmetricity of 
$R$ we derive  that
$$(R^{(1)}\downarrow^{(\mathbf A_{j},\beta)}_{B})^{\sigma} =
(R^{(1)}\downarrow^{(\mathbf A_{i},\alpha)}_{C})^{\sigma}=
R^{(1)}\downarrow^{(\mathbf A_{i},\alpha)}_{C}=R^{(1)}\downarrow^{(\mathbf A_{j},\beta)}_{B}.$$

Let us consider two cases.

Case 1. There exists  $\alpha\in \Perm(\beta)$ such that
$\alpha\neq \beta$ and  
$\proj_{(\mathbf A_{j},\alpha)}(R^{(1)}\downarrow^{(\mathbf A_{j},\beta)}_{B})\neq \proj_{(\mathbf A_{j},\alpha)}(R^{(1)})$.
Then 
$R^{(1)}\downarrow^{(\mathbf A_{j},\beta)}_{B}$
is $\sigma$-symmetric for any permutation $\sigma$
preserving $\alpha$ and any permutation $\sigma$ preserving $\beta$.
Since 
we can compose such permutations, $\alpha\neq \beta$, and $n$ is odd, we derive that 
$R^{(1)}\downarrow^{(\mathbf A_{j},\beta)}_{B}$
is $\sigma$-symmetric for any $\sigma$, hence $R^{(1)}\downarrow^{(\mathbf A_{j},\beta)}_{B}$ is just symmetric.
Define a new reduction $D^{(2)}$ for $R$ by 
$D^{(2)}_{(\mathbf A_{i},\gamma)}:=
\proj_{(\mathbf A_{i},\gamma)}(R^{(1)}\downarrow^{(\mathbf A_{j},\beta)}_{B})$.
By Corollary \ref{CORPropagateToRelations}(r1)
$D^{(2)}\lll D^{(1)}$, 
and by Lemma \ref{LEMBACenterImplies}
$D^{(2)}\le_{\mathcal T} D^{(1)}$ 
for $\mathcal T\in \{\TBA,\TC\}$, 
which completes this case.
 
Case 2. For any $\alpha\in\Perm(\beta)$  such that 
$\alpha\neq \beta$ 
we have
$\proj_{(\mathbf A_{j},\alpha)}(R^{(1)}\downarrow^{(\mathbf A_{j},\beta)}_{B})=\proj_{(\mathbf A_{j},\alpha)}(R^{(1)})$.
Put 
$S = \bigcap\limits_{\alpha \in \Perm(\beta)}R^{(1)}\downarrow^{(\mathbf A_{j},\alpha)}_{B}$.
$S$ can be represented as an intersection of two symmetric relations, 
hence  
$S$ is also symmetric.
If $S$ is not empty 
then 
we define a new reduction $D^{(2)}$ for $R$ by 
$D^{(2)}_{(\mathbf A_{i},\gamma)}:=
\proj_{(\mathbf A_{i},\gamma)}(S)$.
By Corollary \ref{CORPropagateToRelations}(r1)
$D^{(2)}\lll D^{(1)}$,
and by Lemma \ref{LEMBACenterImplies}
$D^{(2)}\le_{\mathcal T} D^{(1)}$
if $\mathcal T\in\{\TBA,\TC\}$.
Thus, the only remaining case is when $S$ is empty.
By Corollary \ref{CORMainStableIntersection},
there should be 
$(\mathbf A_{j},\alpha_1)$
and 
$(\mathbf A_{j},\alpha_2)$
such that 
$\alpha_1,\alpha_2\in\Perm(\beta)$ and 
$R^{(1)}\downarrow^{(\mathbf A_{j},\alpha_1)}_{B}\cap
R^{(1)}\downarrow^{(\mathbf A_{j},\alpha_2)}_{B}= \varnothing$.
Let $\beta = \sigma(\alpha_1)$.
Since 
$\proj_{(\mathbf A_{j},\alpha)}(R^{(1)}\downarrow^{(\mathbf A_{j},\beta)}_{B})=\proj_{(\mathbf A_{j},\alpha)}(R^{(1)})$
for any $\alpha\in\Perm(\beta)$ such that $\alpha\neq \beta$,
we have
\begin{align*}
    (R^{(1)}\downarrow^{(\mathbf A_{j},\alpha_1)}_{B}\cap
R^{(1)}\downarrow^{(\mathbf A_{j},\alpha_2)}_{B})^{\sigma} = 
(R^{(1)}\downarrow^{(\mathbf A_{j},\alpha_1)}_{B})^{\sigma}\cap
(R^{(1)}\downarrow^{(\mathbf A_{j},\alpha_2)}_{B})^{\sigma}& =\\ 
R^{(1)}\downarrow^{(\mathbf A_{j},\beta)}_{B}
\cap 
&R^{(1)}\downarrow^{(\mathbf A_{j},\sigma(\alpha_2))}_{B}
\neq \varnothing.
\end{align*}
This contradiction completes the proof.
\end{proof}

\begin{lem}\label{LEMExistenceOfReduction}
Suppose $\mathbf A_{1},\dots,\mathbf A_{s}\in\mathcal V_{n}$, $n$ is odd, $D^{(1)}$ is a 1-consistent symmetric reduction of $R_{\mathbf A_{1},\dots,\mathbf A_{s}}$, 
$D^{(1)}\lll D^{(0)}$,
$B<_{\TD(\sigma)}^{D_{(\mathbf A_{i},\alpha)}^{(0)}}
D_{(\mathbf A_{i},\alpha)}^{(1)}$ for some 
$(\mathbf A_{i},\alpha)$,
$R_{\mathbf A_{1},\dots,\mathbf A_{s}}^{(1)}$
has no tuple $\gamma$ such that 
$\gamma(\mathbf A_{i},\beta)\in B$ for every $\beta\in\Perm(\alpha)$.
Then 
$\sigma$ is a perfect linear congruence $\sigma$ on $D^{(0)}_{(\mathbf A_{i},\alpha)}$.
\end{lem}

\begin{proof}
    By Corollary \ref{CORMainStableIntersection}
there must be a bridge $\delta$ from $\sigma$ to $\sigma$ 
such that $\widetilde \delta$ contains the relation $\proj_{(\mathbf A_{i},\alpha_1),(\mathbf A_{i},\alpha_2)}
R_{\mathbf A_{1},\dots,\mathbf A_{s}}$
for some $\alpha_1,\alpha_2\in\Perm(\alpha)$.
Since $\proj_{(\mathbf A_{i},\alpha_1),(\mathbf A_{i},\alpha_2)}
R_{\mathbf A_{1},\dots,\mathbf A_{s}}$ is linked, 
by Lemma \ref{LEMMakePerfectCongruenceFromLinked}
$\sigma$ is a perfect linear congruence, which completes the proof.
\end{proof}

\begin{THMExistenceOfReductionTHM}
Suppose $\mathbf A_{1},\dots,\mathbf A_{s}\in\mathcal V_{n}$, $n$ is odd, $D^{(1)}$ is a 1-consistent symmetric reduction of $R_{\mathbf A_{1},\dots,\mathbf A_{s}}$, 
$D^{(1)}\lll D^{(0)}$. 
Then one of the following conditions hold
\begin{enumerate}
    \item $|D^{(1)}_{(\mathbf A_{i},\alpha)}|=1$ 
    for all $(\mathbf A_{i},\alpha)$.
    \item there exists a 1-consistent symmetric reduction $D^{(2)}$ for $R_{\mathbf A_{1},\dots,\mathbf A_{s}}$ such that 
    $D^{(2)}\lll D^{(1)}$ and $D^{(2)}\neq D^{(1)}$.
    \item there exists a perfect linear congruence $\sigma$ on some
    $D^{(0)}_{(\mathbf A_{i},\alpha)}$ such that 
    \begin{enumerate}
        \item $D_{(\mathbf A_{i},\alpha)}^{(1)}\times D_{(\mathbf A_{i},\alpha)}^{(1)}\not\subseteq \sigma$ 
        \item $D_{(\mathbf A_{i},\alpha)}^{(1)}\times D_{(\mathbf A_{i},\alpha)}^{(1)}\subseteq \sigma^{*}$
    \end{enumerate}
\end{enumerate}
\end{THMExistenceOfReductionTHM}
\begin{proof}
If condition 1 holds then we are done.
Otherwise, choose some 
$(\mathbf A_{i},\alpha)$
such that 
$|D_{(\mathbf A_{i},\alpha)}^{(1)}|>1$.
By Lemma \ref{LEMUbiquity}
there exists 
$B<_{\mathcal T(\sigma)}^{D_{(\mathbf A_{i},\alpha)}^{(0)}}
D_{(\mathbf A_{i},\alpha)}^{(1)}$ 
for some $\mathcal T\in\{\TBA,\TC,\TD\}$.
If $\mathcal T$ can be chosen from $\{\TBA,\TC\}$ 
then condition 2 follows from Theorem \ref{ExistenceOfStrongReductionTHM}.
Otherwise, we assume that $\mathcal T=\TD$.
Define a reduction $D^{(\top)}$
as follows.
If $\beta\in\Perm(\alpha)$ put 
$D^{(\top)}_{(\mathbf A_{i},\beta)} = B$,
for all other $(\mathbf A_{j},\beta)$
put 
$D^{(\top)}_{(\mathbf A_{j},\beta)} = D^{(1)}_{(\mathbf A_{j},\beta)}$.
We consider two cases:

Case 1. 
$R_{\mathbf A_{1},\dots,\mathbf A_{s}}^{(\top)}$ is not empty.
Define a new reduction $D^{(2)}$ 
by 
$D^{(2)}_{(\mathbf A_{j},\beta)} = 
\proj_{(\mathbf A_{j},\beta)}(R_{\mathbf A_{1},\dots,\mathbf A_{s}}^{(\top)})$.
By Corollary \ref{CORPropagateToRelations}(r1) 
we have
$D^{(2)}\lll^{D} D^{(1)}$.
Since $D^{(\top)}$ and 
$R_{\mathbf A_{1},\dots,\mathbf A_{s}}$
are symmetric, 
$D^{(2)}$ is also symmetric. Thus, we satisfied condition 2.

Case 2. 
$R_{\mathbf A_{1},\dots,\mathbf A_{s}}^{(\top)}$ is  empty.
Then by Lemma \ref{LEMExistenceOfReduction}
condition 3 is satisfied.
\end{proof}

\begin{THMBuildAReductionFromXYSymmetricTHM}
Suppose 
$\mathbf A_1,\dots,\mathbf A_s\in\mathcal V_{n}$, 
$n$ is odd, 
there exists an $n$-ary term $\tau_{0}$ such that 
$\tau_0^{\mathbf A_{i}}$ is 
XY-symmetric for every $i$. 
Then there exists a 
1-consistent symmetric reduction 
$D^{(\triangle)}\lll D^{(0)}$ of 
$R_{\mathbf A_{1},\dots,\mathbf A_{s}}$
and an $n$-ary term $\tau$ 
such that $\tau^{\mathbf A_{i}}$ is XY-symmetric 
and 
$D_{(\mathbf A_{i},\alpha)}^{(\triangle)}=\{\tau(\alpha)\}$
    for every $i$ and 
    $\alpha\in \TwoTuples(A_{i}^{n})$.
\end{THMBuildAReductionFromXYSymmetricTHM}

\begin{proof}
Let $I$ be the set of all indices of the 
coordinates $R_{\mathbf A_1,\dots,\mathbf A_s}$.
That is, $I$ consists of pairs $(\mathbf A_{i},\alpha)$ 
and $|I|$ is the arity of $R_{\mathbf A_1,\dots,\mathbf A_s}$.

We build a sequence of symmetric reductions 
$D^{(s)}\lll^{D^{(0)}}
D^{(s-1)}\lll^{D^{(0)}}\dots 
\lll^{D^{(0)}} D^{(1)}\lll D^{(0)}$
for $R_{\mathbf A_1,\dots,\mathbf A_s}$.
For every $j\in\{0,1,\dots,s\}$ 
and every $i\in I$ 
we define a congruence  
$\delta_{i}^{j}$  on $D_{i}^{(0)}$ such that 
the following conditions hold for every $j$:
\begin{enumerate}
    \item[(1)] $\delta_{(\mathbf A_{i},\alpha)}^{j} = 
    \delta_{(\mathbf A_{i},\beta)}^{j}$
    whenever $\beta\in\Perm(\alpha)$.
    \item[(2)] there exists a 
    tuple $\gamma \in R_{\mathbf A_1,\dots,\mathbf A_s}^{(j)}$
    such that 
    $(\gamma(\mathbf A_{i},\alpha),
    \gamma(\mathbf A_{i},\beta))\in \delta_{(\mathbf A_{i},\alpha)}^{j}$ for every $(\mathbf A_{i},\alpha)\in I$ and $\beta\in\Perm(\alpha)$.
    \item[(3)] $\delta_{i}^{j+1}\supseteq \delta_{i}^{j}$ 
    for every $i\in I$.
    \item[(4)] if we make some congruences 
    $\delta_{i}^{j}$ smaller then condition (2)
    or condition (3) will not hold.
\end{enumerate}

Thus, for every $j$ we choose  minimal congruences (but not smaller than the 
previous congruences)
such that 
$R_{\mathbf A_1,\dots,\mathbf A_s}^{(j)}$ has a tuple 
whose corresponding elements are equivalent modulo
these congruences.
We start with $s=0$
and $\delta_{i}^{0}$ is the equality relation (0-congruence) for every $i\in I$.

We also need the following condition.
\begin{enumerate}
    \item[(5)] 
    $\delta_{i}^{j}=\delta_{i}^{j+1}$ for every $i\in I$ 
    or 
    $D^{(j+1)}\le_{\mathcal T} D^{(j)}$ 
    for some $\mathcal T\in \{\TBA,\TC\}$.
\end{enumerate}

Suppose we already have $D^{(s)}$ and all 
$\delta_{i}^{s}$ are are defined.
Let us show how to build $D^{(s+1)}$ and $\delta_{i}^{s+1}$.
Using Lemma \ref{LEMUbiquity} we consider three cases:

Case 1. 
There exists $B<_{\mathcal T} D^{(s)}_{i}$ for some $i\in I$
and $\mathcal T\in\{\TBA,\TC\}$.
By Theorem \ref{ExistenceOfStrongReductionTHM}
there exists a 1-consistent reduction
$D^{(s+1)}\le_{\mathcal T} D^{(s)}$.
It remains to define 
$\delta_{i}^{s+1}$ for every $i\in I$.
We choose them freely to satisfy the above condition (1)-(4).

Case 2. $|D^{(s)}_{i}|=1$ for every $i\in I$.
Then we put $D^{(\triangle)} = D^{(s)}$ 
and choose an $n$-ary term $\tau$
such that 
$D_{(\mathbf A_{i},\alpha)}^{(\triangle)}=\{\tau(\alpha)\}$
    for every $i$ and 
    $\alpha\in \TwoTuples(A_{i}^{n})$.
    This completes
the proof in this case. 

Case 3. Otherwise, choose the maximal 
$\ell\in\{0,1,\dots,s\}$ such that$|D_{i}^{(s)}/\delta_{i}^{\ell}|>1$ for some $i\in I$.
Since congruences 
$\delta_{i}^{0}$ are equalities and case 2 does not hold, 
such $\ell$ and $i$ always exist.
By Lemma \ref{LEMUbiquity} choose 
$B'<_{\mathcal T(\sigma')}^{D_{i}^{(0)}/\delta_{i}^{\ell}}
D_{i}^{(s)}/\delta_{i}^{\ell}$.
If $\mathcal T\in\{\TBA,\TC\}$ 
then by Corollary \ref{CORPropagateFromFactor}(t) 
and Lemma \ref{LEMIntersectALL}(t)
we have 
$(\bigcup_{E\in B'} E) \cap 
D_{i}^{(s)} <_{\mathcal T} D_{i}^{(s)}$, 
which means that it is case 1.
Thus, $\mathcal T=\TD$. 
By extending $\sigma'$ 
to $D_{i}^{(0)}$ we get 
an irreducible congruence 
$\sigma\supseteq \delta_{i}^{\ell}$. 
Consider two subcases:

Subcase 3.1. $\ell = s$.
Choose an equivalence class $B$
of $\sigma$ containing 
$\gamma(i)$ for $\gamma$ satisfying condition (2) for 
$j=s$.
Define two reductions $D^{(\bot)}$ and $D^{(s+1)}$.
Put $D_{j}^{(\bot)} = D_{j}^{(s)}$ for every $j\notin\Perm(i)$ and 
$D_{j}^{(\bot)} = B\cap D_{j}^{(s)}$ for $j\in\Perm(i)$.
Put $D_{j}^{(s+1)}=
\proj_{j}(R_{\mathbf A_1,\dots,\mathbf A_s}^{(\bot)})$
for every $j\in I$.
Notice that
$D^{(\bot)}\le_{\TD} D^{(s)}$ 
and by Corollary \ref{CORPropagateToRelations}(r1) we
have $D^{(s+1)}\lll^{D^{(0)}} D^{(s)}$.
Put $\delta_{i'}^{s} = \delta_{i'}^{s+1}$
for every $i'\in I$.
It is not hard to see that conditions (1)-(5) are satisfied.

Subcase 3.2. $\ell<s$.
Choose some $B<_{\TD(\sigma)}^{D_{i}^{(0)}}D_{i}^{(s)}$.
Since $\ell$ was chosen maximal, 
$\sigma \not\supseteq \delta_{i}^{\ell+1}$.
Since 
$\delta_{i}^{\ell+1}$ was chosen minimal, 
$R_{\mathbf A_1,\dots,\mathbf A_s}^{(s)}$ 
has no tuples 
$\gamma$ 
such that 
$\gamma_{i'} \in B$ for all $i'\in \Perm(i)$
(otherwise we could replace $\delta_{i}^{\ell+1}$
by $\delta_{i}^{\ell+1}\cap \sigma$). 
By Lemma \ref{LEMExistenceOfReduction},
$\sigma$ is a perfect linear congruence.
Let $i = (\mathbf A_{k},\alpha)$. 
Then 
there exists $\zeta\le \mathbf A_{k}\times \mathbf A_{k}\times \mathbf Z_{p}$ 
such that $\proj_{1,2} \zeta =\sigma^{*}$
and $(a_{1},a_{2},b)\in\zeta$ 
implies that 
$(a_{1},a_{2})\in \sigma\Leftrightarrow (b=0)$.
Let us define a conjunctive formula $\Theta$
with variables $\{x_j\mid j\in I\}\cup 
\{z_{i'}\mid i'\in\Perm(i)\}$.
For every $j\in I$ and every 
$j'\in\Perm(j)$ we add 
the constraint $\delta_{j}^{\ell+1}(x_{j},x_{j'})$,
for every 
$i'\in\Perm(i)$ we add 
the constraint $\zeta(x_{i},x_{i'},z_{i'})$, 
finally add 
$R_{\mathbf A_1,\dots,\mathbf A_s}$ once with the corresponding variables.
Since $\delta_{i}^{\ell+1}$ was chosen minimal
and $\sigma \not\supseteq \delta_{i}^{\ell+1}$,
$\Theta^{(\ell+1)}$ is not satisfiable if 
$z_{i'} = 0$ for every $i'\in\Perm(i)$. 
Nevertheless, it is satisfied for some $z_{i'}$.
Also, 
since
$\delta_{i}^{\ell}\subseteq\sigma$, 
the formula
$\Theta^{(\ell)}$ has a solution with 
$z_{i'} = 0$ for every $i'\in\Perm(i)$.
By condition (5)  
$D^{(\ell+1)}\le_{T} D^{(\ell)}$
for some $T\in \{\TBA,\TC\}$.
Thus, the projections $L_{1}$ and $L_{0}$ of the solution sets of $\Theta^{(\ell)}$
and $\Theta^{(\ell+1)}$, respectively,
onto the $z$-variables are different (in the point $(0,0,\dots,0)$),
which by Lemma \ref{LEMBACenterImplies}
implies that 
$L_{0}<_{T} L_{1}\le \mathbf Z_{p}^{|\Perm(i)|}$.
This contradicts Lemma \ref{LEMNoAbsCenterPCInLinearAlgebra}.
\end{proof}

\subsection{Proof of Theorem 
\ref{THMpropagateXYSymmetric} (Fixing an operation)}\label{SUBSECTIONFixingAnOperation}

\begin{lem}\label{ZetaPropertiesLEM}
Suppose $\mathbf A = (A;w^{\mathbf A})$,
where $w^{\mathbf A}$ is an $n$-ary idempotent operation,
$0_{\mathbf A}$ is a perfect linear congruence 
witnessed by $\zeta\le \mathbf A\times \mathbf A\times \mathbf Z_{p}$.
Then 
\begin{enumerate}
\item[(1)] 
for every $(a,b)\in 0_{\mathbf A}^{*}$ there exists a unique 
$c$ such that $(a,b,c)\in \zeta$,
\item[(2)] $n-1$ is divisible by $p$,
\item[(3)] $w^{\mathbf A}(a,\dots,a,b,a,\dots,a) = b$ for every $(a,b)\in 0_{\mathbf A}^{*}$ and any position of $b$,
\item[(4)] 
for every $a\in A$ and $c\in\mathbf Z_{p}$ there exists at most one 
$b\in A$ such that $(a,b,c\in \zeta)$,
\item[(5)] 
for every $b\in A$ and $c\in\mathbf Z_{p}$ there exists at most one
$a\in A$ such that $(a,b,c\in \zeta)$,
\item[(6)] 
if $(a,b,d),(b,c,e)\in\zeta$ then 
$(a,c,d+e)\in\zeta$.

\end{enumerate}
\end{lem}
\begin{proof}
Let us prove (1).
Assuming the opposite we take 
$d_1$ and $d_2$ such that 
$(a,b,d_1), (a,b,d_2)\in\zeta$.
Applying $w$ to these tuples and using idempotency of $\mathbf A$ we derive that 
$(a,b,k_1\cdot d_1+k_2\cdot d_2)\in \zeta$
for any $k_1$ and $k_2$ such that 
$k_1+k_2-1$ is divisible by $n-1$.
Since $d_1\neq d_2$ we can choose $k_1$ and $k_2$ such that 
$k_1\cdot d_1+k_2\cdot d_2=0 (mod\;\;p)$, which 
implies 
$(a,b,0)\in \zeta$ and contradicts the definition of $\zeta$.

Let us prove (2). Applying 
$w$ to $n$ tuples $(a,b,d)$ for some $d\neq 0$, we get a tuple 
$(a,b,n\cdot d)$.
By (1) we obtain 
$n\cdot d = d (mod\;\;\;p)$ 
and $n-1$ is divisible by $p$.

Let us prove (3) using (2).
Applying 
$w$ to $n-1$ tuples $(a,b,d)$ 
and one tuple $(b,b,0)$ we get a tuple 
$(w^{\mathbf A}(a,\dots,a,b,a,\dots,a),b,0)$.
Hence $w^{\mathbf A}(a,\dots,a,b,a,\dots,a)=b$.

Let us prove (4). 
Assume that 
$(a,b_1,d),(a,b_2,d)\in \zeta$.
By properties (2) and (3), applying 
$w$ to 
$n-2$ tuples $(a,b_2,d)$,
one tuple $(a,b_1,d)$ and one tuple $(b_2,b_2,0)$
we obtain 
$(b_2,b_1,0)\in\zeta$, which means $b_1=b_2$.

Property (5) is proved in the same way as (4).

To prove (6) apply $w$ to the tuples 
$(a,b,d),(b,c,e)$ and $(n-2)$ tuples 
$(b,b,0)$. 
By (2) and (3) we have
$(a,c,d+e)\in\zeta$, which completes the proof.
\end{proof}

Using claim (1) in Lemma \ref{ZetaPropertiesLEM}
we can define a binary operation 
$\zeta: A\times A\to \mathbf Z_{p}$ 
by 
$\zeta(x_1,x_2)=z\leftrightarrow 
(x_1,x_2,z)\in \zeta$.
Then property (6) 
means 
$\zeta(a,c) = \zeta(a,b) + \zeta(b,c)$
for any $(a,b),(b,c)\in 0_{\mathbf A}^{*}$.

\begin{thm}\label{ExistenceOfInjectiveHomomorphismTHM}
Suppose 
$\mathbf A\in \mathcal V_{n}$,
$0_{\mathbf A}$ is a perfect linear congruence.
Then there exists an algebra $\mathbf C\in (\mathbf A/0_{\mathbf A}^{*}\boxtimes\mathbf Z_{p})\cap \mathcal V_{n}$ such that 
there exists an injective homomorphism $h\colon \mathbf A\to \mathbf C$.
\end{thm}

\begin{proof}
Suppose $0_{\mathbf A}$ is witnessed by 
$\zeta\le \mathbf A\times\mathbf A\times \mathbf Z_{p}$.
Then we can use all the properties of $\zeta$ proved in Lemma \ref{ZetaPropertiesLEM}.

Let us define an algebra $\mathbf C=(A/0_{\mathbf A}^{*}\times\mathbf Z_{p};w^{\mathbf C})$.
Let $\phi:A/0_{\mathbf A}^{*}\to A$ be an injection such that 
$\phi(\theta)\in \theta$ for every equivalence class $\theta$. Thus, we just chose a representative from 
every equivalence class of $0_{\mathbf A}^{*}$.
Put 
$$(w^{\mathbf C})^{(1)}(x_1,\dots,x_{n})
=  w^{\mathbf A/0_{\mathbf A}^{*}}(x_1^{(1)},\dots,x_{n}^{(1)}),$$
$$(w^{\mathbf C})^{(2)}(x_1,\dots,x_{n})
=  \zeta\left(w^{\mathbf A}(\phi(x_1^{(1)}),\dots,\phi(x_{n}^{(1)}))   
, \phi(w^{\mathbf A/0_{\mathbf A}^{*}}(x_1^{(1)},\dots,x_{n}^{(1)}))\right)
+x_{1}^{(2)}+\dots+x_{n}^{(2)}.$$

Let us define an injective homomorphism 
$h\colon \mathbf A\to \mathbf C$
by $h(a) = (a/0_{\mathbf A}^{*}, \zeta(a,\phi(a/0_{\mathbf A}^{*})))$.
By Lemma \ref{ZetaPropertiesLEM}(5) $h$ is injective.
Let us prove that $h$ is a homomorphism. 
It follows immediately from the definition of $h$ and $w^{\mathbf C}$ that 
$$(h(w^{\mathbf A}(x_1,\dots,x_n)))^{(1)} =
 w^{\mathbf A/0_{\mathbf A}^{*}}(x_1/0_{\mathbf A}^{*},\dots,x_n/0_{\mathbf A}^{*})=
(w^{\mathbf C}(h(x_1),\dots,h(x_n)))^{(1)}.$$

Applying 
$w^{\mathbf A}$ to 
the tuples 
$(x_i,\phi(x_i/0_{\mathbf A}^{*}),
h^{(2)}(x_i))\in\zeta$ for $i=1,2,\dots,n$ 
we obtain 
$$\zeta(w^{\mathbf A}(x_1,\dots,x_n),
w^{\mathbf A}(\phi(x_1/0_{\mathbf A}^{*}),\dots,\phi(x_n/0_{\mathbf A}^{*})))
=
h^{(2)}(x_1)+\dots+h^{(2)}(x_n).$$
Using Lemma \ref{ZetaPropertiesLEM}(6)
we derive 
\begin{align*}
h^{(2)}(w^{\mathbf A}(x_1,\dots,x_n)&)=\\
\zeta(w^{\mathbf A}(x_1,\dots,&x_n),
\phi(w^{\mathbf A}(x_1,\dots,x_n)/0_{\mathbf A}^{*})) =\\ 
\zeta(w^{\mathbf A}&(x_1,\dots,x_n),
w^{\mathbf A}(\phi(x_1/0_{\mathbf A}^{*}),\dots,\phi(x_n/0_{\mathbf A}^{*})))
+\\
&\;\;\;\;\;\;\;\zeta(
w^{\mathbf A}(\phi(x_1/0_{\mathbf A}^{*}),\dots,\phi(x_n/0_{\mathbf A}^{*})),
\phi(w^{\mathbf A}(x_1,\dots,x_n)/0_{\mathbf A}^{*}))=\\
h^{(2)}(&x_1)+\dots+h^{(2)}(x_n) 
+\\
&\zeta(
w^{\mathbf A}(\phi(h^{(1)}(x_1)),\dots,\phi(h^{(1)}(x_n))),
\phi(w^{\mathbf A/0_{\mathbf A}^{*}}(h^{(1)}(x_1),\dots,h^{(1)}(x_1))))=\\
&\;\;\;\;\;\;\;\;\;\;\;\;\;\;\;\;\;\;\;\;\;\;\;\;\;\;\;\;\;\;\;\;\;\;\;\;\;\;\;\;\;\;\;\;\;\;\;\;\;\;\;\;
\;\;\;\;\;\;\;\;\;\;
\;\;\;\;\;\;\;\;\;\;(w^{\mathbf C})^{(2)}(h(x_1),\dots,h(x_{n}))
\end{align*}
Hence, 
$h(w^{\mathbf A}(x_1,\dots,x_n)) =
w^{\mathbf C}(h(x_1),\dots,h(x_n))$
and $h$ is a homomorphism.

It follows from the definition that 
$\mathbf C\in \mathbf A/0_{\mathbf A}^{*}\boxtimes\mathbf Z_{p}$.
It remains to show that $w^{\mathbf C}\in\mathcal V_{n}$.
Since $w^{\mathbf A}$ is a WNU and 
the addition $x_{1}^{(2)}+\dots+x_{n}^{(2)}$ in the definition of $(w^{\mathbf C})^{(2)}$ is symmetric, 
$w^{\mathbf C}$ is also a WNU.

Since $h$ is an isomorphism 
from $A$ to $h(A)$,
$w^{\mathbf C}$ is a special WNU on 
$h(A)$.
Notice that 
$h(\phi(x/0_{\mathbf A}^{*})) = 
(x/0_{\mathbf A}^{*},0)$, which means that 
$(b,0)\in h(A)$ for any $b\in A/0_{\mathbf A}^{*}$.
For $x\in C$ put 
$x' = (x^{(1)},0)$.

Let us show that 
$w^{\mathbf C}$ is idempotent.
Since $w^{\mathbf C}$ is idempotent on $h(A)$, 
we have 
$w^{\mathbf C}(x',x',\dots,x') = x'$ and 
$(w^{\mathbf C})^{(2)}(x',x',\dots,x') = 0$.
Hence, 
$(w^{\mathbf C})^{(2)}(x,x,\dots,x) = 
(w^{\mathbf C})^{(2)}(x',x',\dots,x')+
n\cdot x^{(2)} = x^{(2)}$
and $w^{\mathbf C}(x,x,\dots,x)=x$.

 Let us show that 
$w^{\mathbf C}$ is special and
$w^{\mathbf C}(x,x,\dots,x,w^{\mathbf C}(x,x,\dots,x,y)) = w^{\mathbf C}(x,x,\dots,x,y)$.
For $(w^{\mathbf C})^{(1)}$ it follows from the fact that $w^{\mathbf A/0_{\mathbf A}^{*}}$ is special. 
Since $w^{\mathbf C}$ is special on $f(A)$, we have  
$$w^{\mathbf C}(x',\dots,x',w(x',\dots,x',y')) = 
w^{\mathbf C}(x',\dots,x',y').$$ 
Looking at the second components in the above equality
we obtain
$(w^{\mathbf C})^{(2)}(x',x',\dots,x',y')=0$. 
Therefore,
using the definition of $w^{\mathbf C}$ we derive 
\begin{align*}
(w^{\mathbf C})^{(2)}(x,x,\dots,x,w^{\mathbf C}(x,x,\dots,&x,y))=\\
2\cdot(w^{\mathbf C})^{(2)}(x',x',\dots,x',y')&+2\cdot (n-1)\cdot x^{(2)}+y^{(2)}=\\
&\;\;(w^{\mathbf C})^{(2)}(x',x',\dots,x',y')+
(n-1)\cdot x^{(2)}+y^{(2)}=\\
&\;\;\;\;\;\;\;\;\;\;\;\;\;\;\;\;\;\;\;\;\;\;\;\;\;\;\;\;\;\;\;\;\;\;\;\;\;\;\;\;\;\;\;\;\;\;\;\;\;\;\;\;\;\;\;\;\;(w^{\mathbf C})^{(2)}(x,x,\dots,x,y)
\end{align*}

\end{proof}

For a set $B$ and prime $p$ by 
$\mathcal C_{B\boxtimes \mathbf Z_{p}}$ we denote the set of all operations $f$ 
on $B\times Z_{p}$ such that 
$f^{(1)}(x_1,\dots,x_{n})$ depends only on $x_{1}^{(1)},\dots,x_{n}^{(1)}$
and 
 $f^{(2)}(x_1,\dots,x_{n}) = 
 f^{(2,1)}(x_1^{(1)},\dots,x_{n}^{(1)})
 +a_1x_1^{(2)}+\dots+a_nx_{n}^{(2)}$
 for some 
 $f^{(2,1)}:B^{n}\to Z_{p}$ and $a_{1},\dots,a_{n}\in Z_{p}$.
 The subfunction $a_1x_1^{(2)}+\dots+a_nx_{n}^{(2)}$
 we sometimes denote by $f^{(2,2)}$.
Sometimes we write $f^{(1)}(x_1^{(1)},\dots,x_n^{(1)})$ instead of 
$f^{(1)}(x_1,\dots,x_n)$ to point out that $f^{(1)}$ depends only on first components.

For an operation $f$ of arity $n$ we define operations $t_{\ell}^{f}$ for $\ell \in\{n,n-1,\dots,1,0\}$ as follows.
Put $t_{n}^{f}(x_1,\dots,x_{n},y_{1},\dots,y_n) = f(x_1,\dots,x_{n})$.
For $\ell \in\{n-1,\dots,1,0\}$ put
\begin{align*}
t_{\ell}^{f}(x_1,\dots,x_{\ell},y_1,\dots,y_n)
=f(&t_{\ell+1}(x_1,\dots,x_{\ell},y_1,y_1,\dots,y_n),\\
&t_{\ell+1}(x_1,\dots,x_{\ell},y_2,y_1,\dots,y_n),\dots,t_{\ell+1}(x_1,\dots,x_{\ell},y_n,y_1,\dots,y_n))
\end{align*}

\begin{lem}\label{kWNUImplieskWNU}
Suppose $f$ is a $k$-WNU. 
Then $t_{0}^{f}$ is a $k$-WNU.
\end{lem}
\begin{proof}
Suppose tuples $(a_1,\dots,a_n),(b_1,\dots,b_n)\in\{c,d\}^{n}$
contain exactly $k$ elements $d$.
We need to show that $t_{0}^{f}(a_1,\dots,a_n) = t_{0}^{f}(b_1,\dots,b_n)$.
We prove by induction on $\ell$ starting with $\ell = n$ to $\ell =0$ that
$t_{\ell}^{f}(x_1,\dots,x_{\ell},a_1,\dots,a_n) = t_{e\\}^{f}(x_1,\dots,x_{\ell},b_1,\dots,b_n)$
for all $x_1,\dots,x_{\ell}$. 
For $\ell = n$ it is obvious.
From the inductive assumption on $t_{\ell+1}^{f}$
we conclude that 
$$a_{i} = b_{j} \Longrightarrow t_{\ell+1}(x_1,\dots,x_{\ell},a_i,a_1,\dots,a_n)=
t_{\ell+1}(x_1,\dots,x_{\ell},b_j,b_1,\dots,b_n).$$
Then the inductive step follows from the definition of $t_{\ell}^{f}$ and the fact that $f$ is a $k$-WNU.
\end{proof}

Note, that if we write the term defining $t_{0}^{f}(y_1,\dots,y_{n})$ then 
for every $(i_1,\dots,i_n)\in[n]^{n}$ there exists exactly one 
internal formula $f(y_{i_1},\dots,y_{i_n})$.
For two tuples $(i_1,i_2,\dots,i_n),(j_1,j_2,\dots,j_{n})\in [n]^{n}$ 
by 
$t_{0,(i_1,i_2,\dots,i_n)}^{f,(j_1,j_2,\dots,j_{n})}$
we denote the operation defined by the same term as 
$t_{0}^{f}$ but with 
$f(y_{i_1},\dots,y_{i_n})$ replaced 
by $f(y_{j_1},\dots,y_{j_n})$.


For a tuple $\alpha$ and an element $b$ by 
$N_{b}(\alpha)$ we denote the number of elements in $\alpha$ that are equal to $b$.
By $T_{a,b}^{n,k}$ we denote the set of all 
tuples $\gamma\in \{a,b\}^{n}$ such that 
$N_{b}(\gamma) = k$ and $\gamma(1) = a$.

For two tuples 
$\alpha,\beta \in T_{a,b}^{n,k}$  
we define a tuple $\xi(\alpha,\beta)=
(i_1,\dots,i_{n})\in [n]^{n}$
as follows.
Let 
$j_1<\dots<j_k$ and 
$s_1<\dots<s_k$
be the lists of positions of 
$b$ in $\alpha$ and $\beta$ respectively.
Put 
$i_{\ell} = 1$ if $\beta(\ell) = a$,
and 
$i_{\ell} = j_{m}$ if $\ell = s_{m}$. 
For example, 
$\xi((a,a,b,a,b,a,b,b,a),(a,b,a,b,b,a,a,a,b)) = 
(1,3,1,5,7,1,1,1,8)$.
Notice that for any $\alpha,\beta,\gamma\in T_{a,b}^{n,k}$
the tuple $\xi(\alpha,\beta)$
is a permutation of the tuple $\xi(\alpha,\gamma)$.


\begin{lem}\label{SpecialWNUPreservingDeltaLEM}
Suppose $f$ is a special idempotent WNU operation of arity $n$, $f\in \mathcal C_{B\boxtimes \mathbf Z_{p}}$.
Then 
$f^{(2,2)}(x_1^{(2)},\dots,x_n^{(2)}) = 
x_{1}^{(2)}+\dots+x_{n}^{(2)}$
and $p$ divides $n-1$.
\end{lem}

\begin{proof}
By the definition of $\mathcal C_{B\boxtimes \mathbf Z_{p}}$ we 
have 
$f^{(2,2)}(x_1^{(2)},\dots,x_n^{(2)}) = a_{1}x_1+\dots+a_{n}x_{n}$.
Since $f$ is a WNU, we have $a_{1} = a_{2} = \dots = a_{n}$.
Since $f$ is idempotent, $n\cdot a_1=1 \mod p$.
Since $f$ is special we have 
$f(y,y,\dots,y,f(y,y,\dots,y,x)) = f(y,y,\dots,y,x)$, which implies
$a_1\cdot a_1 = a_1$ and $a_1 = 1$.
\end{proof}

The next lemma follows immediately from the definition of $t_{0}^{f}$
and $t_{0,(i_1,i_2,\dots,i_n)}^{f,(j_1,j_2,\dots,j_{n})}$.


\begin{lem}\label{DifferenceTwoTwoReplacementLEM}
Suppose 
$f\in \Pol(\sigma_{B\times \mathbb Z})$,
$f^{(2,2)} (x_1^{(2)},\dots,x_n^{(2)})= x_1^{(2)}+\dots + x_n^{(2)}$, 
where $p$ divides $n-1$,
and
$f^{(1)}(x_{i_1},\dots,x_{i_n}) = f^{(1)}(x_{j_1},\dots,x_{j_n})$.
Then 
$(t_{0,(i_1,i_2,\dots,i_n)}^{f,(j_1,j_2,\dots,j_{n})})^{(1)}=(t_{0}^{f})^{(1)}$
and 
$(t_{0,(i_1,i_2,\dots,i_n)}^{f,(j_1,j_2,\dots,j_{n})})^{(2)}(x_{1},\dots,x_{n})=(t_{0}^{f})^{(2)}(x_{1},\dots,x_{n})
-f^{(2)}(x_{i_1},\dots,x_{i_n}) +  f^{(2)}(x_{j_1},\dots,x_{j_n})$.
\end{lem}

\begin{lem}\label{TwoTwoPartLEM}
Suppose 
$f\in \Pol(\sigma_{B\times \mathbb Z})$,
$f^{(2,2)} (x_1^{(2)},\dots,x_n^{(2)})= x_1^{(2)}+\dots + x_n^{(2)}$, 
where $p$ divides $n-1$.
Then 
$(t_{0}^{f})^{(2,2)}(x_1^{(2)},\dots,x_n^{(2)}) = x_1^{(2)}+\dots + x_n^{(2)}$.
\end{lem}

\begin{proof}

\begin{align*}
(t_{0}^{f})^{(2,2)}(x_1^{(2)},\dots,x_n^{(2)}) =&\\ 
\sum\limits_{i_1,\dots,i_n\in[n]}&
(x_{i_1}^{(2)}+\dots + x_{i_n}^{(2)})=
n^{n-1}\cdot (x_{i_1}^{(2)}+\dots + x_{i_n}^{(2)})=x_{i_1}^{(2)}+\dots + x_{i_n}^{(2)}.
\end{align*}%
\end{proof}

\begin{lem}\label{GeneralizedMaltsev}
Suppose $f\in \mathcal C_{B\boxtimes \mathbf Z_{p}}$,
$f^{(2,2)}(x_{1}^{(2)},\dots,x_{n}^{(2)}) = x_1^{(2)}+\dots + x_n^{(2)}$,
$m\in \mathbb N$.
Then there exists  
$g\in \Clo(f)$ such that 
$g^{(2,2)}(x_1,\dots,x_{2m+1}) =  x_{1}^{(2)}-x_{2}^{(2)}+x_{3}^{(2)}-
\dots - x_{2m}^{(2)}+x_{2m+1}^{(2)}$.
\end{lem}

\begin{proof}
Consider a term defining 
$x_1-x_2+x_3-\dots+x_{2m+1}$ from 
$x_1+x_2+\dots+x_{n}$ in $\mathbf Z_{p}$.
The same term defines the required $g$ from $f$.
\end{proof}

To simplify notations we use operations coming from Lemma \ref{GeneralizedMaltsev} as follows.
Whenever we write 
$\bigoplus\limits_{i=1}^{m}(a_i\cdot(x_i\ominus x_i'))\oplus x_{m+1}$ 
for $a_{1},\dots,a_{m}\in \mathbf Z_{p}$, 
we mean 
$$g(\underbrace{x_1,x_1',\dots,x_1,x_1'}_{2a_1},\underbrace{x_2,x_2',\dots,x_2,x_2'}_{2a_2},\dots,\underbrace{x_m,x_m',\dots,x_m,x_m'}_{2a_m},x_{m+1}),$$
where
$g$ is a $(2\cdot \sum_{i=1}^{m}a_{i}+1)$-ary operation coming from 
Lemma \ref{GeneralizedMaltsev}.
Notice that we use this notation if the only important part of the 
obtained operation $f$ is $f^{(2,2)}$.

We say that an operation $f$ is \emph{symmetric on a tuple $(a_1,\dots,a_n)$}
if 
$f(a_1,\dots,a_n) = 
f(a_{\sigma(1)},\dots,a_{\sigma(n)})$ 
for any permutation $\sigma\colon [n]\to [n]$.
We say that an operation $f$ is \emph{weakly symmetric on a tuple $(a_1,\dots,a_n)$}
if 
$f(a_1,\dots,a_n) = 
f(a_{\sigma(1)},\dots,a_{\sigma(n)})$ 
for any permutation $\sigma\colon [n]\to [n]$ such that $\sigma(1) = 1$.

Suppose $P\subseteq \{(c,d)\mid c,d\in B, c\neq d\}$
and $a,b\in B$. 
We say that an $n$-ary operation $f\in \mathcal C_{B\boxtimes \mathbf Z_{p}}$ is 
\emph{$(P,a,b,k)$-symmetric} if 
    \begin{enumerate}
   	\item[(1)] $p$ divides $n-1$,
       \item[(2)] $f^{(1)}$ is XY-symmetric,
        \item[(3)] $f^{(2,2)}(x_1^{(2)},\dots,x_n^{(2)}) =  x_1^{(2)}+\dots + x_n^{(2)}$,
        \item[(4)] $f^{(2,1)}$ is weakly symmetric on all tuples $\alpha\in\{c,d\}^{n}$ 
        such that $\alpha(1) = c$ and $(c,d)\in P$,
        \item[(5)] $f^{(2,1)}$ is weakly symmetric on all tuples $\alpha\in\{a,b\}^{n}$
        such that 
        $\alpha(1) = a$ and $N_{b}(\alpha)\le k$.
    \end{enumerate}
An operation is called \emph{$P$-symmetric} if it satisfies only items (1)-(4).

\begin{lem}\label{MainIncreasingSymmetricity}
    Suppose $P\subseteq \{(c,d)\mid c,d\in B, c\neq d\}$,
    $0\le k <n-1$,     
$a,b\in B$, and $f\in \mathcal C_{B\boxtimes \mathbf Z_{p}}$ is a $(P,a,b,k)$-symmetric operation.
Then there exists a $(P,a,b,k+1)$-symmetric operation $g\in \Clo(f)$ of arity $n$.
\end{lem}

\begin{proof}
The function $f$ satisfies all the properties we require for $g$ except for the property (5) 
for $N_{b}(\alpha) = k+1$.
If $f$ is also weekly symmetric on such tuples then 
we just take $g = f$.
Otherwise, consider 
two tuples $\alpha, \beta\in\{a,b\}^{n}$ such that 
$\alpha(1) = \beta(1) = a$, 
$N_{b}(\alpha) = N_{b}(\beta) = k+1$, 
and $f^{(2,1)}(\alpha)\neq f^{(2,1)}(\beta)$.

Define a new operation 
\begin{align*}
g(y_1,&\dots,y_n) :=\\
&\bigoplus\limits_{\gamma\in T_{a,b}^{n,k+1}}
\frac{f^{(2,1)}(\alpha)-f^{(2,1)}(\gamma)}
{f^{(2,1)}(\beta)-f^{(2,1)}(\alpha)}
\cdot 
\left(t_{0,\xi(\gamma,\alpha)}^{f,\xi(\gamma,\beta)}(y_1,\dots,y_n)
\ominus
t_{0}^{f}(y_1,\dots,y_n)\right)
\oplus f(y_1,\dots,y_n)
\label{bigTransformation}
\end{align*}

Let us show that $g$ satisfies the required properties (2)-(5).
Since $f^{(1)}$ is XY-symmetric, 
by Lemmas \ref{kWNUImplieskWNU} and \ref{DifferenceTwoTwoReplacementLEM}, $g^{(1)}$ is also XY-symmetric.
Property (3) follows from Lemma \ref{TwoTwoPartLEM}.
Let us prove property (4).
Consider two tuples $\delta,\delta'\in\{c,d\}^{n}$ such that 
$\delta(1) =\delta'(1)= c$, $N_{d}(\delta) = N_{d}(\delta')$, and $(c,d)\in P$.
We need to prove that 
$g^{(2,1)}(\delta) = g^{(2,1)}(\delta')$.
Since $(c,d)\in P$, by Lemma \ref{DifferenceTwoTwoReplacementLEM}
for every $\gamma\in T_{a,b}^{n,k+1}$ we have
$t_{0,\xi(\gamma,\alpha)}^{f,\xi(\gamma,\beta)}(\delta)=
t_{0}^{f}(\delta)$ 
and 
$t_{0,\xi(\gamma,\alpha)}^{f,\xi(\gamma,\beta)}(\delta')
=t_{0}^{f}(\delta')$. 
Hence 
$g^{(2,1)}(\delta)-g^{(2,1)}(\delta') = 
f^{(2,1)}(\delta)-f^{(2,1)}(\delta')=0$.

Let us prove property (5).
Consider two tuples $\delta,\delta'\in T_{a,b}^{n,k'}$, 
where $k'\le k+1$.
We need to prove that 
$g^{(2,1)}(\delta) = g^{(2,1)}(\delta')$.
if $k'\le k$ then it follows from the property (5) for $f$ that
$t_{0,\xi(\gamma,\alpha)}^{f,\xi(\gamma,\beta)}(\delta)=
t_{0}^{f}(\delta)$ 
and 
$t_{0,\xi(\gamma,\alpha)}^{f,\xi(\gamma,\beta)}(\delta')
=t_{0}^{f}(\delta')$,
and therefore 
$g^{(2,1)}(\delta)-g^{(2,1)}(\delta') = 
f^{(2,1)}(\delta)-f^{(2,1)}(\delta')=0$.

Assume that $k' = k+1$.
Then 
$t_{0,\xi(\gamma,\alpha)}^{f,\xi(\gamma,\beta)}(\delta)
=t_{0}^{f}(\delta)$ whenever $\gamma\neq\delta$.
Therefore, in the definition of $g$ 
we only care about elements of the $\bigoplus$ corresponding to $\gamma=\delta$.
Hence, by Lemma \ref{DifferenceTwoTwoReplacementLEM} we have 
\begin{align*}
&g^{(2,1)}(\delta)-g^{(2,1)}(\delta') = \\
&\frac{f^{(2,1)}(\alpha)-f^{(2,1)}(\delta)}
{f^{(2,1)}(\beta)-f^{(2,1)}(\alpha)}
\cdot ((t_{0,\xi(\delta,\alpha)}^{f,\xi(\delta,\beta)})^{(2,1)}(\delta)-
(t_{0}^{f})^{(2,1)}(\delta))
- \\
&\frac{f^{(2,1)}(\alpha)-f^{(2,1)}(\delta')}
{f^{(2,1)}(\beta)-f^{(2,1)}(\alpha)}
\cdot 
((t_{0,\xi(\delta',\alpha)}^{f,\xi(\delta',\beta)})^{(2,1)}(\delta')-
(t_{0}^{f})^{(2,1)}(\delta'))
+f^{(2,1)}(\delta)-f^{(2,1)}(\delta')=
\\
&\frac{f^{(2,1)}(\alpha)-f^{(2,1)}(\delta)}
{f^{(2,1)}(\beta)-f^{(2,1)}(\alpha)}
\cdot (f^{(2,1)}(\beta)-f^{(2,1)}(\alpha))
- \\
&\frac{f^{(2,1)}(\alpha)-f^{(2,1)}(\delta')}
{f^{(2,1)}(\beta)-f^{(2,1)}(\alpha)}
\cdot 
(f^{(2,1)}(\beta)-f^{(2,1)}(\alpha))
+f^{(2,1)}(\delta)-f^{(2,1)}(\delta')
=\\
&f^{(2,1)}(\alpha)-f^{(2,1)}(\delta)-
f^{(2,1)}(\alpha)+f^{(2,1)}(\delta')+
f^{(2,1)}(\delta)-f^{(2,1)}(\delta')
= 0
\end{align*}
\end{proof}

\begin{cor}\label{MainIncreasingSymmetricityCor}
Suppose 
$g\in \mathcal C_{B\boxtimes \mathbf Z_{p}}$,
$g^{(1)}$ is an $n$-ary  XY-symmetric operation,
$p$ divides $n-1$, 
$g^{(2,2)}(x_1^{(2)},\dots,x_n^{(2)}) = 
    x_{1}^{(2)}+\dots+x_{n}^{(2)}$.
Then $\Clo(g)$ contains an $n$-ary operation $h$ such that 
$h^{(1)}$ is XY-symmetric, 
$h^{(2,2)}(x_1^{(2)},\dots,x_n^{(2)}) = 
    x_{1}^{(2)}+\dots+x_{n}^{(2)}$, and $h$ is weakly symmetric on 
all tuples $\alpha$ having two different elements.
\end{cor}

\begin{proof}
Notice that $g$
is $(\varnothing,a,b,0)$-symmetric 
for any $a,b\in B$.
Using Lemma \ref{MainIncreasingSymmetricity}
we can increase the set $P$ of tuples on which $g^{(2,1)}$ is weakly symmetric until we get a required operation.

Formally, we prove as follows.
Consider operations 
$f\in \Clo(g)$ 
that is 
$P$-symmetric for an inclusion maximal set $P$.
If $P$ contains all pairs then we found the required operation.
Otherwise, choose $(a,b)\notin P$. 
Choose a maximal $k$ such that there exists 
a $(P,a,b,k)$-symmeric operation $f\in\Clo(g)$.

Then applying Lemma \ref{MainIncreasingSymmetricity}
we can always increase $k$ if $k<n-1$, which contradicts the maximality of $k$.
Notice that if $k=n-2$ then 
Lemma \ref{MainIncreasingSymmetricity} 
guarantees that the pair $(a,b)$ can be included into $P$, which contradicts our assumption about the maximality of $P$.
\end{proof}

\begin{thm}\label{PropagateXYSYmmetricityToBoxtimesTHM}
Suppose 
$\mathbf A\in ( \mathbf B\boxtimes \mathbb Z_p)\cap\mathcal V_{n}$,
$w^{\mathbf B}$ is XY-symmetric.
Then there exists a term $t$ of arity $n$ such that 
$t^{\mathbf A}$ is XY-symmetric. 
\end{thm}

\begin{proof}
By the definition 
$w^{\mathbf A}\in \mathcal C_{B\boxtimes \mathbf Z_{p}}$.
By Lemma \ref{SpecialWNUPreservingDeltaLEM},
$(w^{\mathbf A})^{(2,2)}(x_1^{(2)},\dots,x_n^{(2)}) = 
    x_{1}^{(2)}+\dots+x_{n}^{(2)}$ 
    and $p$ divides $n-1$.
By Corollary \ref{MainIncreasingSymmetricityCor}, 
there exists an $n$-ary term $\tau$ such that 
$\tau^{\mathbf B}$ is XY-symmetric and 
$\tau^{\mathbf A}$ is weakly symmetric on all tuples containing 
two different elements.
By $t$ define the term  
$$w(\tau(x_1,\dots,x_n),
\tau(x_2,\dots,x_n,x_1),
\tau(x_3,\dots,x_n,x_1,x_2),
\dots,
\tau(x_n,x_1,\dots,x_{n-1})).$$
Let us show that $t$ is XY-symmetric.

Since $\tau^{\mathbf A}$ is weakly symmetric 
on any $\gamma \in \{a,b\}^{n}$,
$\tau^{\mathbf A}(\gamma)$ depends only on $\gamma(1)$
and $N_{b}(\gamma)$.
Notice that $(t^{\mathbf A})^{(1)} =  t^{\mathbf B}$ is XY-symmetric and 
$$(t^{\mathbf A})^{(2,2)}
(x_1^{(2)},\dots,x_{n}^{(2)}) = 
n\cdot (x_{1}^{(2)}+\dots+x_{n}^{(2)})
= x_{1}^{(2)}+\dots+x_{n}^{(2)}.$$
Suppose
$\alpha\in\{a,b\}^{n}$ for some $a,b\in B$.
Then 
\begin{align*}(t^{\mathbf A})^{(2,1)}
(\alpha) =
N_{a}(\alpha)\cdot (\tau^{\mathbf A})^{(2,1)}(a,\dots,a,\underbrace{b,\dots,b}_{N_{b}(\alpha)})+
N_{b}(\alpha)\cdot (\tau^{\mathbf A})^{(2,1)}&(\underbrace{b,\dots,b}_{N_{b}(\alpha)},a,\dots,a)+\\
&(w^{\mathbf A})^{(2,1)}(
\tau^{\mathbf B}(\alpha),
\dots,\tau^{\mathbf B}(\alpha)).
\end{align*}
Hence $(t^{\mathbf A})^{(2,1)}(\alpha)$
depends only on 
$N_{b}(\alpha)$, 
which means that
$(t^{\mathbf A})^{(2,1)}$
is symmetric on $\alpha$.
Therefore, 
$t^{\mathbf A}$ is XY-symmetric.
\end{proof}

\begin{THMpropagateXYSymmetricTHM}
Suppose 
$\mathbf A,\mathbf B\in \mathcal V_{n}$,
$0_{\mathbf A}$ is a perfect linear congruence,
$\mathbf A/{0_{\mathbf A}}^{*}\times \mathbf B$ has an XY-symmetric term operation of arity $n$.
Then $\mathbf A\times \mathbf B$ has an XY-symmetric term operation.
\end{THMpropagateXYSymmetricTHM}
\begin{proof}
    By Theorem \ref{ExistenceOfInjectiveHomomorphismTHM},
there exists an algebra $\mathbf C\in \mathbf ((\mathbf A/0_{\mathbf A}^{*})\boxtimes\mathbf Z_{p})\cap\mathcal V_{n}$ such that 
there exists an injective homomorphism $h\colon \mathbf A\to \mathbf C$. 
Let $\delta$ be the canonical congruence on $\mathbf C$ such that 
$\mathbf C/\delta\cong \mathbf A/0_{\mathbf A}^{*}$.
Put $\mathbf D = \mathbf C \times \mathbf B$
and extend the congruence $\delta$ on $\mathbf D$.
Then 
$\mathbf D\in ((\mathbf D/\delta)\boxtimes \mathbf Z_{p})
\cap \mathcal V_{n}$.
By Theorem \ref{PropagateXYSYmmetricityToBoxtimesTHM}
there exists a term $\tau$ such that 
$\tau^{\mathbf D}$ is XY-symmetric.
Hence, both $\tau^{\mathbf C}$ and $\tau^{\mathbf B}$ are XY-symmetric.
Since $\mathbf A$ is isomorphic to a subalgebra of 
$\mathbf C$, $\tau^{\mathbf A}$ is also XY-symmetric.
\end{proof}











\section{Proof of the properties of strong subuniverses}
\label{SectionProofStrongSubalebras}

In the section we prove all the statements formulated in Section \ref{SectionStrongSubalgebras} and this is the most technical part of the paper.
We start with a few additional notations, then we formulate necessary properties of binary absorbing and central 
subuniverses that are mostly taken from 
\cite{zhuk2021strong}.
In Subsection \ref{SUBSECTIONIntersectionProperty} we show 
that intersection of strong subalgebras 
behaves well, which is one of the main properties of strong subalgebras and definitely 
the most difficult to prove.
In the next subsection we show the properties 
of a bridge connecting linear and PC congruences.
For instance, there we prove that 
Linear and PC congruences can never be connected by a bridge and bridges for the PC congruences 
are trivial.
In the next subsections we show that we should intersect strong subuniverses of the same type to obtain an empty set, and prove that 
factorization of strong subalgebras modulo a congruence behaves well.
Finally, in Subsection \ref{SUBSECTIONProofOfTheRemainingProperties} 
we prove most of the statements formulated in 
Section \ref{SUBSECTIONSTRONGSUBALGEBRASPROPERTIES}.

\subsection{Additional definitions}

In this section we call a relation \emph{symmetric} 
if any permutation of its variables gives the same relation.
For a relation $R\le 
\mathbf A_{1}\times\dots\mathbf A_{n}$ 
by $\LeftLinked(R)$ we 
denote 
the minimal equivalence relation on 
$\proj_{1}(R)$
such that 
$(a_1,a_2,\dots,a_n),(b_1,a_2,\dots,a_n)\in R$ implies 
$(a_1,b_1)\in \LeftLinked(R)$.
Similarly, $\RightLinked(R)$ is  
the minimal equivalence relation on 
$\proj_{n}(R)$
such that 
$(a_1,\dots,a_{n-1},a_n),(a_1,\dots,a_{n-1},b_n)\in R$ implies 
$(a_n,b_n)\in \RightLinked(R)$.

A relation $R\lneq A\times B$ is called \emph{central} if 
there exists $b\in B$ such that $A\times\{b\}\subseteq R$.

\subsection{Subuniverses of types $\TBA,\TC,\TS$}

Here we formulate some properties of strong subuniverse 
that we will use later.


\begin{lem}[\cite{zhuk2021strong}, Lemma 6.25]\label{LEMBACenterSPossibleIntersections}
Suppose $B<_{T_1}A$, $C<_{T_2}A$, 
$B\cap C=\varnothing$, 
$T_{1},T_{2}\in \{\TBA,\TC,\TS\}$.
Then $T_{1}=T_{2} \in\{\TBA,\TC\}$.
\end{lem}

\begin{lem}[\cite{DecidingAbsorption}, Lemma 2.9, \cite{zhuk2021strong}, Lemma 6.1, Theorem 6.9]\label{LEMBACenterSImplyPPDefinition}
Suppose $R\le A^{n}$ is defined by a pp-formula $\Phi$ containing a relation $S$
and $\Phi'$ is obtained by $\Phi$ by 
replacement 
of each appearance of $S$ by 
$S'<_{T} S$, where $T\in \{\TBA,\TC\}$.
Then $\Phi'$ defines a relation 
$R'$ such that 
$R'\le_{T}R$.
\end{lem}

The above lemma implies an easier claim.

\begin{lem}\label{LEMBACenterImplyIntersection}
Suppose $B\le_{T} \mathbf A$ and $C\le \mathbf A$, where 
$T\in\{\TBA,\TC\}$.
Then $B\cap C \le_{T} C$.
\end{lem}

\begin{lem}\label{LEMStrongNonemptyIntersection}
Suppose $C\lll^{A} B\lll A$
and $D<_{\TBA,\TC} B$.
Then $C\cap D\neq \varnothing$.
\end{lem}

\begin{proof}
Assume the converse.
Consider a minimal 
$C''$ such that 
$C\lll^{A} C'<_{T(\sigma)}^{A} C''\lll^{A} B$
and $C''\cap D\neq \varnothing$.
By Lemma \ref{LEMBACenterImplyIntersection}
$C''\cap D<_{\TBA,\TC} C''$, which 
by Lemma \ref{LEMBACenterSPossibleIntersections} implies that 
$T = \TD$.
By Lemma \ref{LEMBACenterSImplyFactor} 
we have 
$(C''\cap D)/\sigma<_{\TBA} C''/\sigma$, 
which contradicts the definition of a dividing subuniverse.
\end{proof}

\begin{lem}\label{LEMBACenterSImplyFactor}
Suppose $B\le_{T} A$
and $\sigma$ is a congruence on $A$, where 
$T\in\{\TBA,\TC,\TS\}$.
Then $B/\sigma \le_{T} A/\sigma$.
\end{lem}

\begin{proof}
For $T = \TBA$ it is straightforward, 
for $T=\TC$ see Lemma 6.8 in \cite{zhuk2021strong},
for $T = \TS$ it is just a combination of 
the results for $\TBA$ and $\TC$. 
\end{proof}


\begin{lem}[\cite{DecidingAbsorption}, Proposition 2.14, \cite{zhuk2021strong}, Lemma 3.2]\label{LemAbsorptionImpliesEssential}
Suppose $B\le \mathbf A$, $n\ge 2$.
Then $B$ is an absorbing subuniverse with an operation of arity $n$ 
if and only if there does not exist 
$S\le A^{n}$ such that 
$S\cap B^{n} = \varnothing$ 
and $S\cap (B^{i-1}\times A\times B^{n-i})\neq \varnothing$
for every $i\in\{1,2,\dots,n\}$.
\end{lem}

\begin{lem}\label{LEMBACenterSOnPowerImplies}
    Suppose 
    $B<_{T} A^{n}$, where $T\in \{\TBA,\TC,\TS\}$.
Then there exists $C<_{T} A$.
\end{lem}
\begin{proof}
For $T\in\{\TBA,\TC\}$ see Lemma 6.24 
in \cite{zhuk2021strong}. 
For $T = \TS$ just repeat the same proof word to word 
replacing $\TBA$ by $\TS$.
\end{proof}


\begin{lem}[\cite{barto2012absorbing}, Proposition 2.15 (i)]\label{LEMBACenterLinkedness}
Suppose 
$R\le_{sd} A_1\times A_2$, 
$B_{1}$ and $B_{2}$ are absorbing subuniverses on $\mathbf A_1$ and 
$\mathbf A_{2}$, respectively, 
$(R\cap (B_{1}\times B_{2}))\le_{sd} B_1\times B_2$,
$R$ is linked.
Then 
$(R\cap (B_{1}\times B_{2}))$ is linked.
\end{lem}

\begin{lem}[\cite{zhuk2021strong}, Theorem 6.15] \label{LEMCentralRelationImplies}
Suppose 
$\mathbf R\le_{sd} \mathbf A\times \mathbf B$,
$C = \{c\in A\mid \forall b\in B\colon (c,b)\in R\}$.
Then one of the following conditions holds:
\begin{enumerate}
    \item $C$ is a central subuniverse of $\mathbf A$;
    \item $\mathbf B$ has a nontrivial binary absorbing subuniverse.
\end{enumerate}
\end{lem}

\begin{lem}[\cite{ZebsNotes}, Theorem 3.11.1]\label{LEMLinkedImpliesBACenter}
Suppose 
$R\lneq_{sd} \mathbf A\times \mathbf B$,
$R$ is linked.
Then there exists a BA or central subuniverse on $\mathbf A$ or $\mathbf B$.
\end{lem}

\begin{lem}[\cite{zhuk2020proof}, Lemma 7.2]\label{LEMAbsorbingEquality}
Suppose 
$0_{\mathbf A}\subseteq \sigma\le {\mathbf A}^{2}$
and 
$\omega<_{BA}\sigma$.
Then $\omega\cap 0_{\mathbf A} \neq \varnothing$.
\end{lem}

The following Lemma can be derived from 
Lemma 3.11.2 and 3.11.3 from \cite{ZebsNotes}, but we will give a 
separate proof.

\begin{lem}\label{LEMBAConLeftOrCenterOnRight}
Suppose $R\le_{sd} A\times B$, 
$A$ is BA and center free,
$\LeftLinked(R)= A^{2}$, and 
$C = \{c\in B\mid A\times\{c\}\subseteq R\}$.
Then $C\neq\varnothing$
and $C\le_{\TBA,\TC}B$.
\end{lem}

\begin{proof}
First, let us show that $C\neq\varnothing$. 
For every $n\ge 2$ put
$W_{n} = 
\{(a_1,\dots,a_{n})\mid \exists b \forall i \colon R(a_{i},b)\}$.
If $W_{|A|}=A^{|A|}$, then $C\neq\varnothing$.
Otherwise, choose the minimal $n\ge 2$ such that $W_{n}\neq A^{n}$.
Since $\LeftLinked(R) = A^2$, 
$\LeftLinked(W_{n}) = A^2$.
Looking at $W_{n}$ as at binary relation 
$W_{n}\le A\times A^{n-1}$ 
and using Lemmas \ref{LEMLinkedImpliesBACenter}
and \ref{LEMBACenterSOnPowerImplies} we derive the existence of 
BA or central subuniverse on $A$, which contradicts our assumptions.

Thus, $C\neq \varnothing$.
By Lemma \ref{LEMCentralRelationImplies}, $C$ is a central subuniverse of $B$.
It remains to show that it is also a BA subuniverse.
Assume the converse, then 
by Lemma \ref{LemAbsorptionImpliesEssential}
there exists a relation 
$S\le B\times B$ such that 
$S\cap (C\times C) = \varnothing$, 
$S\cap (B\times C) \neq \varnothing$,
and 
$S\cap (C\times B) \neq \varnothing$. 

Put 
$W = \{(a_1,\dots,a_{|A|})\mid
\exists (b,c)\in S\colon c\in C\wedge \forall i (a_{i},b)\}$.
By Lemma \ref{LEMBACenterSImplyPPDefinition}, 
$W<_{\TC} A^{|A|}$, 
which by Lemma \ref{LEMBACenterSOnPowerImplies}
implies the existence of a central subuniverse on $A$ and contradicts our assumptions.
\end{proof}



\begin{lem}\label{LEMReverseHomomorphism}
Suppose $f\colon \mathbf A\to \mathbf A'$ is a surjective homomorphism
and $T\in\{\TBA,\TC,\TS,\TL,\TD\}$. Then 
$C'<_{T(\sigma)}^{A}B'\Rightarrow f^{-1}(C')<_{T(f^{-1}(\sigma))}^{A}f^{-1}(B')$.
\end{lem}

\begin{proof}
For $T\in\{\TBA,\TC\}$ it follows from the properties of 
a homomorphism
(see Section 3.15 in \cite{ZebsNotes}). 
For $T=\TS$ there exists $D'\le C$ such that 
$D'<_{\TBA,\TC} B'$.
Then 
$f^{-1}(D')<_{BA,\TC} f^{-1}(B')$, hence 
$f^{-1}(C')<_{S} f^{-1}(B')$.

Suppose $T\in\{\TPC,\TL\}$.
Let $\delta  = f^{-1}(\sigma)$, that is 
$\delta = \{(a,b)\mid (f(a),f(b))\in \sigma\}$.
Then 
$\mathbf A'/\sigma = \mathbf A/\delta$,
$\mathbf B'/\sigma\cong \mathbf B/\delta$ and
$\mathbf C'/\sigma\cong \mathbf C/\delta$, 
which implies the required properties of 
a PC/linear subuniverse.
\end{proof}

\begin{cor}\label{CORReverseHomomorphism}
Suppose 
$R\le_{sd} A_{1}\times\dots\times A_{n}$,
$C_{1}<_{T(\sigma)}^{A_{1}}B_{1}\le A_1$
where 
$T\in\{\TBA,\TC,\TS,\TL,\TD\}$.
Then 
$R\cap (C_{1}\times A_{2}\times\dots\times A_{n})<_{T(\sigma)}^{R}
R\cap (B_{1}\times A_{2}\times\dots\times A_{n})$.
\end{cor}

\begin{proof}
It is sufficient to consider a homomorphism 
$f_{1}\colon \mathbf R \to \mathbf A_{1}$ sending every tuple 
to its first coordinate and apply Lemma \ref{LEMReverseHomomorphism}
to $f_{1}^{-1}$.
\end{proof}

\subsection{Intersection property}\label{SUBSECTIONIntersectionProperty}

In this subsection we prove a fundamental property
of our subuniverses. Precisely, we will show 
(Lemma \ref{LEMIntersectionPCLinearIsGood}) 
that 
if $C\lll A$ and 
$\delta$ is a dividing congruence for $B\lll A$, 
then $(B\cap C)/\delta$ is either empty, 
or of size 1, or equal to $B/\delta$.
\begin{lem}\label{LEMTotallySymmetricWithoutBACenter}
Suppose 
$R\le \mathbf A^{n}$ is symmetric, 
$\mathbf A$ is BA and center free,
$\proj_{1,2}(R) = A^{2}$, and 
$(a_1,\dots,a_n)\in R\Rightarrow 
(a_1,a_1,a_2\dots,a_{n-1})\in R$.
Then 
$R = A^{n}$.
\end{lem}

\begin{proof}
Let us prove the claim by induction on $n$. 
For $n=2$ it follows from the condition  
$\proj_{1,2}(R) = A^{2}$.
Suppose $n>2$. 
Since $\proj_{1,n}(R) =A^{2}$ 
for any $a,c\in A$ 
we have 
$(a,b_2,\dots,b_{n-1},c)\in R$ for some $b_2,\dots,b_{n-1}\in A$.
By the conditions of the lemma we have 
$(a,\dots,a,c),(a,\dots,a,a)\in R$.
Since $a$ and $c$ can be chosen arbitrary, 
we have $\RightLinked(R)=A^{2}$.
By the inductive assumption 
$\proj_{1,2,\dots,n-1}(R) = A^{n-1}$.
Then Lemma \ref{LEMLinkedImpliesBACenter} implies 
the existence of BA or central subuniverse on $A$
or $A^{n-1}$, which together with Lemma \ref{LEMBACenterSOnPowerImplies} contradicts the fact that 
$\mathbf A$ is BA and center free.
\end{proof}

\begin{lem}\label{LEMTotallySymmetricRelationForIrreducible}
Suppose $\sigma$ is a dividing congruence for 
$B\le A$,
$R\le (\mathbf A/\sigma)^{n}$, $R$ is reflexive, symmetric, and 
$(a_1,\dots,a_n)\in R\Rightarrow
(a_1,a_1,a_2\dots,a_{n-1})\in R$.
Then 
either 
$(B/\sigma)^n\subseteq R$, or 
$R = \{(a/\sigma,a/\sigma,\dots,a/\sigma)\mid a\in A\}$.
\end{lem}

\begin{proof}
Consider two cases:

Case 1. $\proj_{1,2}(R)$ is the equality relation.
Since $R$ is symmetric, we derive that 
$R = \{(a/\sigma,a/\sigma,\dots,a/\sigma)\mid a\in A\}$.

Case 2. $\proj_{1,2}(R)$ is not the equality relation.
Since $\sigma$ is irreducible and $R$ is reflexive, 
we derive that $\proj_{1,2}(R)\supseteq \sigma^{*}/\sigma \supseteq 
(B/\sigma)^{2}$.
It remains to apply Lemma \ref{LEMTotallySymmetricWithoutBACenter}
to $R\cap (B/\sigma)^{n}$ to 
show that $(B/\sigma)^n\subseteq R$.
\end{proof}

\begin{lem}\label{LEMSelfIntersectionPC}
Suppose 
$B_{k}<_{T_{k}(\sigma_{k})}^{A} B_{k-1}<_{T_{k-1}(\sigma_{k-1})}^{A} \dots <_{T_{2}(\sigma_{2})}^{A}B_{1}<_{T_{1}(\sigma_{1})}^{A} B_{0} = A$, 
$\delta$ is a congruence on $A$,
$T_{1},T_{2},\dots,T_{k}\in\{\TBA,\TC,\TS,\TD\}$,
$m\in\{1,2,\dots,k\}$, 
and $T_{m} = \TD$.
Then 
$((B_{k}\circ \delta)\cap B_{m-1})/\sigma_{m}$ is either 
$B_{m-1}/\sigma_{m}$, or 
$B_{m}/\sigma_{m}$.
Additionally, 
if 
$((B_{k}\circ \delta)\cap B_{m-1})/\sigma_{m}=
B_{m}/\sigma_{m}$ then 
$\sigma_{m}\supseteq (\delta\cap \sigma_1\cap\dots\cap \sigma_{m-1})$.
\end{lem}
\begin{proof}
We prove the statement by induction on 
$k$. 
Consider two cases:

Case 1. $k=m$.
Define $|A|$-ary relation 
$S_{0}\le (A/\sigma_{k})^{|A|}$ by 
$$S_{0} = \{(a_1/\sigma_{k},\dots,a_{|A|}/\sigma_{k})\mid 
\forall i,j\colon 
(a_{i},a_{j})\in \delta\cap \bigcap_{\ell=1}^{k-1}\sigma_{\ell}
\}.$$
By Lemma \ref{LEMTotallySymmetricRelationForIrreducible} we have one of the two subcases:

Subcase 1A. $S_{0} = \{(b/\sigma_{k},\dots,b/\sigma_{k})\mid b\in A\}.$ 
Then 
$\sigma_{k}\supseteq (\delta\cap \sigma_1\cap\dots\cap \sigma_{k-1})$, 
which is the additional condition we needed to prove.
Therefore, $\delta\cap B_{k-1}^{2} \subseteq  \sigma_{k}\cap B_{k-1}^{2}$, $(B_{k}\circ\delta) \cap B_{k-1} \subseteq  
(B_{k}\circ\sigma_{k}) \cap B_{k-1} =B_{k}$, 
and $((B_{k}\circ\delta) \cap B_{k-1})/\sigma_{k}=B_{k}/\sigma_{k}$.

Subcase 1B. $(B_{k-1}/\sigma_{k})^{|A|}\subseteq S_{0}.$ 
For $n\in\{0,1,\dots,k\}$
define $S_{n}$ by 
$$S_{n} = \{(a_1/\sigma_{k},\dots,a_{|A|}/\sigma_{k})\mid 
\forall i,j\colon 
(a_{i},a_{j})\in \delta\cap \bigcap_{\ell=1}^{k-1}\sigma_{\ell},
\forall i\colon a_{i}\in B_{n}
\}.$$
Notice that for $n = 0$ we get a definition of $S_{0}$. 
Assume that $(B_{k-1}/\sigma_{k})^{|A|} \subseteq S_{k-1}$.
Then there exists an equivalence class $E$ of 
$\delta$ such that $(E\cap B_{k-1})/\sigma_{k} = B_{k-1}/\sigma_{k-1}$.
Hence, $((B_{k}\circ \delta)\cap B_{k-1})/\sigma_{k} = B_{k-1}/\sigma_{k}$, which completes this case.
Thus, we assume that $(B_{k-1}/\sigma_{k})^{|A|} \not\subseteq S_{k-1}$.
Consider the minimal $n$ 
such that 
$(B_{k-1}/\sigma_{k})^{|A|}\not\subseteq S_{n}.$
Consider three subsubcases.

Subsubcase 1B1. $T_{n}\in\{\TBA,\TC\}$.
Combining Lemmas \ref{LEMBACenterSImplyPPDefinition} 
and \ref{LEMBACenterSImplyFactor}, 
we obtain that 
$S_{n}\cap (B_{k-1}/\sigma_{k})^{|A|} <_{T_{n}}(B_{k-1}/\sigma_{k})^{|A|}$, 
which by Lemma \ref{LEMBACenterSOnPowerImplies} implies
the existence of BA or central subuniverse on $B_{k-1}/\sigma_{k}$ and contradicts 
 the definition of  a dividing congruence.

Subsubcase 1B2. $T_{n}=\TS$.
Choose $G\le B_{n}$ such that 
$G<_{\TBA,\TC} B_{n-1}$.
By Lemma \ref{LEMStrongNonemptyIntersection} $G\cap B_{k}\neq \varnothing$, hence even 
if we replace $B_{n}$ by $G$ in the definition of $S_{n}$ 
the intersection $S_{n}\cap (B_{k-1}/\sigma_{k})^{|A|}$
will be nonempty.
Then Lemmas \ref{LEMBACenterSImplyPPDefinition} 
and \ref{LEMBACenterSImplyFactor} imply the  
existence of both BA and central subuniverse on $B_{k-1}/\sigma_{k}$ and contradicts  the definition of  a dividing congruence.


Subsubcase 1B3. $T_{n}=\TD$.
Define a new relation 
$R\le (A/\sigma_{k})^{|A|} \times A/\sigma_{n}$
by
$$R = \{(a_1/\sigma_{k},\dots,a_{|A|}/\sigma_{k},b/\sigma_{n})\mid 
\forall i,j\colon 
(a_{i},a_{j})\in \delta\cap \bigcap_{\ell=1}^{k-1}\sigma_{\ell},
\forall i\colon a_{i}\in B_{n-1}, (a_{i},b)\in \sigma_{n}
\}.$$
Put 
$R' = R\cap ((B_{k-1}/\sigma_{k})^{|A|}\times 
B_{n-1}/\sigma_{n})$.
By the choice of $n$ we have 
$\proj_{1,2,\dots,|A|}(R') = (B_{k-1}/\sigma_{k})^{|A|}$
and 
$(B_{k-1}/\sigma_{k})^{|A|}\times B_{n}/\sigma_{n}\not\subseteq R'$.
Notice that 
$\proj_{|A|+1}(R') = ((B_{k-1}\circ \sigma_{k})\cap B_{n-1})/\sigma_{n}$.
By the inductive assumption either
$\proj_{|A|+1}(R') = B_{n-1}/\sigma_{n}$, or 
$\proj_{|A|+1}(R') = B_{n}/\sigma_{n}$.
The second case contradicts the above conditions, 
therefore 
$\proj_{|A|+1}(R') = B_{n-1}/\sigma_{n}$.

Since $\proj_{1,2,\dots,|A|}(R')= (B_{k-1}/\sigma_{k})^{|A|}$,
there exists $d\in B_{n-1}/\sigma_{n}$ such that 
$(B_{k-1}/\sigma_{k})^{|A|}\times \{d\}\subseteq R'$.
Then $R'$ can be viewed as 
a binary relation 
with a center containing $d$.
Then Lemma \ref{LEMCentralRelationImplies}
implies the existence of a BA subuniverse on $B_{k-1}/\sigma_{k}$
or a center on $B_{n-1}/\sigma_{n}$, which contradicts the definition of a dividing subuniverse.

Case 2. $k>m$. 
By the inductive assumption 
$((B_{k-1}\circ \delta)\cap B_{m-1})/\sigma_{m}$ is either 
$B_{m-1}/\sigma_{m}$, or $B_{m}/\sigma_{m}$.
In the second case
we also have 
$((B_{k}\circ \delta)\cap B_{m-1})/\sigma_{m} = B_{m}/\sigma_{m}$, which completes this case.
Thus, we assume that 
$((B_{k-1}\circ \delta)\cap B_{m-1})/\sigma_{m}=B_{m-1}/\sigma_{m}$.
Let us consider three cases depending on the type of the reduction $T_{k}$.

Subcase 2A. $T_{k}\in\{\TBA,\TC\}$.
Combining Lemmas \ref{LEMBACenterSImplyPPDefinition}
and \ref{LEMBACenterSImplyFactor} 
we obtain 
$((B_{k}\circ \delta)\cap B_{m-1})/\sigma_{m}\le_{T_{k}}
((B_{k-1}\circ \delta)\cap B_{m-1})/\sigma_{m}$, 
which by the definition of 
a dividing congruence implies 
$((B_{k}\circ \delta)\cap B_{m-1})/\sigma_{m}=
((B_{k-1}\circ \delta)\cap B_{m-1})/\sigma_{m}=B_{m-1}/\sigma_{m}$ and completes this case.

Subcase 2B. $T_{k}=\TS$.
Choose $G\le B_{k}$ such that 
$G<_{\TBA,\TC} B_{k-1}$.
Then 
$((G\circ \delta)\cap B_{m-1})/\sigma_{m}\le_{\TBA,\TC}
((B_{k-1}\circ \delta)\cap B_{m-1})/\sigma_{m}$, 
which by the definition of 
a dividing congruence implies 
$((B_{k}\circ \delta)\cap B_{m-1})/\sigma_{m}=
((B_{k-1}\circ \delta)\cap B_{m-1})/\sigma_{m}=B_{m-1}/\sigma_{m}$ and completes this case.

Subcase 2C. $T_{k}=\TD$.
Define a binary relation 
$R\le B_{k-1}/\sigma_{k}\times 
B_{m-1}/\sigma_{m}$ by 
$$R = \{(a/\sigma_{k},b/\sigma_{m})\mid 
a\in B_{k-1},b\in B_{m-1},
(a,b)\in\delta\}.$$
By the above assumption 
$R$ is subdirect.
Since $k>m$,
we have 
$B_{k-1}/\sigma_{k}\times B_{m}/\sigma_{m}\subseteq R$,
hence the right-hand side of $R$ has a center.
By the definition of a dividing congruence 
both 
$B_{k-1}/\sigma_{k}$ and $B_{m-1}/\sigma_{m}$
are BA and center-free. Then 
Lemma \ref{LEMCentralRelationImplies}
implies that 
$R =B_{k-1}/\sigma_{k}\times B_{m-1}/\sigma_{m}$
and $((B_{k}\circ \delta)\cap B_{m-1})/\sigma_{m} = 
B_{m-1}/\sigma_{m}$.
\end{proof}


\begin{lem}\label{LEMIntersectionPCLinearIsGood}
Suppose 
$B\lll A$, $C\lll A$, $B\cap C\neq\varnothing$.
Then 
\begin{enumerate}
    \item[(d)] if $\delta$ is a dividing congruence for $B\le A$,
    then $|(B\cap C)/\delta|=1$ or 
$(B\cap C)/\delta = B/\delta$.
Moreover, if 
$|(B\cap C)/\delta|=1$ then 
$\delta\supseteq \delta_{1}\cap \dots\cap \delta_{s}$,
where 
$\delta_{1},\dots,\delta_{s}$ are the dividing
congruences
from the definition of 
$B\lll A$ and $C\lll A$.
\item[(s)] if $G<_{\TBA,\TC}B$, then 
$G\cap C\neq \varnothing$.
\end{enumerate}
\end{lem}
\begin{proof}
Suppose 
\begin{align*}
    B=B_{k}<_{T_{k}(\sigma_{k})}^{A} B_{k-1}<_{T_{k-1}(\sigma_{k-1})}^{A} \dots <_{T_{2}(\sigma_{2})}^{A}B_{1}
<_{T_{1}(\sigma_{1})}^{A} B_{0} = A \\
C=C_{\ell}<_{\mathcal T_{\ell}(\omega_{\ell})}^{A} C_{\ell-1}<_{\mathcal T_{\ell-1}(\omega_{\ell-1})}^{A} \dots <_{\mathcal T_{2}(\omega_{2})}^{A}C_{1}
<_{\mathcal T_{1}(\omega_{1})}^{A} C_{0} =  A
\end{align*}
where $k$ and $\ell$ were chosen minimal.
Put $\sigma = \bigcap_{i=1}^{k} \sigma_{i}$ 
and $\omega = \bigcap_{i=1}^{\ell} \omega_{i}$.
We prove the lemma by induction on 
$k+\ell$.
 




Let us prove (s) first.
By the inductive assumption 
$G\cap C_{\ell-1}\neq\varnothing$.
Let us consider the type $\mathcal T_{\ell}$ 
of the subuniverse $C_{\ell}$.

$\mathcal T_{\ell}=\TD$. 
By the inductive assumption  
$(B\cap C_{\ell-1})/\omega_{\ell}$
is either of size 1, or equal to 
$C_{\ell-1}/\omega_{\ell}$.
In the first case we have 
$G\cap C_{\ell}=G\cap B\cap C_{\ell}=
G\cap B\cap C_{\ell-1}=G\cap C_{\ell-1}\neq\varnothing$.
In the second case 
$(G\cap C_{\ell-1})/\omega_{\ell}<_{\TBA,\TC}
C_{\ell-1}/\omega_{\ell}$, 
which contradicts the definition of 
a divisible congruence $\omega_{\ell}$.

$\mathcal T_{\ell}\in\{\TBA,\TC\}$. 
By Lemma \ref{LEMBACenterImplyIntersection} 
$G\cap C_{\ell-1}<_{\TBA,\TC} B_{k}\cap C_{\ell-1}$,
$B_{k}\cap C_{\ell}<_{\mathcal T_{\ell}} B_{k}\cap C_{\ell-1}$,
and by Lemma \ref{LEMBACenterSPossibleIntersections} 
$G\cap C_{\ell}\neq \varnothing$.

$\mathcal T_{\ell}=\TS$.
Let $G'\le C_{\ell}$ and 
$G'<_{\TBA,\TC} C_{\ell-1}$.
By the inductive assumption 
$B_{k}\cap G'\neq \varnothing$.
Then by Lemma \ref{LEMBACenterImplyIntersection} 
$G\cap C_{\ell-1}<_{\TBA,\TC} B_{k}\cap C_{\ell-1}$,
$B_{k}\cap G'<_{\TBA,\TC} B_{k}\cap C_{\ell-1}$,
and by Lemma \ref{LEMBACenterSPossibleIntersections} 
$G\cap G'\neq \varnothing$.
Hence $G\cap C_{\ell}\neq\varnothing$.

Let us prove (d).
Define an $|A|$-ary relation 
$S\le (A/\delta)^{|A|}$ by 
$$S = \{(a_1/\delta,\dots,a_{|A|}/\delta)\mid 
\forall i,j\colon 
(a_{i},a_{j})\in \sigma\cap \omega
\}.$$
By Lemma \ref{LEMTotallySymmetricRelationForIrreducible}
we have one of the two cases.

Case 1. $S= \{(b/\delta,\dots,b/\delta)\mid b\in A\}.$ 
Therefore, $\sigma\cap\omega\subseteq \delta$, which is the 
additional property we need.
Notice that 
$(a,b)\in\sigma\cap \omega$ for any $a,b\in B\cap C$, 
hence 
$a/\delta = b/\delta$ for any $a,b\in B\cap C$, 
and $|(B\cap C)/\delta|=1$, which completes this case.


Case 2. $(B/\delta)^{|A|}\subseteq S.$ 
For $m\in\{0,1,\dots,k\}$ and $n\in\{0,1,\dots,\ell\}$
put 
$$S_{m,n} = \{(a_1/\delta,\dots,a_{|A|}/\delta)\mid 
\forall i,j\colon 
(a_{i},a_{j})\in \sigma\cap \omega,
\forall i\colon a_{i}\in B_{m}, a_i\in C_{n}
\}.$$
Notice that $S_{0,0} = S$.
Consider two subcases.

Subcase 2A. $(B/\delta)^{|A|}\not \subseteq S_{k,0}$.
Choose a minimal $m$ such that 
$(B/\delta)^{|A|}\not \subseteq S_{m,0}$.
Consider 3 subsubcases depending on the type  $T_{m}$.

Subsubcase 2A1.
$T_{m}\in\{\TBA,\TC\}$.
Combining Lemmas \ref{LEMBACenterSImplyPPDefinition} 
and \ref{LEMBACenterSImplyFactor}, 
we obtain 
$S_{m,0}\cap (B/\delta)^{|A|} <_{T_{m}}(B/\delta)^{|A|}$, 
which by Lemma \ref{LEMBACenterSOnPowerImplies} implies
the existence of BA or central subuniverse on $B/\delta$ and contradicts 
the definition of a dividing congruence.

Subsubcase 2A2.
$T_{m}=\TS$.
Choose $G\le B_{m}$ such that 
$G<_{\TBA,\TC} B_{m-1}$.
By Lemma \ref{LEMStrongNonemptyIntersection} $G\cap B\neq \varnothing$, hence even 
if we replace $B_{m}$ by $G$ in the definition of $S_{m,0}$ 
the intersection $S_{m,0}\cap (B/\delta)^{|A|}$
will be nonempty.
Then Lemmas \ref{LEMBACenterSImplyPPDefinition} 
and \ref{LEMBACenterSImplyFactor} imply the  
existence of both BA and central subuniverse on $B/\delta$ and contradicts  the definition of  a dividing congruence.

Subsubcase 2A3.
$T_{m}=\TD$.
Define a new relation 
$R\le (A/\delta)^{|A|} \times A/\sigma_{m}$
by
$$R = \{(a_1/\delta,\dots,a_{|A|}/\delta,b/\sigma_{m})\mid 
\forall i,j\colon 
(a_{i},a_{j})\in \sigma\cap \omega,
\forall i\colon a_{i}\in B_{m-1},
(a_{i},b)\in \sigma_{m}
\}.$$
Put 
$R' = R\cap ((B/\delta)^{|A|}\times B_{m-1}/\sigma_{m})$.
By the choice of $m$
we have 
$\proj_{1,2,\dots,|A|}(R') = (B/\delta)^{|A|}$
and $(B/\delta)^{|A|}\times B_{m}/\sigma_{m}\not\subseteq R'$.
Also, 
$\proj_{|A|+1}(R') = 
((B\circ\delta)\cap B_{m-1})/\sigma_{m}$.
Lemma \ref{LEMSelfIntersectionPC} 
implies that 
$((B\circ\delta)\cap B_{m-1})/\sigma_{m}=B_{m-1}/\sigma_{m}$ 
Since $\proj_{1,2,\dots,|A|}(R') = (B/\delta)^{|A|}$,  there exists $d\in B_{m-1}/\sigma_{m}$ such that 
$(B/\delta)^{|A|}\times \{d\}\subseteq R'$.
Then $R'$ can be viewed as 
a binary relation 
with a center containing $d$.
Then Lemma \ref{LEMCentralRelationImplies}
implies the existence of a BA subuniverse on $B/\delta$
or a center on $B_{m-1}/\sigma_{m}$, which contradicts the definition of a dividing subuniverse.

Subcase 2B. $(B/\delta)^{|A|}\subseteq S_{k,0}$
and $(B/\delta)^{|A|}\not \subseteq S_{k,\ell}$.
Choose a minimal $n$ such that 
$(B/\delta)^{|A|}\not \subseteq S_{k,n}$.
Consider 3 subsubcases depending on the type  $\mathcal T_{n}$.

Subsubcase 2B1.
$\mathcal T_{n}\in\{\TBA,\TC\}$.
Combining Lemmas \ref{LEMBACenterSImplyPPDefinition} 
and \ref{LEMBACenterSImplyFactor}, 
we obtain 
$S_{k,n}\cap (B/\delta)^{|A|} <_{\mathcal T_{n}}(B/\delta)^{|A|}$, 
which by Lemma \ref{LEMBACenterSOnPowerImplies} implies
the existence of BA or central subuniverse on $B/\delta$ and contradicts 
the definition of a dividing congruence.

Subsubcase 2B2.
$\mathcal T_{n}=\TS$.
Choose 
$G\le C_{n}$ such that 
$G<_{\TBA,\TC} C_{n-1}$.
By the inductive assumption 
$G\cap B_{k}\neq\varnothing$.
Then even if we replace $C_{n}$ by $G$ in the definition of 
$S_{m,n}$ we get a nonempty intersection with 
$(B/\delta)^{|A|}$.
By 
Lemmas \ref{LEMBACenterSImplyPPDefinition} 
and \ref{LEMBACenterSImplyFactor} 
we obtain 
both BA and central subuniverse on 
$(B/\delta)^{|A|}$, 
which by Lemma \ref{LEMBACenterSOnPowerImplies} implies
the existence of BA and central subuniverse on $B/\delta$ and 
contradicts 
the definition of a dividing congruence.

Subsubcase 2B3.
$\mathcal T_{n}=\TD$. 
Define a new relation 
$R\le (A/\delta)^{|A|} \times A/\omega_{n}$
by
$$R = \{(a_1/\delta,\dots,a_{|A|}/\delta,b/\omega_{n})\mid 
\forall i,j\colon 
(a_{i},a_{j})\in \sigma\cap \omega,
\forall i\colon a_{i}\in B_{k}, a_{i}\in C_{n-1},
(a_{i},b)\in \omega_{n}
\}.$$
Put 
$R' = R\cap ((B/\delta)^{|A|}\times C_{n-1}/\omega_{n})$.
By the choice of $n$
we have 
$\proj_{1,2,\dots,|A|}(R') = (B/\delta)^{|A|}$
and $(B/\delta)^{|A|}\times C_{n}/\omega_{n}\not\subseteq R'$.
Also, 
$\proj_{|A|+1}(R') = 
(B_{k}\cap C_{n-1})/\omega_{n}$.
By the inductive assumption 
either 
$(B_{k}\cap C_{n-1})/\omega_{n} = C_{n-1}/\omega_{n}$, 
or $(B_{k}\cap C_{n-1})/\omega_{n}|= C_{n}/\omega_{n}$.
In the second case 
we have  
$S_{k,n-1}\cap (B/\delta)^{|A|} = S_{k,n}\cap (B/\delta)^{|A|}$,
which contradicts the choice of $n$.
Thus, we assume that 
$\proj_{|A|+1}(R') = C_{n-1}/\omega_{n}$.
Since 
$\proj_{1,2,\dots,|A|}(R') = (B/\delta)^{|A|}$, 
there exists $d\in C_{n-1}/\omega_{n}$ such that 
$(B/\delta)^{|A|}\times \{d\}\subseteq R'$.
Then $R'$ can be viewed as 
a binary relation 
with a center containing $d$.
Then Lemma \ref{LEMCentralRelationImplies}
implies the existence of a BA subuniverse on $B/\delta$ or 
a center on $C_{n-1}/\omega_{n}$, which contradicts the definition of a dividing congruence.

Subcase 2C.
$(B/\delta)^{|A|}\subseteq S_{k,\ell}$.
Hence, 
$(B\cap C)/\delta = B/\delta$, which completes the proof.
\end{proof}

\begin{cor}\label{CORIntersectionPCLinearIsGood}
Suppose $R\le_{sd} \mathbf A_{1}\times \mathbf A_{2}$,
$B_1\lll A_1$, $B_2\lll A_2$,
$\sigma$ is a dividing congruence for $B_{1}\lll A_{1}$.
Then 
$\proj_{1}(R\cap (B_{1}\times B_{2}))/\sigma$ is 
either empty, or of size 1, or 
equal to $B_1/\sigma$.
\end{cor}

\begin{proof}
Consider homomorphisms 
$f_{1}\colon \mathbf R\to \mathbf A_{1}$ and 
$f_{2}\colon \mathbf R\to \mathbf A_{2}$
that maps each tuple to the first and the second coordinate respectively.
By Lemma \ref{LEMReverseHomomorphism}
$f_{1}^{-1}(B_{1})\lll \mathbf R$
and $f_{2}^{-1}(B_{2})\lll \mathbf R$.
Put $\delta = f^{-1}(\sigma)$.
By Lemma \ref{LEMIntersectionPCLinearIsGood}
$(f_{1}^{-1}(B_{1})\cap f_{2}^{-1}(B_{2}))/\delta$
is either empty, or of size 1, or equal to $R/\delta$.
Translating this into the language of $\sigma$
and $R\cap (B_{1}\times B_{2})$ we obtained the required statement.
\end{proof}

\begin{cor}\label{CORParallelogramPropertyForD}
Suppose 
$B_{1}\lll A$, $B_{2}\lll A$, 
$C_{1}<_{\TD(\sigma_{1})}^{A} B_{1}$,
$C_{1}'<_{\TD(\sigma_{1})}^{A} B_{1}$,
$C_{2}<_{\TD(\sigma_{2})}^{A} B_{2}$,
$C_{2}'<_{\TD(\sigma_{2})}^{A} B_{2}$,
$C_{1}'\cap C_{2}\neq \varnothing$, 
$C_{1}\cap C_{2}'\neq \varnothing$, and 
$C_{1}'\cap C_{2}'\neq \varnothing$.
Then 
$C_{1}\cap C_{2}\neq \varnothing$.
\end{cor}

\begin{proof}
By Lemma \ref{LEMIntersectionPCLinearIsGood} 
$(B_{1}\cap B_{2})/\sigma_1$ is either of size 1, or equal to
$B_{1}/\sigma_1$.
In the first case we 
have 
$C_{1}\cap B_{2}=B_{1}\cap B_{2}$
and $C_{1}\cap C_{2} = B_{1}\cap C_{2}\neq\varnothing$.
Similarly, if
$(B_{1}\cap B_{2})/\sigma_2$ is of size 1, then 
$C_{1}\cap C_{2} = C_{1}\cap B_{2}\neq\varnothing$.

Thus, we assume that 
$(B_{1}\cap B_{2})/\sigma_1 = B_{1}/\sigma_1$ 
and
$(B_{1}\cap B_{2})/\sigma_2 = B_{2}/\sigma_2$.
Let 
$S = \{(a/\sigma_1,a/\sigma_2)\mid a\in B_{1}\cap B_{2}\}$.
Then $S\le_{sd} (B_{1}/\sigma_1)\times (B_{2}/\sigma_2)$.

Applying Lemma \ref{LEMIntersectionPCLinearIsGood}  (d)
to $C<_{\TD(\sigma_1)}^{A} B_{1}$, $B_{2}\lll A$,
and the congruence $\sigma_2$
we derive one of the two cases:

Case 1. There exists $a\in B_{1}/\sigma_1$
such that $(a,b)\in S$ for every $b\in B_{2}/\sigma_2$.
If $S$ is full, we immediately obtain 
$C_{1}\cap C_{2}\neq \varnothing$.
Otherwise, Lemma \ref{LEMCentralRelationImplies}
implies the existence of a BA or center on $B_{1}/\sigma_{1}$ or $B_{2}/\sigma_2$, 
which contradicts our assumptions.

Case 2.
For every $a\in B_{1}/\sigma_1$
there exists a unique $b\in B_{2}/\sigma_2$
such that $(a,b)\in S$.
Choosing $a\in B_{1}/\sigma_1$ corresponding to $C_{1}'$ we 
derive that $C_{2} = C_{2}'$.
Hence, 
$C_{1}\cap C_{2} = C_{1}\cap C_{2}'\neq\varnothing$.
\end{proof}

\subsection{Properties of PC or Linear congruences}\label{SUBSECTIONIrreduciblePCOrLinear}







To prove the following lemma we will need several standard algebraic definitions. Two algebras $\mathbf A$ and $\mathbf B$ are called \emph{polynomially equivalent} if $A = B$
the clone generated by $\mathbf A$ and all constants operations 
coincides with the clone generated by $\mathbf B$ and all constant operations.
An algebra is called \emph{affine} 
if it is polynomially equivalent to an $\mathbf R$-module.
An algebra is called \emph{Abelian} 
if all operations $t \in \Clo(\algA)$ (of arbitrary arity $n+1$) satisfy the
following condition
\begin{align*}
\forall x,y,u_1,\dots,u_{n},v_1,\dots,v_n \in A
\colon& \\t(x,u_1,\dots,u_n) = &t(x,v_1,\dots,v_n) \Rightarrow t(y,u_1,\dots,u_n) = t(y,v_1,\dots,v_n).
\end{align*}

An equivalent definition is given in the following lemma.

\begin{lem}[\cite{hobby1988structure}]\label{LEMAbelianEquivalentDefinition}
An algebra $\mathbf A$ is Abelian if and only if there exists a congruence 
$\delta$ on $\mathbf A^{2}$ such that $\{(a,a)\mid a\in A\}$ is a block of $\delta$.
\end{lem}

\begin{lem}\label{LemBridgeEquivalentToAbelianness}
An algebra $\mathbf A$ is Abelian if and only if 
there exists a bridge $\delta$ from $0_{\mathbf A}$ to $0_{\mathbf A}$ 
such that $\widetilde \delta = \proj_{1,2}(\delta) = \proj_{3,4}(\delta) = A^{2}$.
\end{lem}

\begin{proof}
A congruence on $\mathbf A^{2}$ from Lemma 
\ref{LEMAbelianEquivalentDefinition} is a required bridge.

To obtain a congruence from a bridge, 
compose the bridge with itself sufficient number of times to obtain 
a reflexive symmetric transitive relation on $\mathbf A^{2}$.
\end{proof}

\begin{lem}[\cite{hobby1988structure}]\label{LEMAbelianEqualAffineForWNU}
Suppose a finite algebra $\mathbf A$ has a WNU term operation.
Then $\mathbf A$ is Abelian if and only if it is affine. 
\end{lem}

\begin{lem}\label{LEMNiceBridgeGivesAbelianGroup}
Suppose 
$\sigma$ is a congruence on $\mathbf A$, 
$\delta$ is a bridge from $\sigma$ to $\sigma$
such that $\proj_{1,2}(\delta) = \widetilde \delta = A^{2}$
and $\delta(x_1,x_2,x_3,x_4) = \delta(x_3,x_4,x_1,x_2)$.
Then 
there exists an
abelian group $(G;+,-)$
such that $(A/\sigma; 
\delta/\sigma)\cong 
(G; x_1-x_2=x_3-x_4)$.
\end{lem}
\begin{proof}
By Lemmas  \ref{LemBridgeEquivalentToAbelianness} and
\ref{LEMAbelianEqualAffineForWNU},
$\mathbf A/\sigma$ is affine and, therefore, polynomially equivalent to an $\mathbf R$-module
$\mathbf G = (G,+,0,-,(r)_{r \in \mathbf R})$.
Composing the bridge $\delta$ with itself we get a bridge 
$\delta_{0}\supseteq \delta$ which is an equivalence relation on 
$\mathbf A^{2}$.
Then $\delta_{0}/\sigma$ is a congruence on the
$\mathbf R$-module $\mathbf G^{2}$.
Congruences in $\mathbf R$-modules come from submodules ($u$ and $v$ are congruent if $u-v$ is
in a submodule).
Since the diagonal of $(A/\sigma)^2$ 
is the block of the congruence $\delta_{0}/\sigma$,
the corresponding submodule is $\{(x,y): x=y\}$.
Hence $\delta_{0}/\sigma =
\{((x_1,y_1),(x_2,y_2)): x_1-x_2=y_1-y_2\}$.
It remains to show that $\delta = \delta_{0}$.
Since $\delta/\sigma$ is preserved by the Maltsev operation 
$x-y+z$ and 
$\widetilde \delta =\proj_{1,2}(\delta) = A^{2}$, 
we derive that 
$\proj_{1,2,3}(\delta) = A^{3}$.
Since the last element of $\delta/\sigma$ is uniquely determined 
by the first three, we derive $\delta = \delta_{0}$.
\end{proof}

\begin{lem}\label{LEMBlockOfGoodBridgeDoesNotHaveBAC}
Suppose $\delta$ is a bridge from $0_{A}$ to $0_{A}$, 
$\proj_{1,2}(\delta)\subseteq \widetilde \delta$, 
$\proj_{1,2}(\delta)$ is linked. 
Then $A$ is BA and center free.
\end{lem}

\begin{proof}
We prove by induction on the size of $A$.
First, we build a symmetric bridge $\sigma$ from $\delta$ by 
$$\sigma(x_1,x_2,x_3,x_4) = \exists x_5 \exists x_6\;
\delta(x_1,x_2,x_5,x_5)
\wedge 
\delta(x_3,x_4,x_5,x_6).$$
It follows from the definition of the bridge that 
$\sigma$ satisfies all the properties 
satisfied by $\delta$ but additionally we have 
$\sigma(x_1,x_2,x_3,x_4) =\sigma(x_3,x_4,x_1,x_2)$.
Put $\omega = \proj_{1,2}(\sigma)$.

Assume that there exists  
$B<_{T} A$, where $T\in\{\TBA,\TC\}$.
Since $\proj_{1,2}(\sigma)$ is linked 
either
$((A\setminus B)\times B)\cap \omega \neq\varnothing$, 
or $(B\times (A\setminus B)\cap \omega \neq\varnothing$.
Without loss of generality let it be the first case.
Then choose $(a,b) \in ((A\setminus B)\times B)\cap \omega$.
Consider two cases:

Case 1. There exists a subalgebra 
$D\lneq A$ such that $a,b\in D$. 
Let $E$ be the maximal subuniverse of $D$ such that 
$\omega\cap E^{2}$ is linked and $a,b\in E$.
Then we apply the inductive assumption to 
$\sigma\cap E^4$ and derive that 
$E\cap B$ cannot be a proper BA or central subuniverse. 
By Lemma \ref{LEMBACenterImplyIntersection}
$E\cap B<_{T}E$, which gives a contradiction.

Case 2. 
There does not exist a subalgebra
$D\lneq A$ such that $a,b\in D$.
Hence $\{a\}\circ\omega = A$.
Put 
$\xi(x_1,x_2) = \sigma(a,x_1,x_2,b)$.
Since 
$(a,b,a,b)$ and $(a,a,b,b)$ are from $\sigma$, 
both $\proj_{1}(\xi)$ 
and 
$\proj_{2}(\xi)$ 
contains $a$ and $b$. 
Be the definition of case 2
we derive 
$\proj_{1}(\xi) = \proj_{2}(\xi) =A$.
Put 
$C = B\circ \xi$. 
By Lemma \ref{LEMBACenterSImplyPPDefinition}
$C\le_{T} A$.
Notice that $b\notin C$ as otherwise we would have a tuple 
$(a,c,b,b)\in \sigma$ for some $c\in B$, which is not possible because $a\notin B$.
Thus, we have 
$B<_{T}A$, $C<_{T} A$, 
$b\in B$, $a\in C$, $b\notin C$. 
If $T =BA$ then 
applying the binary absorbing operation to 
$(a,b,a,b)$ and $(a,a,b,b)$ 
we get a tuple $(a,b_1,b_2,b)$, 
where $b_{1},b_{2}\in B$. This implies that 
$C\cap B\neq \varnothing$.
If $T =C$ then by Corollary \ref{CORMainStableIntersection} we have
\begin{align*}
&(\{a\}\times B\times C\times B)\cap\sigma\neq\varnothing \\
&(\{a\}\times C\times C\times C)\cap\sigma\neq\varnothing \;\;\;\;\;\;\;\Rightarrow \;\;\;\;(\{a\}\times C\times C\times B)\cap\sigma\neq\varnothing\\
&(\{a\}\times C\times B\times B)\cap\sigma\neq\varnothing  
\end{align*}
Consider 
$C' = \proj_{2}(\sigma\cap (\{a\}\times C\times C\times B))$.
Using Lemma \ref{LEMBACenterSImplyPPDefinition}
and the fact that 
$\proj_{1}(\xi) = A$,
we derive 
$C'\le_{C} C$.
Moreover, by the definition of a bridge 
$a\notin C'$ or
$C\cap B\neq\varnothing$.
Thus, in all the remaining cases we either 
have $C\cap B\neq\varnothing$, or 
$C'<_{C} C$.
In the first case put 
$C' = B\cap C$.

To complete the proof we consider the case when 
$\varnothing\neq C'<_{T} C<_{T}A$.
Notice that $|C|>1$.
Applying the inductive assumption to
$\sigma\cap C^{4}$ and using the fact that 
$\{a\}\circ \omega = A$, we get a contradiction 
with $C'<_{T} C$.
\end{proof}

\begin{lem}\label{LEMNontrivialReflexiveBridgeImplies}
Suppose 
$\sigma$ is an irreducible congruence on $A$, $\delta$ is a bridge
from $\sigma$ to $\sigma$ such that 
$\widetilde \delta\supsetneq \sigma$. Then 
\begin{enumerate}
    \item $\sigma^{*}$ is a congruence;
    \item $B/\sigma$ is BA and center free for each block $B$ of $\sigma^{*}$.
    \item if $\delta(x_1,x_2,x_3,x_4) = \delta(x_3,x_4,x_1,x_2)$, then there exists a prime $p$  
    such that for every block $B$ of $\sigma^{*}$ we have
    $(B/\sigma;(\delta\cap B^{4})/\sigma)\cong 
    (\mathbb Z_{p}^{n_{B}}; x_1-x_2=x_3-x_4)$, where $n_{B}\ge 0$.
\end{enumerate}
\end{lem}

\begin{proof}
Since the relations 
$\sigma$, $\sigma^{*}$, $\delta$ are stable under $\sigma$ 
we can factorize them by $\sigma$ 
and consider $0_{A/\sigma}$, $\sigma^{*}/\sigma$, 
and $\delta/\sigma$ instead.
To avoid new notations we assume  that 
$\sigma = 0_{A}$. 

Consider some block $B$ of $\LeftLinked(\sigma^{*})$
that is not a block of $\sigma$ (equivalently, of size greater than one).
Put $\delta' = \delta \cap B^{4}\cap (\sigma^{*}\times\sigma^{*})$.
Since $\widetilde \delta\supseteq \sigma^{*}$, $\delta'$ satisfies all the conditions of 
Lemma \ref{LEMBlockOfGoodBridgeDoesNotHaveBAC}, which 
implies that  
$B$ is BA and center free.
Applying Lemma \ref{LEMLinkedImpliesBACenter}
to $\sigma^{*}\cap B^{2}$ we derive that 
$B^{2}\subseteq \sigma^{*}$.
Therefore, $\LeftLinked(\sigma^{*})=\sigma^{*}$ and 
$\sigma^{*}$ is a congruence.

To prove the rest
consider a block $B$ of $\sigma^{*}$ of size at least 2
and apply Lemma \ref{LEMNiceBridgeGivesAbelianGroup} 
to the bridge $\delta' = \delta\cap B^{4}$.
Then, $(B;\delta)\cong 
(G_{B};x_1-x_2=x_3-x_4)$
for some Abelian group 
$(G_{B};+,-)$.
It remains to show there exists a prime $p$ such that 
each Abelian group $G_{p}$ is isomorphic to 
$(\mathbf Z_{p}^{n};+,-)$.



First, we simplify the bridge $\delta$ 
and consider $\omega = \delta\cap (\sigma^{*}\times\sigma^{*})$.
Notice that $\omega$ satisfies the same properties.

Assume that 
$G_{B}$ has elements of coprime orders $p_1$ and $p_2$ 
or $G_{B}$ has a element of order $p_{1}^{k}$, where $k\ge 2$.
Composing the relation $x_1-x_2=x_3-x_4$ we can define a 
relation $(k+1)\cdot x_1 = k\cdot x_2+x_3$ for any $k\in\mathbb N$.
In fact 
$$
((k+1)\cdot x_1 = k\cdot x_2+x_3) =
\exists x_4 \; (k\cdot x_1 = (k-1)\cdot x_2+x_4)
\wedge (x_1-x_2 = x_3-x_4).
$$
Hence $p_1\cdot x_1 = p_1\cdot x_2$ is also pp-definable
from $x_1-x_2 = x_3-x_4$.
Let this pp-definition define the binary relation 
$S$ if we replace $x_1-x_2=x_3-x_4$ by $\omega$.
It follows from the definition that 
$S$ is a congruence on $A$ and 
$S\supsetneq \sigma$ and $\sigma^{*}\not\subseteq S$, 
which contradicts the irreducibility of $\sigma$.
Hence the order of any element of $G_{B}$ is a prime number. 
Similarly, if elements of different groups 
$G_{B_1}$ and $G_{B_2}$ have different orders, we 
we can define a similar relation ``$p_1\cdot x_1 = p_1\cdot x_2$''
and again get a contradiction with the irreducibility of $\sigma$.
\end{proof}

\begin{LEMLinearEquivalentConditionsLEM}
Suppose $\sigma$ is an irreducible congruence on $\mathbf A$.
Then the following conditions are equivalent:
\begin{enumerate}
\item $\sigma$ is a linear congruence;
\item there exists a bridge $\delta$ from $\sigma$ to $\sigma$ such that 
$\widetilde\delta\supsetneq \sigma$.
\end{enumerate}
\end{LEMLinearEquivalentConditionsLEM}
\begin{proof}
$1\Rightarrow 2$. By property 3 of linear congruence we have a bridge 
$\delta$ such that $\widetilde \delta\supseteq\sigma$.

$2\Rightarrow 1$. 
We derive another bridge
$\delta'(x_1,x_2,x_3,x_4)=
\exists x_5\exists x_6\;
\delta(x_1,x_2,x_5,x_6)\wedge \delta(x_3,x_4,x_5,x_6)$.
and apply Lemma \ref{LEMNontrivialReflexiveBridgeImplies}
 to it.
\end{proof}

\begin{lem}\label{LEMPCCongruencePropertyInductiveStep}
Suppose 
$\sigma$ is a congruence on $\mathbf A$, 
$\delta$ is a reflexive bridge from $\sigma$ to $\sigma$
satisfying 
\begin{enumerate}
    \item $\delta(x_1,x_2,x_3,x_4) = \delta(x_3,x_4,x_1,x_2)$;
    \item $(a,b,a,b),(b,a,b,a)\in\delta$ for every $(a,b)\in\proj_{1,2}(\delta)$;
    \item $\RightLinked(\proj_{1,2,3}(\delta)) = A^{2}$.
\end{enumerate}
Then there exists $a,b\in A$ such that 
$a\neq b$ and $(a,a,b,b)\in \delta$.
\end{lem}
\begin{proof}
As before, we assume that 
$\sigma$ is the equality relation as otherwise we can factorize all the relations by $\sigma$.
We prove the lemma by induction on the size of $A$.
Consider two cases:

Case 1. There exists $B<_{T} A$ such that $|B|>1$ and 
$T\in\{\TBA,\TC\}$. Put $\delta' = \delta\cap B^{4}$.
By Lemma \ref{LEMBACenterLinkedness}
$\RightLinked(\proj_{1,2,3}(\delta')) = B^{2}$.
Assume that $\delta'$ is not a bridge, then 
$\proj_{1,2}(\delta')$ is the equality relation.
Then $\RightLinked(\proj_{1,2,3}(\delta')) = B^{2}$
implies the existence of 
$(a,a,b,b)\in\delta'$ such that $a\neq b$, which is what we need.
If $\delta'$ is a bridge then the inductive 
assumption implies the
existence of the corresponding 
$(a,a,b,b)\in\delta'$.

Case 2. Otherwise, there exists $\{a\}<_{T} A$, where $T\in\{\TBA,\TC\}$.
Choose $(a,b,c,d)\in\delta$ such that $c\neq a$, which exists 
because $\proj_{1,3}(\delta)$ is linked.
If $a=b$ then $c=d$ and we found the required pair $(a,c)$.
If 
$\{a\}\circ \proj_{1,3}(\delta)\neq A$
then $\{a\}\circ \proj_{1,3}(\delta)<_{T}A$
and $a,c\in \{a\}\circ \proj_{1,3}$, which is
Case 1.
Otherwise, consider a tuple 
$(a,e,b,f)\in\delta$.
By the assumption on $\delta$ 
we have $(b,a,b,a)\in\delta$. 
Let $g$ be a ternary absorbing operation on $A$ 
witnessing that $\{a\}$ absorbs $A$.
Applying this operation to the tuples
$(a,e,b,f),(b,a,b,a),(a,a,a,a)$ we 
obtain 
$(a,a,g(b,b,a),a)\in \delta$, hence
$g(b,b,a) = a$.
It remains to apply
$g$ to $(b,a,b,a),(b,a,b,a),(a,e,b,f)$ 
and obtain $(a,a,b,a)\in\delta$ which contradicts the definition of a bridge.

Case 3. There does not exist a BA or central subuniverse on $A$.
Consider $\proj_{1,2,3}(\delta)$ as a binary relation
and put $C = \{(a,b)\mid 
\{a\}\times\{b\}\times A\subseteq \proj_{1,2,3}(\delta)\}$.
By Lemma \ref{LEMBAConLeftOrCenterOnRight}  
$C\le_{BA} \proj_{1,2}(\delta)$.
Since $(a,a)\in \proj_{1,2}(\delta)$ for any $a\in A$, 
Lemma \ref{LEMAbsorbingEquality} implies that 
$(a,a)\in C$ for some $a\in A$. Then 
$(a,a,b,b)\in \delta$ for any $b\neq a$, which completes this case.
\end{proof}

\begin{LEMPCBridgesAreTrivialLEM}
Suppose $\sigma$ is a PC congruence on $A$.
Then any reflexive bridge $\delta$ from $\sigma$ to $\sigma$
such that $\proj_{1,2}(\delta) = \proj_{3,4}(\delta)=\sigma^{*}$
can be represented as 
$\delta(x_1,x_2,x_3,x_4)= \sigma(x_1,x_3)\wedge \sigma(x_2,x_4)$ or 
$\delta(x_1,x_2,x_3,x_4)= \sigma(x_1,x_4)\wedge \sigma(x_2,x_3)$.
\end{LEMPCBridgesAreTrivialLEM}

\begin{proof}
Define a new bridge 
$\xi$ by  
$$\xi(x_1,x_2,x_5,x_6) = \exists x_3\exists x_4 \;\;
\delta(x_1,x_2,x_3,x_4)\wedge \delta(x_5,x_6,x_3,x_4).$$
Consider 
$\RightLinked(\proj_{1,2,3}(\xi))$
and $\RightLinked(\proj_{1,2,4}(\xi))$.
Since $\sigma$ is irreducible, 
we have one of the three cases:

Case 1. $\RightLinked(\proj_{1,2,3}(\xi))=\RightLinked(\proj_{1,2,4}(\xi))=\sigma$. Hence, for any 
$(a,b,c,d)\in \xi$ the elements $c/\sigma$ and $d/\sigma$ 
are uniquely determined by $a/\sigma$ and $b/\sigma$.
Since $\xi$ is symmetric, we have $(a,b,a,b)\in \xi$. 
Therefore, $(a,c)\in\sigma$ and $(b,d)\in\sigma$.
Since $\proj_{1,2}(\delta) = \proj_{3,4}(\delta)$,
for any 
$(a,b,c,d)\in \delta$ the elements $c/\sigma$ and $d/\sigma$ 
are also uniquely determined by $a/\sigma$ and $b/\sigma$.

Define two new relations $\zeta_1$ and $\zeta_2$ and check whether  
one of them is a bridge showing that $\sigma$ is linear.
Put 
\begin{align*}
\zeta_{1}(x_1,x_2,x_3,x_4) =&
\exists y\exists y' \exists z\exists z'
\exists t_1\exists t_2\exists t_3\exists t_4 \;\;\\
&\delta(x_1,y,z,t_1)\wedge 
\delta(x_2,y,z',t_2)\wedge
\delta(x_3,y',z,t_3)\wedge 
\delta(x_4,y',z',t_4)\\
\zeta_{2}(x_1,x_2,x_3,x_4) =&
\exists y\exists z \;\;
\delta(x_1,y,z,x_3)\wedge 
\delta(x_2,y,z,x_4)
\end{align*}
Choose some $(a,b,c,d)\in\delta$ such that 
$(a,b)\in\sigma^{*}\setminus\sigma$.
Since $(a,b,c,d),(b,b,b,b)\in\delta$, 
we have $(a,b,a,b)\in\zeta_{1}$ and 
$\proj_{1,2}(\zeta_1)\supseteq \sigma^{*}$.
Consider several subcases:

Subcase 1A. $\zeta_1$ is a bridge.
Since $\widetilde \zeta_{1}$ must be equal to $\sigma$ 
and $(a,a,c,c)\in\zeta_1$ for any $(a,b,c,d)\in\delta$, 
we have $\proj_{1,3}(\delta)=\sigma$.

Subcase 1B. $\zeta_1$ is not a bridge.
Then there exists $(a,a,b,c)\in\zeta_1$ such that 
$(b,c)\notin\sigma$.
Let the evaluation of the variables in the definition of $\zeta_1$ be $y=d$, $y' =  d'$,
$z=e$, $z' =  e'$, and $t_{i} = f_{i}$ for $i=1,2,3,4$.
Since the first two coordinates of $\delta$ 
uniquely (up to $\sigma$) determine the last two, 
we have 
$(e,e')\in\sigma$.
Hence 
$(b,d',e,f_3),(c,d',e,f_4)\in\delta$.
Hence $\proj_{1,2}(\zeta_2)\supsetneq \sigma$
and using the fact that the first two and the last two coordinates of $\delta$ uniquely determine each other, 
we derive that $\zeta_{2}$ is a bridge.
Since $\widetilde \zeta_2$ must be equal to 
$\sigma$, 
we have 
$\proj_{1,4}(\delta) = \sigma$.

Thus, we derived that either 
$\proj_{1,4}(\delta) = \sigma$
or 
$\proj_{1,3}(\delta) = \sigma$.
Repeating the same argument but switching $x_1$ and $x_2$
we derive that 
$\proj_{2,4}(\delta) = \sigma$
or 
$\proj_{2,3}(\delta) = \sigma$.
This completes this case.

Case 2. $\RightLinked(\proj_{1,2,3}(\xi))\supseteq\sigma^{*}$.
Choose a block $B$ of $\RightLinked(\proj_{1,2,3}(\xi))$
that is not a block of $\sigma$.
Let us check that 
$a,b,d\in B$ for any $c\in B$ and $(a,b,c,d)\in \xi$.
Since $(a,b,a,b)\in\xi$, we have $a\in B$. 
Since 
$\RightLinked(\proj_{1,2,3}(\xi))\supseteq\sigma^{*}$ 
and $(a,b),(c,d)\in\sigma^{*}$, we have 
$b,d\in B$.
It remains to apply 
Lemma \ref{LEMPCCongruencePropertyInductiveStep}
to $\xi \cap B^{4}$. 

Case 3. $\RightLinked(\proj_{1,2,4}(\xi))\supseteq\sigma^{*}$.
This case can be considered in the same way as Case 2.
\end{proof}

\begin{LEMNoBridgeBetweenDifferentTypesLEM}
Suppose $\sigma_1$ is a linear congruence, 
$\sigma_2$ is an irreducible congruence, 
$\delta$ is a bridge from $\sigma_1$ to $\sigma_2$.
Then $\sigma_2$ is a also linear congruence.
\end{LEMNoBridgeBetweenDifferentTypesLEM}

\begin{proof}
Without loss of generality we can assume that the relation 
$\widetilde \delta$ is rectangular as otherwise we can compose it with itself many times to obtain rectangularity.
To simplify we replace the bridge 
$\delta$ by 
$\delta\cap(\sigma_{1}^{*}\times\sigma_{2}^{*})$.
Let $\sigma_1$ and $\sigma_2$ be 
congruences on algebras $\mathbf A_{1}$ and $\mathbf A_{2}$, respectively.

Assume that
$\RightLinked(\widetilde\delta) \supsetneq\sigma_{2}$.
Then composing $\delta$ with itself 
we derive a bridge from $\sigma_2$ to $\sigma_2$ 
witnessing that $\sigma_2$ is linear.

Thus, assume that 
$\RightLinked(\widetilde\delta) = \sigma_{2}$.
Notice that 
$\LeftLinked(\widetilde\delta)\supseteq \sigma_{1}^{*}$
as otherwise 
$\LeftLinked(\widetilde\delta)=\sigma_1$, 
$A_{1}/\sigma_1\cong A_2/\sigma_2$, and $\sigma_2$ is also linear.

Put $\delta'(x_1,x_2,x_3,x_4) =
\delta(x_1,x_2,x_3,x_4) \wedge \widetilde\delta(x_1,x_3)$
and consider two cases:

Case 1. $\delta'$ is a bridge, then we obtained a new bridge with the property 
$\proj_{1,3}(\delta') = \widetilde\delta' $.

Case 2. $\delta'$ is not a bridge. Hence, 
there does not exist $(a,b,c,d)\in \delta$ such that 
$(a,b)\notin\sigma_1$ and $(a,c)\in\widetilde\delta$.
Put 
$\delta''(x_1,x_2,x_3,x_4) = 
\exists z \;
\widetilde\delta(x_1,x_3)\wedge \delta(x_2,z,x_3,x_4)$.
Let us show that $\delta''$ is a bridge.
If $(x_{1},x_{2})\in\sigma_1$ then $(x_{2},x_3)\in\widetilde\delta$
and by the assumption we have 
$(x_2,z)\in\sigma_1$ and $(x_3,x_4)\in\sigma_2$.
If $(x_{3},x_{4})\in\sigma_2$ then $(x_{2},z)\in\sigma_1$
and $(x_2,x_3)\in\widetilde\delta$. Hence $(x_1,x_2)\in\sigma_1$.
As $\delta'\neq \delta$, there exists 
$(a,b,c,d)\in\delta$ with $(a,c)\notin\widetilde\delta$.
Choosing $e\in A_{1}$ such that $(e,c)\in\widetilde\delta$
we derive that $(e,a,c,d)\in\delta''$ and $(e,a)\notin\sigma_1$.
Hence, $\delta''$ is a bridge.

Thus, in both cases we build a bridge $\omega$ from $\sigma_1$ to $\sigma_2$ such that 
$\proj_{1,3}(\omega)=\widetilde \omega=\widetilde \delta$.

Assume that $\sigma_2$ is not linear. 
Define a bridge 
$\xi_1$ by 
$$\xi_1(x_1,x_2,x_3,x_4) = \exists x_5\exists x_6 \;\;
\omega(x_5,x_6,x_1,x_2)\wedge 
\omega(x_5,x_6,x_3,x_4).$$
By Lemma \ref{LEMPCBridgesAreTrivial} 
$\xi_1(x_1,x_2,x_3,x_4) =\sigma_2(x_1,x_3)\wedge 
\sigma_2(x_2,x_4)$.
Similarly, we define a bridge
$\xi_2$ by 
$$\xi_2(x_1,x_2,x_3,x_4) = \exists x_5,x_6 \;\;
\omega(x_5,x_6,x_1,x_2)\wedge 
\omega(x_6,x_5,x_3,x_4).$$
Using the facts that
$(x_5,x_1),(x_6,x_3)\in\widetilde\omega$, $(x_5,x_6)\in\sigma_{1}^{*}$,  and 
$\LeftLinked(\widetilde\omega)\supseteq \sigma_{1}^{*}$
we derive that $(x_1,x_3)\in\sigma_2$.
Then Lemma \ref{LEMPCBridgesAreTrivial} 
implies that 
$\xi_2(x_1,x_2,x_3,x_4) =\sigma_2(x_1,x_3)\wedge 
\sigma_2(x_2,x_4)$.
Hence $(b,a,c,d)\in\omega$ whenever $(a,b,c,d)\in\omega$.

Define a new relation $\zeta$ by
$$\zeta(x_1,x_2,x_3,x_4)=
\exists y_1 \exists y_2 \exists y_3 \;\;
\omega(y_1,y_2,x_1,x_2)\wedge
\omega(y_1,y_3,x_1,x_3) \wedge
\omega(y_2,y_3,x_1,x_4).
$$
If $x_1=x_2$ then $y_1=y_2$ and by the property of $\xi_1$
we have $x_3=x_4$.
Consider some $(c,d)\in\sigma_{2}^{*}$.
Then for some $(a,b)\in\sigma_{1}^{*}$ we have
$(a,b,c,d)\in\omega$.
Since $(a,b)\in\sigma_{1}^{*}\subseteq \LeftLinked(\widetilde\omega)$
and $(a,c)\in\widetilde\omega$, 
we have $(a,a,c,c),(b,b,c,c)\in\omega$.
Sending $(x_1,x_2,x_3,x_4)$ to $(c,d,c,d)$ 
and $(y_1,y_2,y_3)$ to $(a,b,a)$ we witness 
that $(c,d,c,d)\in\zeta$.
Sending $(x_1,x_2,x_3,x_4)$ to $(c,c,d,d)$ 
and $(y_1,y_2,y_3)$ to $(a,a,b)$ we witness 
that $(c,c,d,d)\in\zeta$.
Hence $\zeta(x_1,x_2,x_3,x_4)\wedge\zeta(x_3,x_4,x_1,x_2)$ defines
a bridge witnessing that $\sigma_2$ is linear.
\end{proof}

\begin{LEMBridgeTOPCCongruenceLEM}
Suppose 
$\delta$ is a bridge from 
a PC congruence $\sigma_1$ on $\mathbf A_{1}$ to 
an irreducible congruence $\sigma_2$ on 
$\mathbf A_{2}$, 
$\proj_{1,2}(\delta) = \sigma_1^{*}$, and
$\proj_{3,4}(\delta) = \sigma_2^{*}$.
Then \begin{enumerate}
    \item $\sigma_2$ is a PC congruence;
    \item $\mathbf A_1/\sigma_1\cong \mathbf A_2/\sigma_2$;
    \item $\{(a/\sigma_{1},b/\sigma_{2})\mid (a,b)\in\widetilde \delta\}$ is bijective; 
    \item $\delta(x_1,x_2,x_3,x_4) = 
\widetilde \delta(x_1,x_3)\wedge \widetilde \delta(x_2,x_4)$ or
$\delta(x_1,x_2,x_3,x_4) = 
\widetilde \delta(x_1,x_4)\wedge \widetilde \delta(x_2,x_3)$.
\end{enumerate}
\end{LEMBridgeTOPCCongruenceLEM}
\begin{proof}
By Lemma \ref{LEMNoBridgeBetweenDifferentTypes}
$\sigma_2$ must be also be PC congruence.
Then composing the bridge $\delta$ with itself we must get a trivial bridge.
Hence $\widetilde \delta$ is bijective 
and gives an isomorphism 
$\mathbf A_{1}/\sigma_1\cong \mathbf A_{2}/\sigma_2$.
Define a new bridge 
$\xi$ by 
$$\xi(x_1,x_2,x_5,x_6) = 
\exists x_3,x_4 \; 
\delta(x_1,x_2,x_3,x_4)\wedge 
\widetilde \delta(x_5,x_3)\wedge 
\widetilde \delta(x_6,x_4).$$
Then $\xi$ satisfies Lemma \ref{LEMPCBridgesAreTrivial}, which immediately implies the required condition 4.
\end{proof}


    
\subsection{Types interaction}

\begin{lem}[\cite{zhuk2020proof}, Lemma 8.19]\label{LEMBridgeBetweenCongruences}
Suppose
$\omega$, $\sigma_{1}$, and $\sigma_{2}$ are congruences on $\mathbf A$, 
$\omega\cap\sigma_{1} = \omega\cap\sigma_{2}$,
and $\omega\setminus\sigma_{1}\neq\varnothing$.
Then there exists a bridge $\delta$ from 
$\sigma_{1}$ to $\sigma_{2}$ such that 
$\widetilde \delta = \sigma_1\circ\sigma_2$.
\end{lem}

\begin{proof}
Let us define a bridge $\delta$ by
$$\delta(x_{1},x_{2},y_{1},y_{2}) =
\exists z_{1}\exists z_{2}\;
\sigma_{1}(x_{1},z_{1})\wedge
\sigma_{1}(x_{2},z_{2})\wedge
\omega(z_{1},z_{2})\wedge \sigma_{2}(z_{1},y_{1})\wedge \sigma_{2}(z_{2},y_{2}).$$ 
As it follows from the definition, the first two variables of $\delta$ are stable under
$\sigma_{1}$ and the last two variables 
are stable under $\sigma_{2}$.

Let us show that
for any $(a_{1},a_{2},a_{3},a_{4})\in\delta$ we have
$(a_{1},a_{2})\in\sigma_{1}\Leftrightarrow(a_{3},a_{4})\in\sigma_{2}$.
In fact, if $(x_{1},x_{2})\in\sigma_{1}$, then 
$(z_{1},z_{2})\in\sigma_{1}$.
Since $\omega\cap\sigma_{1} = \omega\cap\sigma_{2}$, 
we have $(z_{1},z_{2})\in\sigma_{2}$.
Therefore, $(y_{1},y_{2})\in\sigma_{2}$.

Choose $(a,b)\in\omega\setminus\sigma_{1}$.
Then $(a,b,a,b)\in\delta$ (put $z_{1} = a$, $z_{2} = b$), which gives 
the last necessary property of the bridge.

It follows immediately from the definition of $\delta$ 
that $\widetilde \delta = \sigma_1\circ\sigma_2$.
\end{proof}

\begin{lem}\label{LEMTwoStableIntersection}
Suppose $C_{1}<_{T_{1}(\sigma_{1})}^{A} B_1\lll A$,
$C_{2}<_{T_{2}(\sigma_{2})}^{A} B_2\lll A$,
$T_{1},T_{2}\in\{\TBA,\TC,\TS,\TD\}$,
$C_{1}\cap B_{2}\neq\varnothing$,
$B_{1}\cap C_{2}\neq\varnothing$,
$C_{1}\cap C_{2} = \varnothing$.
Then 
\begin{enumerate}
    \item $T_{1}=T_{2}\in \{\TBA,\TC,\TD\}$;
    \item if $T_{1} = T_{2} = \TD$,   then there is a bridge $\delta$
    from $\sigma_{1}$ to $\sigma_2$
    such that 
    $\widetilde \delta = \sigma_{1}\circ\sigma_{2}$.
\end{enumerate} 
\end{lem}

\begin{proof}
Assume that $T_{1} = \TS$. Choose $S_1\le C_{1}$ such that $S_{1}<_{\TBA,\TC} B_{1}$. 
By Lemma \ref{LEMIntersectionPCLinearIsGood}(s) 
$S_{1}\cap C_{2}\neq\varnothing$ and therefore 
$C_{1}\cap C_{2}\neq\varnothing$, which gives a contradiction.

Similarly, we prove that the case $T_{2} = \TS$ cannot happen.

Assume that $T_{1},T_{2}\in\{\TBA,\TC\}$.
Then by Lemma \ref{LEMBACenterImplyIntersection} 
$C_{1}\cap B_{2}<_{T_{1}}B_{1}\cap B_{2}$,
$B_{1}\cap C_{2}<_{T_{2}}B_{1}\cap B_{2}$
and 
the claim follows from Lemma \ref{LEMBACenterSPossibleIntersections}.

Assume that
$T_{1}\in\{\TBA,\TC\}$ and $T_{2}=\TD$.
By Lemma \ref{LEMIntersectionPCLinearIsGood}(d),
$(B_{1}\cap B_{2})/\sigma_{2} = C_{2}/\sigma_{2}$
or 
$(B_{1}\cap B_{2})/\sigma_{2} = B_{2}/\sigma_{2}$.
First case would imply 
that 
$B_{1}\cap B_{2} = B_{1}\cap C_{2}$, which contradicts 
$C_{1}\cap C_{2}=\varnothing$ and 
$C_{1}\cap B_{2}\neq \varnothing$.
Thus, we assume that $(B_{1}\cap B_{2})/\sigma_{2} = B_{2}/\sigma_{2}$.
Combining Lemmas 
\ref{LEMBACenterImplyIntersection}
and \ref{LEMBACenterSImplyFactor} 
we obtain 
$(C_{1}\cap B_{2})/\sigma_{2}<_{T_{1}}
(B_{1}\cap B_{2})/\sigma_{2}= B_{2}/\sigma_{2}$, 
which contradicts the definition of a dividing congruence.

It remains to consider the case when 
$T_{1}=T_{2}=\TD$. 
Let $\omega_{1},\dots,\omega_{s}$ be all the dividing congruences
coming from 
$B_{1}\lll^{A} A$ and $B_{2}\lll^{A} A$.
Put $\omega = \cap_{i=1}^{s} \omega_{i}$.
By Lemma \ref{LEMIntersectionPCLinearIsGood}(d)
we have 
$\sigma_{1}\supseteq \sigma_{2}\cap \omega$
and 
$\sigma_{2}\supseteq \sigma_{1}\cap \omega$.
By choosing 
$c_1\in C_{1}\cap B_{2}$ and $c_2\in B_{1}\cap C_{2}$, 
we obtain $(c_1,c_2)\in \omega\setminus \sigma_{1}$.
Thus, 
$\sigma_{1}\cap\omega = \sigma_2\cap \omega$
and $\omega\setminus\sigma_{1}\neq\varnothing$,
and 
Lemma \ref{LEMBridgeBetweenCongruences}
implies the existence of a bridge $\delta$ 
from $\sigma_1$ to $\sigma_2$ 
such that $\widetilde \delta = \sigma_{1}\circ\sigma_{2}$.
\end{proof}

\subsection{Factorization of strong subalgebras}

In this section first we prove several technical lemmas, 
then we show we can factorize subalgebras of type $T$ by a congruence keeping the type of the subalgebra 
(see Lemma \ref{LEMMainExistenceOfIrreducibleCongruence}).



\begin{lem}\label{LEMCenterCanBePushedIn}
Suppose 
$R\le_{sd} A_1\times A_2$,
$R\cap (B_{1}\times B_{2})\neq \varnothing$,
$B_{1}\lll A_{1}$, $B_{2}\lll A_2$,
$\sigma$ is a congruence on $A_1$,
$B_{1}/\sigma$ is BA and center free,
there exists $c\in A_{2}$ such that 
 $(E\times\{c\})\cap R\neq\varnothing$ for every $E\in B_{1}/\sigma$.
Then 
there
exists $c\in B_{2}$ such that 
 $(E\times\{c\})\cap R\neq\varnothing$ for every $E\in B_{1}/\sigma$.
\end{lem}
\begin{proof}
Consider a minimal
$B_{2}'$ such that 
$B_{2}\lll^{A_{2}} B_{2}''<_{T(\delta)}^{A_{2}} B_{2}'\lll A_{2}$
and $c$ can be chosen from $B_{2}'$.
Define 
$$S' = \{(a_1/\sigma,\dots,a_{|A_1|}/\sigma)\mid 
\exists c\in B_{2}'\; \forall i\colon a_{i}\in B_{1}, (a_i,c)\in R\},$$
$$S'' = \{(a_1/\sigma,\dots,a_{|A_1|}/\sigma)\mid 
\exists c\in B_{2}''\; \forall i\colon a_{i}\in B_{1},(a_i,c)\in R\}.$$

If $T\in \{\TBA,\TC\}$ then by Lemmas
\ref{LEMBACenterImplyIntersection} and \ref{LEMBACenterSImplyFactor}
$S''<_{T} S'= (B_{1}/\sigma)^{|A_1|}$ and by Lemma 
\ref{LEMBACenterSOnPowerImplies}
there exists a BA or central subuniverse on $B_{1}/\sigma$,
which contradicts our assumptions.

If $T=\TS$ then choose 
$D_{2}\le B_2''$ such that 
$D_{2}<_{\TBA,\TC} B_{2}'$.
Combining Corollary \ref{CORReverseHomomorphism} and Lemma \ref{LEMIntersectionPCLinearIsGood}(s) 
we obtain that $R\cap (B_{1}\times D_{2})\neq \varnothing.$ 
Hence, Lemmas \ref{LEMBACenterImplyIntersection} and \ref{LEMBACenterSImplyFactor}
imply the existence of both BA and central subuniverse on 
$S'= (B_{1}/\sigma)^{|A_1|}$ and by Lemma 
\ref{LEMBACenterSOnPowerImplies}
there exists a BA and central subuniverse on $B_{1}/\sigma$,
which contradicts our assumptions.

Suppose 
$T=\TD$.
Define a relation $S\le (B_{1}/\sigma)^{|A_{1}|}\times B_{2}'/\delta$ by 
$$S = \{(a_1/\sigma,\dots,a_{|A_1|}/\sigma,b/\delta)\mid 
\exists c\in B_{2}'\; \forall i\colon a_{i}\in B_{1},(a_i,c)\in R, (c,b)\in\delta\}.$$
By the choice of $B_{2}'$ 
there exists $d\in B_{2}'/\delta$ such that 
$(B_{1}/\sigma)^{|A_1|}\times \{d\}\subseteq S$
but $(B_{1}/\sigma)^{|A_1|}\times B_{2}''/\delta\neq S$.
Since 
$\proj_{|A_{1}|+1}(S)/\delta
=\proj_{2}(R\cap (B_{1}\times B_{2}'))/\delta$,
Corollary \ref{CORIntersectionPCLinearIsGood}
implies that $\proj_{|A_{1}|+1}(S)/\delta=B_{2}'/\delta$.
Hence, 
$S$ is a central relation.
Combining Lemmas \ref{LEMCentralRelationImplies}
and \ref{LEMBACenterSOnPowerImplies} 
we get a contradiction with the assumption that 
$B_{1}/\sigma$ and 
$B_{2}'/\delta$
 are BA and center free.
\end{proof}

\begin{lem}\label{LEMLeftLinkedStayFull}
Suppose 
$R\le_{sd} A_1\times A_2$,
$R\cap (B_{1}\times B_{2})\neq \varnothing$, 
$B_{1}\lll A_{1}$, $B_{2}\lll A_2$,
$\sigma$ is a congruence on $A_1$,
$B_{1}/\sigma$ is BA and center free,
$(\LeftLinked(R)\cap B_{1}^{2})/\sigma = (B_{1}/\sigma)^{2}$. 
Then 
$(\LeftLinked(R\cap (B_{1}\times B_{2})))/\sigma=(B_{1}/\sigma)^{2}$. 
\end{lem}
\begin{proof}
By $R_{n}$ 
we denote the binary relation defined by 
$R_{n} = 
\underbrace{R\circ R^{-1}\circ R\circ R^{-1} \circ\dots\circ
R\circ R^{-1}}_{2n}.$
For sufficiently large $n$
we have 
$(R_{n}\cap B_{1}^{2})/\sigma = (B_{1}/\sigma)^{2}$.

Consider two cases:

Case 1. $(R_{1}\cap B_{1}^{2})/\sigma = (B_{1}/\sigma)^{2}$.
Applying Lemma \ref{LEMBAConLeftOrCenterOnRight}
to $S = \{(a/\sigma,b)\mid a\in B_{1}, (a,b)\in R\}$
we derive the existence of $c\in A_2$ 
such that for every $E\in B_{1}/\sigma$ 
we have $(E\times \{c\})\cap R\neq\varnothing$.
By Lemma \ref{LEMCenterCanBePushedIn}
$c$ can be chosen from $B_{2}$
and $(\LeftLinked(R\cap (B_{1}\times B_{2})))/\sigma=(B_{1}/\sigma)^{2}$.

Case 2. $(R_{1}\cap B_{1}^{2})/\sigma \neq (B_{1}/\sigma)^{2}$.
Consider the maximal $n=2^k$ such that 
$(R_{n}\cap B_{1}^{2})/\sigma \neq (B_{1}/\sigma)^{2}$.
Hence, 
$(R_{2n}\cap B_{1}^{2})/\sigma = (B_{1}/\sigma)^{2}$.
Applying Lemma \ref{LEMBAConLeftOrCenterOnRight}
to $S = \{(a/\sigma,b)\mid a\in B_{1}, (a,b)\in R_{n}\}$
we derive the existence of $c\in A_1$ 
such that for every $E\in B_{1}/\sigma$ 
we have $(E\times \{c\})\cap R_{n}\neq\varnothing$.
By Lemma \ref{LEMCenterCanBePushedIn},
$c$ can be chosen from $B_{1}$.
Hence, the relation 
$(R_{n}\cap B_{1}^2)/\sigma$ is central, which 
by Lemma \ref{LEMCentralRelationImplies} 
implies that 
$(R_{n}\cap B_{1}^2)/\sigma= (B_{1}/\sigma)^2$ 
and contradicts our assumption.
\end{proof}

\begin{lem}\label{LEMCongruenceEitherCutOrDoNothing}
Suppose, 
$\sigma$ is a dividing congruence for $B\lll A$, 
$\delta$ is a congruence on $A$, and
$(\delta\cap B^{2})/\sigma \neq B^{2}/\sigma$.
Then 
\begin{enumerate}
    \item[(1)] $(\delta\cap B^{2}) \subseteq (\sigma\cap B^{2})$;
    \item[(2)] $\sigma\supseteq \delta\cap\omega$;
\item[(3)] $(\delta\vee (\sigma \cap\omega))\cap B^{2} = \sigma\cap B^{2}$;
\item[(4)] $(\delta\vee (\sigma \cap\omega))\cap \omega  = \sigma\cap\omega$;
\end{enumerate}
where $\omega$ is 
the intersection of
all the dividing
congruences
coming from  
$B\lll A$.
\end{lem}


\begin{proof}
Let us prove (2) first.
Consider two equivalence classes 
$C_{1}$ and $C_{2}$ of $\sigma$ such that 
$C_{1}\cap B\neq\varnothing$,
$C_{2}\cap B\neq\varnothing$,
and $((C_{1}\cap B)\times (C_{2}\cap B))\cap \delta=\varnothing$.
Then 
$C_1\cap B<_{\TD(\sigma)}^{A}B\lll A$ and by Lemma \ref{LEMSelfIntersectionPC}
for $m=k$ and $B_{k} = C_1\cap B$ 
we obtain $(((C_1\cap B)\circ\delta)\cap B)/\sigma = \{C_1\}$
and $\sigma\supseteq \delta\cap\omega$.

(1) follows immediately from (2).

Let us prove (4). 
Consider the binary relation 
$R = \delta\circ (\sigma\cap \omega)$.
Notice that 
$\LeftLinked(R) = \delta\vee (\sigma\cap \omega)$.
Put 
$\delta' = \LeftLinked(R)$.
Consider two cases:

Case 1. $(\delta'\cap B^{2})/\sigma = (B/\sigma)^{2}$.
By Lemma \ref{LEMLeftLinkedStayFull} 
$\LeftLinked(R\cap B^{2})/\sigma = (B/\sigma)^{2}$.
Then there exist $a_1,a_2,c\in B$ and $b_1,b_2\in A$
such that $(a_1,b_1),(a_2,b_2)\in\delta$, 
$(b_1,c),(b_2,c)\in\sigma\cap\omega$, 
$(a_1,a_2)\notin\sigma$.
Since $a_1,a_2,c\in B$, 
we have 
$(a_1,c), (a_2,c)\in\omega$.
Therefore, $(a_1,b_1),(a_2,b_2)\in\omega$, 
hence by (2) we have 
$(a_1,b_1),(a_2,b_2)\in\sigma$, which 
contradicts 
the assumption $(a_1,a_2)\notin\sigma$.

Case 2. $(\delta'\cap B^{2})/\sigma \neq (B/\sigma)^{2}$.
Applying (2) to $\delta'$ we obtain 
$\sigma\supseteq \delta'\cap \omega$.
Since 
$\delta'\supseteq (\sigma\cap \omega)$,
we obtain $\delta'\cap \omega = \sigma\cap \omega$.

Condition (3) immediately follows from (4).
\end{proof}

\begin{lem}\label{LEMMainExistenceOfIrreducibleCongruence}
Suppose 
$B\lll A$, 
$\sigma$ is a congruence on $A$ such that 
$|B/\sigma|>1$, 
$B/\sigma$ is BA and center free.
Then 
there exists
a dividing congruence $\delta$ 
for $B\lll A$ such that $\delta\supseteq \sigma$.
\end{lem}

\begin{proof}
Let $\delta\supseteq \sigma$ be a maximal congruence such that 
$|B/\delta|>1$.
It follows from 
Lemma \ref{LEMReverseHomomorphism}
that $B/\delta$ is BA and center free.
Let us show that $\delta$ is irreducible.
Assume that it is not true 
and $\delta$ can be represented as an intersection of binary 
relations $S_{1},\dots,S_{k}\supsetneq \delta$.
Then for some $i$ we have 
$B^{2}\not\subseteq S_{i}$.
Notice that 
$\LeftLinked(S_i)$ is a congruence that is larger than $\delta$, 
hence by the choice of $\delta$ we have 
$B^{2}\subseteq \LeftLinked(S_i)$.

By Lemma \ref{LEMLeftLinkedStayFull} 
$\LeftLinked(S_i\cap (B\times B))/\sigma = (B/\sigma)^{2}$.
Hence, 
$(S_i\cap (B\times B))/\sigma$ is a linked relation, which by Lemma 
\ref{LEMLinkedImpliesBACenter} implies the existence 
of a BA or central subuniverse on $B/\delta$ and contradicts our assumption.
Thus, the congruence $\delta$ is irreducible.

\end{proof}

\begin{lem}\label{LEMFactorByDelta}
Suppose $\delta$ is a congruence on $A$,
$C<_{T(\sigma)}^{A} B\lll A$, where 
$T\in\{\TPC,\TL\}$.
Then 
$C/\delta=B/\delta$, or
$C/\delta<_{\TS}B/\delta$, or
$C/\delta<_{T}^{A/\delta}B/\delta$.
Moreover, if
$C/\delta<_{T}^{A/\delta}B/\delta$
then $\delta\cap B^{2}\subseteq \sigma\cap B^{2}$.
\end{lem}

\begin{proof}
Let $R\le A/\sigma\times A/\delta$ be defined by 
$R = \{(a/\sigma,a/\delta)\mid a\in B\}$.
Consider two cases:

Case 1.
$\LeftLinked(R) = (B/\sigma)^{2}$.
By Lemma \ref{LEMBAConLeftOrCenterOnRight}
there exists $U\subseteq A/\delta$
such that 
$(B/\sigma)\times U\subseteq R$.
Notice that
$C/\delta\supseteq U$ hence 
$C/\delta=B/\delta$ or
$C/\delta<_{\TS}B/\delta$.

Case 2. 
$\LeftLinked(R) \neq (B/\sigma)^{2}$.
Let $\omega$ be the intersection of all the dividing congruences 
coming from 
$B\lll A$.
Since
$(\delta\cap B^{2})/\sigma\neq B^{2}/\sigma$,
by Lemma \ref{LEMCongruenceEitherCutOrDoNothing}(4)
$\sigma\cap \omega=\delta'\cap \omega$, where 
$\delta' = \delta\vee (\sigma \cap\omega)$.
Since $B/\delta'\cong B/\sigma$, 
$B/\delta'$ is BA and center free.
By Lemma \ref{LEMMainExistenceOfIrreducibleCongruence} applied to $\delta'$
there exists $\delta''\supseteq \delta'$ that is a dividing congruence
for $B\lll A$.
By Lemma \ref{LEMCongruenceEitherCutOrDoNothing}(2) 
$\delta''\cap \omega \subseteq \sigma$,
which together with 
$\delta''\cap \omega \supseteq \delta' \cap \omega = \sigma \cap \omega$
implies $\delta'' \cap \omega = \sigma \cap \omega$.
Hence
$C/\delta<_{\mathcal T(\delta''/\delta)}^{A/\delta} B/\delta$,
where $\mathcal T\in\{\TPC,\TL\}$.

To show that $T = \mathcal T$, we apply Lemma \ref{LEMBridgeBetweenCongruences} to 
$\delta'' \cap \omega = \sigma \cap \omega$ and obtain a bridge from 
$\delta''$ to $\sigma$.
This, by Lemma \ref{LEMNoBridgeBetweenDifferentTypes},
implies that 
congruences $\delta''$ and $\sigma$ must be of the same type 
and therefore $T= \mathcal T$.
\end{proof}

\subsection{Proof of the remaining statements from Section \ref{SUBSECTIONSTRONGSUBALGEBRASPROPERTIES}}
\label{SUBSECTIONProofOfTheRemainingProperties}

\begin{LEMUbiquityLEM}
Suppose $B\lll A$, $|B|>1$.
Then there exists $C<_{T}^{A} B$, where $T\in \{\TBA, \TC, \TL, \TPC\}$.
\end{LEMUbiquityLEM}

\begin{proof}
If $B$ has a nontrivial BA or central subuniverse we take this subuniverse as $C$.
Otherwise, we apply Lemma \ref{LEMMainExistenceOfIrreducibleCongruence},
where $\sigma$ is the equality relation.
Then there exists an irreducible congruence $\delta$ for $B$.
It remains to choose any block $D$ of $B/\delta$ and 
put $C= D\cap B$.
\end{proof}

\begin{LEMIntersectALLLEM}
Suppose $B\lll A$, $D\lll A$. 
Then 
\begin{enumerate}
    \item[(i)] $B\cap D\dot\lll A$;
    \item[(t)] $C<_{T(\sigma)}^{A}B \Rightarrow C\cap D\dot\le _{T(\sigma)}^{A} B\cap D$.
\end{enumerate}
\end{LEMIntersectALLLEM}
\begin{proof}
Let us prove (t) first.
If $T\in\{\TBA,\TC\}$ then it follows
from Lemma \ref{LEMBACenterImplyIntersection}.
If $T=\TS$ then 
consider $E\le C$ such that $E<_{\TBA,\TC} B$.
If $C\cap D= \varnothing$, then the lemma holds.
Otherwise, 
Lemma \ref{LEMIntersectionPCLinearIsGood}(s) implies that 
$E\cap D\neq \varnothing$.
By Lemma \ref{LEMBACenterImplyIntersection}.
$E\cap D<_{\TBA,\TC} B\cap D$.
Hence $C\cap D<_{\TS}B\cap D$.
If $T\in\{\TPC,\TL\}$ then by Lemma \ref{LEMIntersectionPCLinearIsGood}(d) 
$(B\cap D)/\sigma$ is either empty, 
or of size 1, or equal to 
$B/\sigma$.
In the first case $C\cap D=\varnothing$ and we are done.
In the last case we have 
$C\cap D< _{T(\sigma)}^{A} B\cap D$.
In the second case we either have 
$C\cap D = B\cap D$, or 
$C\cap D=\varnothing$, which is what we need.

Let us prove (i).
Consider 
the sequence 
$B=B_{k}<_{T_{k}}B_{k-1}<_{T_{k-1}}\dots<_{T_{2}}
B_{1}<_{T_{1}}< A$
and apply (t) to each 
$B_{i}<_{T_{i}} B_{i-1}$.
Then we 
have $B\cap D\dot\lll^{A} D$, which together with 
$D\lll A$ implies 
$B\cap D\dot\lll A$.
\end{proof}

\begin{LEMPropagationLEM}
Suppose $f\colon \mathbf A\to \mathbf A'$ is a surjective homomorphism,
then 
\begin{enumerate}
\item[(f)] $C\lll^{A} B \Rightarrow f(C)\lll f(B)$;
\item[(b)] $C'\lll^{A'} B' \Rightarrow f^{-1}(C')\lll f^{-1}(B)$;    
\item[(ft)] $C<_{T(\sigma)}^{A} B\lll A \Longrightarrow 
(f(C)=f(B)
\text{ or } 
f(C)<_{S}f(B)
\text{ or }
f(C)<_{T}^{A'}f(B))$;
\item[(bt)] $C'<_{T(\sigma)}^{A'} B' \Rightarrow f^{-1}(C')<_{T(f^{-1}(\sigma))}^{A'}f^{-1}(B')$;
\item[(fs)] $T\in\{\TBA,\TC, \TS\}$ and $C<_{T}B$ $\Longrightarrow f(C)\le_{T}f(B)$;
\item[(fm)] 
$C\le_{\mathcal{M}T}^{A}B\lll A$ and $f(B)$ is S-free $\Longrightarrow f(C)\le_{\mathcal{M}T}^{A'}f(B)$;
\item[(bm)] 
$C'\le_{\mathcal{M}T}^{A'}B'\lll A'\Longrightarrow f^{-1}(C)\le_{\mathcal{M}T}^{A}f^{-1}(B)$.
\end{enumerate}
\end{LEMPropagationLEM}
\begin{proof}
Let $\delta$ be the congruence defined by the homomorphism $f$, 
that is $f(\mathbf A)\cong \mathbf A/\delta$.  
(bt).  It follows from Lemma \ref{LEMReverseHomomorphism}.

(b). It is sufficient to apply (bt) several times.

(ft). It follows from Lemma \ref{LEMFactorByDelta}.

(f). It is sufficient to apply (ft) several times.

(fs). It follows from Lemma \ref{LEMBACenterSImplyFactor}.


(fm). 
Let $C = C_1\cap \dots\cap C_{t}$, 
where $C_{i}<_{T(\sigma_{i})}^{A} B$.
We prove by induction on $t$.
If $t=1$ then it follows from Lemma \ref{LEMFactorByDelta}.
Assume that 
$C_{i}/\delta = B/\delta$ for some $i$.
To simplify notations assume that $i=t$.
Then consider 
$D_{j} = C_{j}\cap C_{t}$ for $j=1,2,\dots,t-1$.
By Lemma \ref{LEMIntersectALL} 
$D_{j}\le_{T(\sigma_{i})}^{A}C_t$ for every $j$.
By the inductive assumption 
$C/\delta = (D_{1}\cap\dots\cap D_{t-1})/\delta\le_{\mathcal{M}T}^{A}C_{t}/\delta = B/\delta$.
Hence
$C/\delta\le_{\mathcal{M}T}^{A/\delta}B/\delta$.
Thus, it remains to consider the case when 
$C_{i}/\delta \neq B/\delta$ for every $i$.
By the additional condition to 
Lemma \ref{LEMFactorByDelta} 
we have 
$\delta\cap B^{2}\subseteq \sigma_{i}\cap B^{2}$.
By Lemma \ref{LEMFactorByDelta}
we have 
$C_{i}/\delta<_{T}^{A/\delta} B/\delta$.

Let us show that 
$(C_{1}\cap \dots\cap C_{t})/\delta =
C_{1}/\delta\cap\dots\cap C_{t}/\delta$.
The inclusion $\subseteq$ is obvious.
Let us prove $\supseteq$.
Suppose 
$E \in C_{1}/\delta\cap\dots\cap C_{t}/\delta$.
Since 
$(C_{i}\circ \delta)\cap B = C_{i}$ for every $i$, 
we have $E\cap B\subseteq C_{1}\cap \dots \cap C_{t}$.
Hence $E\in (C_{1}\cap \dots \cap C_{t})/\delta$.

Thus, we showed that
$(C_{1}\cap \dots\cap C_{t})/\delta= 
C_{1}/\delta\cap\dots\cap C_{t}/\delta<_{\mathcal{M}T}^{A/\delta}B/\delta$.

(bm). It is sufficient to apply (bt) and consider the intersection of 
the corresponding PC and linear subuniverses.
\end{proof}

\begin{CORPropagateFromFactorCOR}
Suppose $\delta$ is a congruence on $\mathbf A$,
$\mathbf B,\mathbf C\le \mathbf A$. Then 
\begin{enumerate}
    \item[(f)] $C/\delta\lll^{A/\delta}
B/\delta \Longleftrightarrow C\circ\delta\lll^{A}
B\circ \delta$;
    \item[(t)] $C/\delta<_{T}^{A/\delta}
B/\delta \Longleftrightarrow C\circ\delta<_{T}^{A}
B\circ \delta$.
\end{enumerate}
\end{CORPropagateFromFactorCOR}

\begin{proof}
(t,$\Rightarrow$). It follows from Lemma \ref{LEMReverseHomomorphism}. 

(t,$\Leftarrow$).
if $T\in\{\TBA,\TC,\TS\}$ then
it follows from Lemma \ref{LEMBACenterSImplyFactor}. 
If $T\in\{\TPC,\TL,\TD\}$ then it follows immediately from the definition.

(f). To prove (f) it is sufficient to apply (t) several times.
\end{proof}

\begin{CORPropagateMultiplyByCongruenceCOR}
Suppose $\delta$ is a congruence on $\mathbf A$. Then 
\begin{enumerate}
\item[(f)] $C\lll^{A}B  \Rightarrow C\circ \delta \lll^{A} B\circ \delta$;
    \item[(t)] $C<_{T(\sigma)}^{A}B\lll^{A}A\Longrightarrow (C\circ\delta= B\circ \delta \text{ or }
    C\circ\delta<_{\TS}^{A} B\circ \delta \text{ or }
    C\circ\delta<_{T}^{A} B\circ \delta)$;
    \item[(e)] $\delta\subseteq \sigma$ and $C<_{T(\sigma)}^{A}B\lll^{A}A\Longrightarrow 
    C\circ\delta<_{T}^{A} B\circ \delta$.
\end{enumerate}
\end{CORPropagateMultiplyByCongruenceCOR}

\begin{proof}
(f). Corollary \ref{CORPropagationModuloCongruence}(f)
implies  
$C/\delta\lll^{A/\delta}B/\delta$, then 
Corollary \ref{CORPropagateFromFactor}(f)
implies $C\circ \delta \lll^{A} B\circ \delta$.

(t).  
Again the proof is just a combination of 
Corollary \ref{CORPropagationModuloCongruence}(t)
and 
Corollary \ref{CORPropagateFromFactor}(t).

(e). We use (t) and notice that 
$C\circ\delta<_{\TS}^{A} B\circ \delta$
implies that 
$B/\sigma$ has BA and central subuniverse, which contradicts 
the definition of a dividing congruence.
\end{proof}

\begin{CORPropagateToRelationsCOR}
Suppose $R\le_{sd} A_{1}\times\dots\times A_{n}$,
$B_{i}\lll A_{i}$ for $i\in[n]$. Then 
\begin{enumerate}
    \item[(r)] $R\cap (B_1\times\dots\times B_{n}))\dot\lll R$;
    \item[(r1)] $\proj_{1}(R\cap (B_1\times\dots\times B_{n}))\dot\lll A_{1}$;
    \item[(b)] $\forall i\colon  C_{i}\lll^{A_{i}}B_{i}\Longrightarrow (R\cap (C_1\times\dots\times C_{n}))\dot\lll^{R} 
    (R\cap (B_1\times\dots\times B_{n}))$;   
    \item[(b1)] $\forall i\colon  C_{i}\lll^{A_{i}}B_{i}\Longrightarrow \proj_{1}(R\cap (C_1\times\dots\times C_{n}))\dot\lll^{A_1} \proj_{1}(R\cap (B_1\times\dots\times B_{n}))$;
    \item[(m)] $\forall i\colon  C_{i}\le_{\mathcal{M}T}^{A_{i}}B_{i} \Longrightarrow R\cap (C_1\times\dots\times C_{n})\dot\le_{\mathcal{M}T}^{R} R\cap (B_1\times\dots\times B_{n})$;
    \item[(m1)] $\forall i\colon  C_{i}\le_{\mathcal{M}T}^{A_{i}}B_{i}$, 
    $\proj_{1}(R\cap (B_1\times\dots\times B_{n}))$ is S-free $\Longrightarrow$
    
    \;\;\;\;\;\;\;\;\;\;\;\;\;\;\;\;\;\;\;\;\;\;\;\;\;\;\;\;\;\;
\;\;\;\;\;\;\;\;\;\;\;\;\;\;\;\;\;\;\;\;\;\;\;$\proj_{1}(R\cap (C_1\times\dots\times C_{n}))\dot\le_{\mathcal{M}T}^{A_{1}} \proj_{1}(R\cap (B_1\times\dots\times B_{n}))$.
\end{enumerate}
\end{CORPropagateToRelationsCOR}

\begin{proof}
(r). Let $f_{i}:\mathbf{R}\to \mathbf A_{i}$ 
be the homomorphism sending each tuple to its $i$-th coordinate.
By Lemma \ref{LEMPropagation}(b)
$f_{i}^{-1}(B_{i})\lll R$ for every $i$.
Then by Lemma \ref{LEMIntersectALL}(i)
$R\cap (B_1\times\dots\times B_{n})=
\bigcap_{i=1}^{n} f_{i}^{-1}(B_{i}) \dot\lll 
R$.

(r1). Additionally to (r) we apply 
$f_{1}$ to the intersection and use Lemma 
\ref{LEMPropagation}(f).

(b). By (r) we have 
$R' := R\cap (B_1\times\dots\times B_{n}))\dot\lll R$
and 
$R'':= R\cap (C_1\times\dots\times C_{n}))\dot\lll R$.
By Lemma \ref{LEMIntersectALL}(i)
$R'' = R'\cap R''\dot\lll^{R} R'$.

(b1). Again, additionally to (b) 
we apply 
$f_{1}$ and use Lemma 
\ref{LEMPropagation}(f).

(m). By the definition of type $\mathcal{M}T$ 
each $C_{i}$ can be represented 
as
$C_{i,1}\cap\dots \cap C_{i,n_i}$ 
where 
$C_{i,j}\le_{T}^{A_{i}}B_{i}$ for all $i$ and $j$.
We will use notations $R'$ and $R''$ from (b).
By Lemma \ref{LEMPropagation}(b)
$f_{i}^{-1}(C_{i,j})\le_{T}^{R} f_{i}^{-1}(B_{i})$.
By Lemma \ref{LEMIntersectALL}(t) and property (r), that we already proved, we have
$$f_{i}^{-1}(C_{i,j})\cap R'\le_{T}^{R} f_{i}^{-1}(B_{i})\cap R' = R'.$$
Hence  
$R'' = R'\cap\bigcap_{i=1}^{t}\bigcap_{j=1}^{n_i} f_{i}^{-1}(C_{i,j})
\dot\le_{\mathcal{M}T}^{R} R'$.

(m1). It is sufficient to apply Corollary 
\ref{CORPropagationModuloCongruence}(m) to (m).
\end{proof}

\begin{THMMainStableIntersectionTHM}
Suppose 
\begin{enumerate}
    \item $C_{i}<_{T_{i}(\sigma_{i})}^{A} B_{i}\lll A$, where 
    $T_{i}\in\{\TBA, \TC,\TS,\TL,\TPC\}$ for $i=1,2,\dots,n$, $n\ge 2$;
    \item $\bigcap\limits_{i\in[n]} C_{i} = \varnothing$;
    \item $B_{j}\cap\bigcap\limits_{i\in[n]\setminus\{j\}} C_{i} \neq  \varnothing$
    for every $j\in[n]$.
\end{enumerate}
Then one of the following conditions hold:
\begin{enumerate}
    \item[(ba)] $T_1=\dots= T_n=\TBA$;
    \item[(l)] 
        $T_1=\dots= T_n=\TL$ and 
        for every $k,\ell\in [n]$
        there exists a bridge $\delta$ from 
    $\sigma_{k}$ and $\sigma_{\ell}$ 
    such that 
    $\widetilde \delta= \sigma_{k}\circ \sigma_{\ell}$;
    \item[(c)] $n=2$ and $ T_1= T_2=\TC$;
    \item[(pc)] $n=2$, $ T_1= T_2=\TPC$,
    and $\sigma_{1} = \sigma_{2}$.
\end{enumerate}     
\end{THMMainStableIntersectionTHM}

\begin{proof}
For $n=2$ it follows from Lemmas \ref{LEMTwoStableIntersection}
and \ref{LEMNoBridgeBetweenDifferentTypes}.
Put $B_{2}' = B_{2}\cap C_3\cap \dots \cap C_{n}$
and $C_{2}' = C_{2}\cap B_{2}'$.
By Lemma \ref{LEMIntersectALL}
$C_{2}'\le_{T_{2}(\sigma_{2})}B_{2}'\lll A$.
Also we have 
$C_{1}\cap C_{2}' = \varnothing$, 
$C_{1}\cap B_{2}'\neq\varnothing$,
and
$B_{1}\cap C_{2}'\neq\varnothing$.
Then Lemmas \ref{LEMTwoStableIntersection} and 
\ref{LEMNoBridgeBetweenDifferentTypes}
imply that $T_{1} = T_{2}$, and 
if $T_{1} \in \{\TPC,\TL\}$ then 
    there is a bridge $\delta$
    from $\sigma_{1}$ to $\sigma_2$
    such that 
    $\widetilde \delta = \sigma_{1}\circ\sigma_{2}$.
If  $T_1=\TPC$ then
$\sigma_{1}\circ\sigma_{2}=\sigma_{1} = \sigma_{2}$,
as otherwise composing the bridge $\delta$ with itself we would get a nontrivial 
bridge from $\sigma_{i}$ to $\sigma_{i}$.
Additionally, this implies that 
$n$ cannot be greater than 2 for $T_1=\TPC$, 
as in this case
the intersection $C_{1}\cap C_{2}$ must be empty.
    
Thus, we proved the required conditions for $T_{1}(\sigma_1)$ 
and $T_{2}(\sigma_2)$. Similarly, we can prove this for any 
$T_{i}(\sigma_i)$ and $T_{j}(\sigma_j)$.
\end{proof}

\begin{CORMainStableIntersectionCOR}
Suppose 
\begin{enumerate}
    \item $R\le_{sd} A_{1} \times \dots\times 
    A_{n}$;
    \item $C_{i}<_{T_{i}(\sigma_{i})}^{A_{i}} B_{i}\lll A_{i}$, where 
    $T_{i}\in\{\TBA, \TC,\TS,\TL,\TPC\}$ for $i=1,2,\dots,n$, $n\ge 2$;
    \item $R\cap (C_{1}\times\dots\times C_{n}) = \varnothing$;
    \item 
    $R\cap (C_{1}\times\dots\times C_{j-1}\times 
    B_{j}\times C_{j+1}\times\dots\times C_{n})  \neq  \varnothing$
    for every $j\in[n]$.
\end{enumerate}
Then one of the following conditions hold:
\begin{enumerate}
    \item[(ba)] $T_1=\dots=T_n=\TBA$;
    \item[(l)] 
        $T_1=\dots= T_n=\TL$ and 
        for every $k,\ell\in [n]$
        there exists a bridge $\delta$ from 
    $\sigma_{k}$ and $\sigma_{\ell}$ such that $\widetilde \delta=
    \sigma_{k}\circ \proj_{k,\ell}(R)\circ \sigma_{\ell}$;
    \item[(c)] $n=2$ and $T_1=T_2=\TC$;
    \item[(pc)] $n=2$, $T_1=T_2=\TPC$, $A_{1}/\sigma_{1} \cong A_{2}/\sigma_{2}$,
    and the relation $\{(a/\sigma_1,b/\sigma_{2})\mid (a,b)\in R\}$ is bijective. 
\end{enumerate}      
\end{CORMainStableIntersectionCOR}

\begin{proof}
Again, let $f_{i}:\mathbf{R}\to \mathbf A_{i}$ 
be the homomorphism sending each tuple to its $i$-th coordinate.
Denote 
$C_{i}' = f_{i}^{-1}(C_{i})$
and
$B_{i}' = f_{i}^{-1}(B_{i})$.
By Lemma \ref{LEMPropagation}(b) and (bt) 
we have 
$C_{i}'\le{T_{i}(\sigma_{i}')}^{R} B_{i}'\lll R$, 
where $\sigma_{i}' = f_{i}^{-1}(\sigma_i)$.
Since $B_{i}'$ and $C_{i}'$ satisfy conditions 
of Theorem \ref{THMMainStableIntersection} we 
obtain most of the properties and the only nontrivial one is 
the fact that 
$\widetilde \delta =\proj_{k,\ell}(R)$ for any $k$ and $\ell$.
Notice that 
$\widetilde\delta$ for the bridge coming from 
Theorem \ref{THMMainStableIntersection} is equal to 
$\sigma_{k}'\circ\sigma_{\ell}'$.
Translating congruences $\sigma_{k}'$ 
and $\sigma_{\ell}'$ to $\sigma_{k}$ and $\sigma_{\ell}$,
we derive that 
$(a,b) \in \widetilde\delta$ if and only if 
there exists a tuple 
$(a_1,\dots,a_{n})\in R$ such that 
$(a_k,a)\in\sigma_{k}$ and 
$(a_{\ell},b)\in\sigma_{\ell}$.
This implies $\widetilde \delta= \sigma_{k}\circ \proj_{k,\ell}(R)\circ \sigma_{\ell}$. 
The additional condition for $T_{1} =\TPC$ follows 
from the fact that $\sigma_{1}' = \sigma_{2}'$.
\end{proof}

\begin{LEMMultiTypeStillStableLEM}
Suppose $C\le_{\mathcal{M}T}^{A}B$. 
Then $C<_{T}^{A}\dots<_{T}^{A}B$ and $C\lll^{A} B$.
\end{LEMMultiTypeStillStableLEM}
\begin{proof}
Suppose 
$C = C_{1}\cap \dots \cap C_{n}$ 
where $C_{i}<_{T}^{A} B$.
Put 
$D_{j} = \cap_{i=1}^{j} C_{i}$.
By Lemma \ref{LEMIntersectALL}(t) we have 
$D_{j+1}\le_{T}^{A} D_{j}$.
Since $D_{n} = C$ and $D_{1} = C_{1}$ we obtain the required property. 
\end{proof}

\begin{LEMLInearOnTheTopIsEasyLEM}
Suppose $\sigma$ is a linear congruence on $\mathbf A\in\mathcal V_{n}$ 
such that 
$\sigma^{*} = A^{2}$.
Then $\mathbf A/\sigma\cong \mathbf Z_{p}$ for some prime $p$.
\end{LEMLInearOnTheTopIsEasyLEM}

\begin{proof}
Since $\sigma$ is linear, 
applying Lemma \ref{LEMBuildingPerfectCongruence}
a nontrivial bridge from $\sigma$ to $\sigma$ 
we derive that $\sigma$ is a perfect linear congruence.
Hence we have $\zeta\le \mathbf A\times \mathbf A\times \mathbf Z_{p}$ 
with $\proj_{1,2}(\zeta) = A^{2}$.
Choose some element $a\in A$ and put 
$\xi(x,z) = \zeta(x,a,z)$.
Then $\xi$ is a bijective relation
giving an isomorphism 
$\mathbf A/\sigma\cong \mathbf Z_{p}$.
\end{proof}

\begin{LEMPCOnTheTopIsEasyLEM}
Suppose $\sigma$ is a PC congruence on $\mathbf A$ and  
$\sigma^{*} = A^{2}$. 
Then $\mathbf A/\sigma$ is a PC algebra.
\end{LEMPCOnTheTopIsEasyLEM}

\begin{proof}
To show that 
$\mathbf A/\sigma$ is a PC algebra
it is sufficient to show that 
any reflexive $R\le (\mathbf A/\sigma)^{m}$ can be represented as a conjunction of equality relations. 
Assume the converse and 
consider a relation $R\le (\mathbf A/\sigma)^{m}$ of the minimal arity that is not like this. Then projection
of $R$ onto any subset of coordinates gives a full relation.
Choose some tuple 
$(a_1,\dots,a_m)\notin R$.
Notice that $\{a_{i}\}<_{\TPC}^{A/\sigma} A/\sigma$ for every $i$.
Applying Corollary \ref{CORMainStableIntersection} 
we derive that $m$ must be equal to 2 and using the fact that $R$ is reflexive, we derive that $R$ is the equality relation, which contradicts our assumptions. 
\end{proof}





\begin{lem}\label{LEMPreserveLinkdnessOneStepAUX}
Suppose $R\le_{sd} \mathbf A_1\times\mathbf A_2$,
$C_{i}\le_{\TD(\sigma_{i})}^{A_{i}} B_{i}\lll A_1$ for $i=1,2$,
$S$ is the rectangular closure of $R$,
$R\cap (B_1\times C_2)\neq \varnothing$,
$R\cap (C_1\times B_2)\neq \varnothing$,
and 
$S\cap (C_1\times C_2)\neq \varnothing$. Then 
$R\cap (C_{1}\times C_{2})\neq \varnothing$.
\end{lem}

\begin{proof}
Let $\delta_{i}$ be the intersection of all 
irreducible congruences 
coming from $C_{i}\lll A_{i}$ for $i=1,2$.
Let $C_{i}' = C_{i}\circ \delta_{i}$,
$B_{i}' = B_{i}\circ \delta_{i}$,
By Corollary \ref{CORPropagateMultiplyByCongruence}(e)
$C_{i}'<_{\TD(\sigma_{i})}^{A_{i}} B_{i}'\lll A_{i}$ 
for $i=1,2$.

Assume that $R\cap (C_{1}'\times 
C_{2}')\neq \varnothing$.
Notice that $C_{i} = C_{i}'\cap B_{i}$.
If $R\cap (C_{1}\times C_{2}) \neq \varnothing$, then we are done.
Otherwise, 
without loss of generality (we can switch 1 and 2 if it is not true)
there are 
$B_{1}\lll^{A_{1}} F_{1}<_{T(\xi)}^{A_{1}} E_{1}\lll A_{1}$
and
$B_{2}\lll^{A_{2}}  E_{2}\lll A_{2}$
such that 
$R\cap ((C_{1}'\cap F_{1})\times (C_{2}'\cap E_{2})) = \varnothing$
and 
$R\cap ((C_{1}'\cap E_{1})\times (C_{2}'\cap E_{2})) \neq \varnothing$.
Since 
$B_{1}'\cap F_{1}\supseteq B_{1}$, 
we have
$R\cap ((B_{1}'\cap F_{1})\times (C_{2}'\cap E_{2})) \neq \varnothing$.
and by Theorem \ref{THMMainStableIntersection}
$T=\TD$. 
Since $\xi\supseteq \delta_{1}$, 
$C_{1}'\cap F_{1}=C_{1}'\cap E_{1}$, which contradicts our assumptions.

Thus, we assume that 
$R\cap (C_{1}'\times 
C_{2}')= \varnothing$.
Let 
$R' = \delta_1\circ R\circ \delta_2$.
Notice that
$R'\cap (C_{1}'\times C_{2}') = \varnothing$.

Consider two cases:

Case 1.
$(\LeftLinked(R')\cap (B_{1}')^{2})/\delta_1=
(B_{1}')^{2}/\delta_1$.
By Lemma \ref{LEMLeftLinkedStayFull}
$\LeftLinked(R'\cap (B_{1}'\times C_{2}'))/\delta_1
=(B_{1}')^{2}/\delta_1$, which implies that 
$R'\cap (C_{1}'\times C_{2}')\neq\varnothing$.

Case 2. $(\LeftLinked(R')\cap (B_{1}')^{2})/\delta_1\neq
(B_{1}')^{2}/\delta_1$.
As $R'\cap (B_{1}'\times C_{2}')\neq \varnothing$, 
there should be 
$C_{1}''<_{\TD(\sigma_1)}^{A_{1}}B_{1}'$
such that 
$R'\cap (C_{1}''\times C_{2}')\neq\varnothing$.
Notice that $C_{1}''\neq C_{1}'$.
Since $S\cap(C_{1}\times C_{2})\neq\varnothing$, 
$C_{1}'\times C_{1}''\subseteq \LeftLinked(R')$, which contradicts
condition 1 of
Lemma \ref{LEMCongruenceEitherCutOrDoNothing}.
\end{proof}

\begin{lem}\label{LEMPreserveLinkdnessOneStep}
Suppose $R\le_{sd} \mathbf A_1\times\mathbf A_2$,
$C_{1}\le_\mathcal {\TD(\sigma)}^{A_{1}} B_{1}\lll A_1$, $B_{2}\lll A_2$,
$S$ is a rectangular closure of $R$,
$R\cap (B_1\times B_2)\neq \varnothing$,
$S\cap (C_1\times B_2)\neq \varnothing$.
Then 
$R\cap (C_{1}\times B_{2})\neq \varnothing$.
\end{lem}

\begin{proof}
Assume the converse.
Consider $C_{2}'$ and $B_{2}'$ such that 
$B_{2}\lll^{A_2} C_{2}'<_{T}^{A_{2}} B_{2}'\lll A_{2}$, 
$R\cap (C_{1}\times C_{2}') = \varnothing$,
and $R\cap (C_{1}\times B_{2}') \neq \varnothing$.
By Theorem \ref{THMMainStableIntersection}
$T = \TD$. 
Then by Lemma \ref{LEMPreserveLinkdnessOneStepAUX} 
$R \cap ( C_{1}\times C_{2}')\neq\varnothing$, which contradicts our 
assumptions.
\end{proof}
 
\begin{LEMPreserveLinkdnessLEM}
Suppose $R\le_{sd} \mathbf A_1\times\mathbf A_2$,
$C_{i}\le_\mathcal {\TMD}^{A_{i}} B_{i}\lll A_i$
for $i\in\{1,2\}$, 
$S$ is a rectangular closure of $R$,
$R\cap (B_1\times B_2)\neq \varnothing$,
$S\cap (C_1\times C_2)\neq \varnothing$.
Then 
$R\cap (C_{1}\times C_{2})\neq \varnothing$.
\end{LEMPreserveLinkdnessLEM}

\begin{proof}
To prove the lemma it is sufficient to combine 
Lemma \ref{LEMMultiTypeStillStable}
and Lemma \ref{LEMPreserveLinkdnessOneStep}.
\end{proof}

\begin{lem}\label{LEMMultiplyByAllLinear}
$C_{1}<_{\mathcal M T}^{A} B_{1}\lll A$,
$T\in\{\TPC,\TL,\TD\}$,
$B_{2}\lll A$, 
$C_{1}\cap B_{2}=\varnothing$, 
$B_{1}\cap B_{2}\neq\varnothing$.
Then 
$(C_{1}\circ (\omega_{1}\cap \dots \cap \omega_{s})) \cap B_{2}=\varnothing$, 
where 
$\omega_{1},\dots,\omega_s$ are all congruences
of type $T$ 
on $\mathbf A$
such that $\omega_{i}^{*}\supseteq B_{1}^{2}$.
\end{lem}

\begin{proof}
We prove the claim by induction
on the size of $B_{1}$ starting with $B_{1} = A$. 
Thus, the inductive assumption is that the lemma holds for all 
greater $B_{1}$.

Let $C_{1} = C_{1}^{1}\cap \dots\cap C_{1}^{t}$, 
where $C_{1}^{i}<_{\TD(\sigma_{i})}^{A}B_{1}$ for every $i\in[t]$.
Notice that $\sigma_{i}\in\{\omega_1,\dots,\omega_s\}$.
By the definition of $<_{\TD}^{A}$ we have
$C_{1}^{i} = (C_{1}^{i}\circ \sigma_{i})\cap B_{1}$.
By Corollary \ref{CORPropagateMultiplyByCongruence}(e)
and (f) we have
$C_{1}^{i}\circ \sigma_{i}<_{\TD(\sigma_{i})}^{A}B_{1}\circ\sigma_{i}\lll A$.
Applying Theorem \ref{THMMainStableIntersection}
to $
(C_{1}^{1}\circ \sigma_{1})\cap\dots
\cap (C_{1}^{t}\circ \sigma_{t})\cap
B_{1}\cap B_{2}=\varnothing$ we obtain one of the two cases:

Case 1. $(C_{1}^{1}\circ \sigma_{1})\cap\dots
\cap (C_{1}^{t}\circ \sigma_{t})\cap B_{2}=\varnothing$. 
Since 
$$(C_{1}^{1}\circ \sigma_{1})\cap\dots
\cap (C_{1}^{t}\circ \sigma_{t})\supseteq 
(C_{1}^{1}\cap \dots \cap C_1^{t})\circ 
(\sigma_{1}\cap \dots\cap \sigma_{t})\supseteq 
C_{1}\circ 
(\omega_{1}\cap \dots\cap \omega_{s}),$$
we derive that the required intersection is empty and complete this case.

Case 2. There exist 
$B_{1}\lll^{A} B_{1}'<_{\TD}^{A}B_{1}''\lll^{A} A$ 
such that $\bigcap_{i\in[t]}(C_1^{i}\circ \sigma_{i})\cap B_{1}'\cap B_{2}=\varnothing$
and $\bigcap_{i\in[t]}(C_1^{i}\circ \sigma_{i})\cap B_{1}''\cap B_{2}\neq\varnothing$.
By Lemma \ref{LEMIntersectALL}(i)
$\bigcap_{i\in[t]}(C_1^{i}\circ \sigma_{i})\cap B_{2}\lll A$.
Applying the inductive assumption to 
$B_{1}''$ we derive that 
$(B_{1}'\circ \bigcap_{i\in[s]}\omega_{s}) \cap\bigcap_{i\in[t]}(C_1^{i}\circ \sigma_{i})\cap B_{2}=\varnothing$.
Since 
$$(B_{1}'\circ \bigcap_{i\in[s]}\omega_{s}) \cap \bigcap_{i\in[t]}(C_1^{i}\circ \sigma_{i})
\supseteq 
(C_{1}\circ \bigcap_{i\in[s]}\omega_{s}) \cap \bigcap_{i\in[t]}(C_1^{i}\circ \sigma_{i})
\supseteq C_{1}\circ \bigcap_{i\in[s]}\omega_{s}$$
we obtain the required condition.
\end{proof}

\begin{LEMMaximalMultExtentionLEM}
Suppose $C_{1}<_{\mathcal{M}T}^{A}B_1\lll A$, 
$B_{2}\lll A$, 
$C_{1}\cap B_{2} = \varnothing$, 
$B_{1}\cap B_{2}\neq \varnothing$,
$\sigma$ is a maximal congruence on $\mathbf A$ such that 
$(C_{1}\circ \sigma)\cap B_2 = \varnothing$.
Then 
$\sigma = \omega_{1}\cap\dots\cap \omega_{s}$, 
where 
$\omega_{1},\dots,\omega_s$ are 
congruences
of type $T$ 
on $\mathbf A$
such that $\omega_{i}^{*}\supseteq B_{1}^{2}$.
\end{LEMMaximalMultExtentionLEM}

\begin{proof}
By Corollary \ref{CORPropagationModuloCongruence}(m)
consider two cases:

Case 1. There exists $E<_{S}B_{1}/\sigma$.
By Theorem \ref{THMMainStableIntersection}
$E\cap B_{2}/\sigma\neq\varnothing$ and 
$E\cap C_{1}/\sigma\neq\varnothing$, 
hence $C_{1}/\sigma\cap B_{2}/\sigma\neq\varnothing$, which contradicts our assumptions.

Case 2. $C_{1}/\sigma\le_{\mathcal{M}T}B_{1}/\sigma$. 
Since $C_{1}/\sigma\cap B_{2}/\sigma=\varnothing$
and $B_{1}/\sigma\cap B_{2}/\sigma\neq\varnothing$, 
we have $C_{1}/\sigma<_{\mathcal{M}T}B_{1}/\sigma$.
By Lemma \ref{LEMMultiplyByAllLinear}
$(C_{1}/\sigma\circ (\delta_{1}\cap \dots\cap \delta_{r}))
\cap B_{2}/\sigma=\varnothing$,
where 
$\delta_{1},\dots,\delta_r$ are all the congruences
of type $T$ such that $\delta_{i}^{*}\supseteq B_{1}^{2}$.
Extend each congruence $\delta_{i}$ to $\mathbf A$ so that 
$\mathbf A/\omega_{i}\cong (\mathbf A/\sigma)/\delta_{i}$.
Then $\omega_{i}^{*}\supseteq B_{1}^{2}$ and
$(C_{1}\circ(\omega_{1}\cap \dots\cap\omega_{r}))\cap B_{2}=\varnothing$.
Since $\omega_{1}\cap \dots\cap\omega_{r}\supseteq\sigma$ 
and $\sigma$ is a maximal congruence satisfying this condition, 
we obtain $\sigma = \omega_{1}\cap \dots\cap\omega_{r}$, 
which completes the proof.
\end{proof}

\bibliographystyle{plain}
\bibliography{refs}
\end{document}